\documentclass[11pt]{article}

\usepackage[T1]{fontenc}
\usepackage[utf8]{inputenc}
\usepackage[english]{babel}
\usepackage{lmodern}
\usepackage{microtype}
\usepackage[margin=1in]{geometry}

\usepackage{amsmath,amssymb,amsthm}
\usepackage{mathtools}
\usepackage{bbm}

\mathtoolsset{showonlyrefs}

\numberwithin{equation}{section}

\usepackage{graphicx}
\usepackage{float}
\usepackage{wrapfig}
\usepackage{subcaption}
\usepackage{array}
\usepackage{multirow}
\usepackage{tabularx}
\usepackage{booktabs}
\usepackage[thinlines]{easytable}

\usepackage{algorithm}
\usepackage{algpseudocode}

\usepackage{enumitem}
\usepackage[toc,page]{appendix}
\usepackage{refcount}
\usepackage{natbib}

\usepackage[table,dvipsnames]{xcolor}
\usepackage[normalem]{ulem}
\usepackage{soul}

\usepackage[hidelinks]{hyperref}

\theoremstyle{definition}

\newtheorem{theo}{Theorem}[section]
\newtheorem{lemma}[theo]{Lemma}

\newtheorem{cor}[theo]{Corollary}

\newtheorem{rem}[theo]{Remark}
\newtheorem{ass}[theo]{Assumption}

\definecolor{violet}{RGB}{148,0,211}

\newcommand{\R}{\mathbb{R}}
\newcommand{\N}{\mathbb{N}}
\newcommand{\E}{\mathbb{E}}
\newcommand{\PR}{\mathbb{P}}

\newcommand{\Var}{\operatorname{Var}}
\newcommand{\cov}{\operatorname{cov}}

\newcommand{\lb}{\left(}
\newcommand{\rb}{\right)}

\begin{document}

\title{Change Point Detection and Localization in \\ High-Dimensional Time Series}

\author{
Patrick Bastian\thanks{These authors contributed equally to this work.}
\quad
Daria Tieplova\footnotemark[1]
\quad
Nina Dörnemann
\quad
Tim Kutta
\\[0.6em]
\small Department of Mathematics, Aarhus University, Aarhus, Denmark
\\[0.3em]
\small
\texttt{patrick.bastian@math.au.dk},
\texttt{dtieplova@math.au.dk}
\\
\small
\texttt{ndoernemann@math.au.dk},
\texttt{tim.kutta@math.au.dk}
}

\date{}

\maketitle

\begin{abstract}
We present new inference tools for change point detection in high-dimensional time series. We discuss two distinct statistical applications: First, sequential change point testing in an incoming data-stream. Second, retrospective localization of multiple changes, with confidence intervals at a globally controlled error level. Test statistics are built on the maximum norm to generate power against sparse and asynchronous changes. Both problems are tackled by related multiscale statistics that search for changes in the data at many different levels of resolution. For fixed dimension, our statistical approaches can be validated using traditional Hölderian invariance principles. In this paper, we present the high-dimensional analogue: Hölder-Gauss-approximations, which can be (roughly) interpreted as the Gaussian approximation for a Hölder-norm of the high-dimensional partial sum process. Such approximations are of interest beyond change point detection and can be used for other problems such as for stationarity testing in high dimensions.

We evaluate finite-sample performance in a simulation study and give an application to air contamination due to wildfires in California, which occurs asynchronously across a panel of measuring stations.
\end{abstract}

\medskip


\bigskip

\section{Introduction} 

In this paper, we develop new tools for change point detection and localization in the mean parameter of a high-dimensional time series $(X_n)_{n \in \mathbb{N}}$ in $\mathbb{R}^d$. We study two related statistical problems. First, sequential change point testing in an incoming data stream, where the analyst compares a stationary benchmark sample of size $N$ to subsequently arriving observations. Second, retrospective localization of multiple changes from a single offline sample of $N$ data vectors. In both problems, changes may be sparse, in the sense that only few components are affected, and asynchronous, in the sense that different components may change at different times. This is the regime in which maximum-norm based methods are natural.\\
In sequential change point testing, the procedure should asymptotically control the type-I error and, after detecting a change, ideally identify those components where a change has occurred. In retrospective localization, there may be many changes, at many different locations, and the main task is to provide confidence intervals for the change points at a globally controlled error rate. Although these two problems are statistically different, we show that both can be treated by related high-dimensional multiscale statistics.\\
The idea of multiscale approaches becomes easily understandable when juxtaposing it with a MOSUM (Moving Sum) approach. A MOSUM scan statistic for changes is essentially based on comparing two sums of data points, each of length $h$, before and after a change point candidate $k$, to see if their difference is unusually large, i.e.
\[
\hat \gamma^{\mathrm{MOSUM}}(k) 
:= 
\bigg\|
\sum_{i=k-h+1}^k X_i 
-
\sum_{i=k+1}^{k+h}X_i
\bigg\|.
\]
Here $\|x\|= \max_{j=1,\ldots,d}|x_j|$ denotes the maximum norm, and we have suppressed scaling constants in $\hat \gamma^{\mathrm{MOSUM}}$. MOSUM scans are performed by maximizing over candidate values $k$, typically over all $k=h,\ldots,N-h$. Asymptotics are formulated for $N \to \infty$ and $h=h_N$ growing, but negligible compared to $N$.\\
The delicate part of MOSUM, and indeed of any similar kernel-based scanning procedure, is the choice of the ``right'' bandwidth $h$. For illustration, consider multiple change point detection. MOSUM performs best if changes have approximately a distance of order $h$ to each other. If distances are smaller than $h$, then nearby changes may mask each other. If distances are much larger than $h$, then MOSUM leaves power on the table by using an unnecessarily small subsample. For real-valued time series, this is already troublesome, because the relevant change point distances are rarely known beforehand. In high dimensions, with sparse and asynchronous changes, the issue becomes very severe. Different components may have different relevant temporal scales, and choosing a tailored bandwidth for each component becomes quite impractical.\\
Multiscale statistics provide a natural way out. Rather than fixating on a single bandwidth $h$, they search over many, or even all, possible bandwidth choices. Small scales are then available for close and abrupt changes, while large scales are available for weaker and more isolated changes. Besides power, multiscale statistics bring further benefits: narrow confidence intervals in retrospective localization and short detection delays in sequential testing.\\
The statistical gain comes at a mathematical cost. Maximization over many scales is non-trivial to justify and requires a careful weighting scheme. For fixed dimension, the validation of such statistics can often be based on so-called Hölderian invariance principles (see \cite{rackauskas:suquet:2004} for details). These are essentially refinements of Donsker's theorem, stating that the partial sum process
\[
P_{N}(t):= 
\begin{cases}
\frac{1}{\sqrt{N}} \sum_{i=1}^{n} X_{i}, 
& Nt =n \in \mathbb{N}_0,\\[0.2cm]
\textnormal{linearly interpolated}, 
& \textnormal{else},
\end{cases}
\]
after centering, converges to a Brownian motion on a space of Hölder continuous functions. With a continuous mapping argument, the fine Hölder topology then justifies convergence for multiscale statistics that are essentially maxima of weighted increments of $P_N$.\\
In this work, we develop a high-dimensional analogue of this idea, which we call a Hölder-Gauss approximation. Roughly speaking, it gives a Gaussian approximation for weighted maxima of increments of the high-dimensional partial sum process. This is the mathematical backbone of the paper. The main difficulty is that not all scales can be approximated by Gaussian variables: very short sums are not approximately normal. Our argument therefore combines Gaussian approximation on sufficiently large scales with uniform bounds on short scales, where the weighting scheme makes the remaining contribution negligible. Suppressing noise on small scales while still retaining power against local signals is the central trade-off.
 We summarize our main contributions.
\begin{itemize}
    \item[(I)] We develop novel Hölder-Gauss approximations for high-dimensional partial sum processes. These results justify the use of weighted multiscale statistics even when the dimension grows exponentially fast with the sample size.

    \item[(II)] We propose new sequential change point detectors for high-dimensional data streams. The procedures are based on maximum-norm multiscale statistics, allow for sparse and asynchronous changes, and come with asymptotic level guarantees. We also derive support estimators for the changed components and obtain detection-delay bounds.

    \item[(III)] We develop a retrospective segmentation method for high-dimensional time series. The method scans over components, locations and scales, and returns confidence intervals for multiple change points at a globally controlled error level.

    \item[(IV)] We investigate the finite-sample behavior of the proposed procedures in simulations and demonstrate their use on air-quality data from California, where wildfire-related contamination spreads asynchronously across a panel of measuring stations.
\end{itemize}

\noindent The remainder of this paper is organized as follows: We discuss some related literature in the rest of the Introduction. In Section~\ref{sec:2}, we develop the sequential monitoring procedure and prove its main theoretical properties. Section~\ref{sec:theo} contains the core Hölder-Gauss approximation for the online statistic. In Section~\ref{sec:seg}, we turn to retrospective segmentation and construct confidence intervals for multiple changes in high-dimensional time series, with the corresponding Hölder-Gauss approximation given in Section \ref{sec:seg:theory}. Finite-sample properties are studied in Section~\ref{sec:3}, using simulations and wildfire data. Mathematical proofs are gathered in the Appendix.

\subsection{Related literature}  \label{sec:lit}

We review a selected number of related works without trying to be exhaustive.\\
The new Hölder-Gauss approximations proposed in this work can be seen as a high-dimensional analogue of Hölderian invariance principles. Such results were proved originally by \cite{lamperti:1962} for i.i.d. scalar data, with adaptations to dependent data, including under mixing assumptions, by \cite{hamadouche:2000,rackauskas:suquet:2009,giraudo:2017,kutta:dette:wang:2025}. Versions for data in Hilbert and Banach spaces can be found, for example, in \cite{rackauskas:suquet:2009} and \cite{kutta:dornemann:2025}. Traditional papers on Hölderian invariance principles often illustrate their statistical use by testing epidemic alternatives -- essentially a two-change model, where the two changes may be close to or far apart from each other. While specific, this problem is prototypical for broader classes of scan-based testing problems, including the sequential and multiple change point problems discussed below.\\
Our proof of Hölder-Gauss approximations partly builds on high-dimensional Gaussian approximations for maximum norms, in the spirit of \cite{chernozhukov:chetverikov:kato:2013}; more specifically, we draw on results by \cite{chetverikov:wilhelm:kim:2021} when studying our new statistics. These results are used to approximate the contribution of scans over long scales. The short-scale contribution is shown to be asymptotically negligible by concentration inequalities. This negligibility is the counterpart, in our setting, of the vanishing Hölder modulus of continuity on small scales in traditional Hölderian invariance principles; see, for example, the discussion in \cite{rackauskas:suquet:2009}.\\
Statistically, we study two distinct problems: sequential change point detection and retrospective change point localization. Sequential change point testing in this paper refers to procedures which compare incoming data to a training period and then monitor sequentially over an open-ended test period; this formulation goes back to \cite{chu:stinchcombe:white:1996}. An appealing aspect of this setting is that no in-control parameters need to be specified in advance and no parametric assumptions need to be imposed on the error model. It should therefore be distinguished from formulations without a training period, such as \cite{chen:wang:samworth:2024}, which is tailored to high-dimensional Gaussian data streams.\\
Within the framework of \cite{chu:stinchcombe:white:1996}, there are several active directions of research. One line of work studies changes in increasingly complex time series models, with applications in economics; see, for example, \cite{horvath:lazar:liu:wang:xue:2025} and \cite{horvath:trapani:wang:2025}. Another line extends scan statistics, such as the statistic proposed by \cite{fremdt:2015}, to multiscale weighting, which can yield substantially shorter detection delays and often higher power; see \cite{kutta:dornemann:2025} and \cite{bastian:kutta:2025}. Our monitoring statistic is related to this latter line and can be viewed as a strongly weighted, high-dimensional adaptation of \cite{fremdt:2015}.
A further recent trend is the extension of sequential monitoring to more complex data types, including functional data \citep{kutta:kokoszka:2025}, distributional panels \citep{horvath:kokoszka:wang:2021}, and networks \citep{dornemann:kokoszka:kutta:lee:2025}. However, these works on complex data structures do not employ multiscale statistics, despite their benefits in terms of detection delay. The closest work to our monitoring procedure is \cite{gosmann:stoehr:heiny:dette:2022}, who develop monitoring with respect to the max-norm. Our approach differs in two important respects. First, it is multiscale and hence leads to shorter delays, as demonstrated both theoretically and empirically below. Second, their Gaussian approximations are formulated for a finite monitoring horizon and therefore do not by themselves yield open-ended monitoring guarantees in the sense of \cite{chu:stinchcombe:white:1996}. Since we allow open-ended monitoring, our proof requires an additional decomposition of the monitoring statistic into an early part, where the Hölder-Gauss approximation is applied, and a later part controlling the monitoring tail. The latter step relies on a tail bound for the modulus of continuity of the partial sum process on a non-compact domain, roughly comparable in role to Hájek - Rényi inequalities in traditional scalar monitoring; see \cite{aue:kirch:2024}.\\
The second statistical application studied in this paper is retrospective segmentation, or multiple change point localization, in high-dimensional time series. A useful review of general segmentation methods is given by \cite{cho:kirch:2024}. In this area, many works focus on detection and localization, whereas our focus is on simultaneous inference: we construct confidence intervals for the collection of change locations (across all dimensions) with globally controlled error rate.\\
In standard univariate (or multivariate) models, several procedures provide such global guarantees. One example is the classical MOSUM approach as discussed, for example, in \cite{eichinger:kirch:2018}. Multiscale confidence methods have been developed under the label of SMUCE since \cite{frick:munk:sieling:2014}, with extensions to dependent data by \cite{dette:eckle:vetter:2020}. SMUCE may be viewed as an additive multiscale method, because it discounts small scales by subtracting a scale-dependent penalty. A delicate aspect of SMUCE is that its location confidence statements are coupled to the selected number of changes, which can lead to counterintuitive coverage behavior; see the discussion in \cite{fryzlewicz:2024}. The Narrowest Significance Pursuit of \cite{fryzlewicz:2024} avoids prior estimation of the number of changes and delivers simultaneous confidence intervals for change locations. Theoretically, its guarantees are finite-sample, yet it does require quantiles of a multiresolution norm depending on the noise distribution. More recently, Hölder-type scalings have been employed for multiple change point localization by \cite{kutta:dette:wang:2025} and \cite{koehne:mies:2025}. Expanding such Hölder-type statistics to high dimensions is the main objective of this paper. Hölder-type scalings may be viewed as multiplicative multiscale statistics, because they suppress the noise contribution on small scales by a scale-dependent factor. For multivariate data and epidemic change testing, this idea goes back at least to \cite{rackauskas:suquet:2004}; its use for localizing many changes is more recent.\\
We end by pointing to some  high-dimensional segmentation methods that seem most related to our work.
 Earlier methods include sparsified or projected binary segmentation approaches, such as \cite{cho:2016} and \cite{wang:samworth:2018}, which focus on detecting and localizing changes in high-dimensional mean structures. More recent work has begun to address inference. For high-dimensional mean-shift models, \cite{kaul:michailidis:2025} develop refined local least-squares estimators and derive asymptotically valid confidence intervals for multiple change point locations under sparsity assumptions. Related inference results for high-dimensional regression models under temporal dependence are obtained by \cite{xu:wang:zhao:yu:2024}. Other recent procedures, such as \cite{cho:owens:2024}, provide computationally efficient segmentation methods for high-dimensional regression settings under dependence and non-Gaussianity, with consistency guarantees for the number and locations of changes. These works show that high-dimensional change point inference is an active area of research. Their inferential constructions, however, are typically based on refined estimators obtained after an initial localization step. This differs from the approach pursued here, where confidence intervals arise directly from a global multiscale statistic and are simultaneous over all detected changes. In this sense, our segmentation results complement the existing high-dimensional segmentation literature by providing a SMUCE/NSP-type inferential object in a high-dimensional setting.

\section{Statistical methodology: Sequential detection} \label{sec:2}

In this section, we develop a sequential test for detecting changes in the mean parameter of a high-dimensional time series. Our approach relies on a novel Hölder-Gauss approximation, which allows us to approximate our multiscale detector by the supremum of a Gaussian process over different scales and over long time horizons. We also propose a set estimator for those components where a change has occurred, and prove that the probability of falsely including an unchanged component is asymptotically controlled at a prescribed level. 

\subsection{Model and method}\label{sec:21}

\noindent  Let $(X_n)_{n \ge 1}$ denote a sequence of random vectors in $\mathbb{R}^d$, where we use the notation $X_n = (X_{n,1}, \ldots, X_{n,d})^\top \in \R^d$. The aim of the analyst is detecting changes in the mean parameter of the $X_n$ and localizing the components where said changes have occurred. Our procedure does not require any in-control parameters. Rather, changes are compared to an initial training period, which serves as a benchmark for the subsequent inference. To be precise, for a training size $N \in \mathbb{N}$, the analyst starts with a sample $
\{X_1,\dots,X_N\}$ of random vectors which is assumed to be free of changes. Then, after time $N$, the monitoring period starts, where data arrive sequentially; first $X_{N+1}$, then $X_{N+2}$, and so forth, theoretically indefinitely. 
At every time $k$ in the monitoring period, the analyst scans for changes across all components. We formalize this process with a data-generation scheme and a testing formulation.\\
For each component $j \in \{1,\ldots,d\}$ with a change, let $N + k_j^\star \in \mathbb{N}$ denote the time where the mean of the data changes from a pre-change value $\mu_j^{(1)} \in \mathbb{R}$ to a post-change value $\mu_j^{(2)} \in \mathbb{R}$. For components without a change, one may formally set $k_j^\star=\infty$. With this notation in hand, we express our data as follows:
\begin{align} \label{e:amoc}
X_{i,j} = \mu_j^{(1)} \, \mathbb{I}\{i \leq N + k_j^\star\} 
+ \mu_j^{(2)} \, \mathbb{I}\{i > N + k_j^\star\} 
+ \varepsilon_{i,j}, 
\qquad 1 \leq j \leq d, \; i \geq 1.  
\end{align}
Here, $(\varepsilon_{i,j})$ denotes a centered array of noise variables. 
It is often helpful to refer to all components $j$ that undergo a change, and we thus define the set
\begin{align} \label{eq_def_C}
\mathcal{C}(\infty) := \{ j \in \{1,\dots,d\} : \mu_j^{(1)} \neq \mu_j^{(2)}, k_j^\star<\infty \}.
\end{align}
If $\mathcal{C}(\infty) = \emptyset$, there is no change in any component. The global null hypothesis is that no component ever changes during the monitoring period. The sequential procedure raises an alarm once the observed data provide evidence against this null. We thus consider the hypothesis pair
\begin{align} \label{e:hypothesis}
&H_0: \mathcal{C}(\infty) = \emptyset, \qquad \textnormal{vs.} \qquad
H_1: \mathcal{C}(\infty) \neq \emptyset. 
\end{align}
To test this hypothesis pair, we suggest the use of a sequential CUSUM statistic based on the maximum norm $\| \cdot \|$. It can be calculated at any time $k \ge 1$ in the monitoring period, using the available data $X_1, X_2,\ldots,X_{N+k}$, and is defined as follows:
\begin{align} \label{e:statmain}
\widehat{\Gamma}_N^{\beta}(k):=  &\max_{0 \leq \ell < k} \bigg\{w_{\beta}(\ell, k, N) \cdot \Big\|\frac{k-\ell}{N} \sum_{n=1}^{N} X_{n}-\sum_{n=N+\ell+1}^{N+k}X_{n} \Big\| \bigg\}\,.  
\end{align}
Essentially, for fixed $(\ell,k)$, the statistic compares the mean during the training period to the mean in a recent monitoring block $N+\ell+1,\ldots,N+k$. Notice that at time $k$, scans are performed over all possible $\ell \in \{0,\ldots,k-1\}$, i.e. over the entire past. Scans are combined using a weighted maximum, with the weight function 
 $w_{\beta}$ described below. The above data aggregation scheme was first developed by \cite{fremdt:2015} for regression models and has, to the best of our knowledge, never been used in the context of high-dimensional data. For aggregation, \cite{fremdt:2015} used the weight function from \cite{horvath:kokoszka:huskova:steinebach:2003}, which emphasizes early changes. Since we want to develop an approach that performs well for early and late changes, and inspired by \cite{kutta:dornemann:2025},
we equip the CUSUM with a more powerful Hölder-type weight function 
$w_{\beta}$,
defined for $\beta \in [0,1/2)$ as
\begin{align} \label{eq_def_omega}
      w_{\beta}(\ell, k, N) =  \bigg\{(k-\ell)^{\beta} \big( N+k\big)^{1-\beta} N^{-1/2}\bigg\}^{-1}, \qquad 0 \le \ell<k<\infty.
\end{align}
The weight parameter $\beta \in [0,1/2)$ is user-determined. Essentially, larger values of $\beta$ lead to a higher weighting of small scales $k-\ell$, which in turn leads to faster detection of changes. As we will demonstrate in the theory section below, the price of larger $\beta$ is that Gaussian approximations for the detector $\widehat{\Gamma}_N^{\beta}(k)$ only hold for lower dimensions, though permissible growth rates are still stretched-exponential in $N$ for any $\beta \in [0,1/2)$. The term ``Hölder weights'' for \eqref{eq_def_omega} is derived from the fact that, at least heuristically, $\widehat{\Gamma}_N^{\beta}(k)$ corresponds to a Hölder norm for the partial sum process in $\mathbb{R}^d$. 

We now suggest the following statistical procedure: At any time $k$ in the monitoring period, calculate $\widehat{\Gamma}_N^{\beta}(k)$ and check whether it exceeds a threshold value $q$. If $\widehat{\Gamma}_N^{\beta}(k)\le q$, monitoring continues. If $\widehat{\Gamma}_N^{\beta}(k)> q$, a change is declared and the null hypothesis $H_0$ in \eqref{e:hypothesis} is rejected. An adequate threshold value $q$ depends on the targeted nominal level of the testing problem and the distribution of $\sup_{k \ge 1}\widehat{\Gamma}_N^{\beta}(k)$ under $H_0$. Deriving this distribution is the main theoretical challenge of this work; see Theorem \ref{thm:main}. For practical computation, we propose in Algorithm \ref{alg:bootstrap_quantile} a parametric bootstrap procedure. The use of Algorithm \ref{alg:bootstrap_quantile} ensures an asymptotic approximation for the user-determined level $\alpha \in (0,1)$.

\begin{algorithm}[H]
\caption{Bootstrap estimation of $\hat q^N_{1-\alpha}$}
\label{alg:bootstrap_quantile}
\begin{algorithmic}[1]
\Require Simulation horizon $T$, number of bootstrap samples $B$, level $\alpha$

\Ensure Estimated quantile $\hat q^N_{1-\alpha}$

\State Compute the centered sample covariance matrix $\hat\Sigma$ from $X_1, \ldots, X_{N}$
\For{$b = 1,\dots,B$}
    \State Generate independent samples $Z_1^{(b), \star},\dots,Z_{N+T}^{(b), \star} \sim \mathcal N(0,\hat\Sigma)$
    \For{$k = 1,\dots,T$}
        \State Compute $\hat\Gamma_N^{\beta, \star}(k)$ according to \eqref{e:statmain} using $\{Z_t^{(b),\star}\}_{t=1}^{N+T}$ as data
    \EndFor
    \State Compute
    \[
        \hat \Gamma_N^{\beta,(b),\star} \;=\; \max_{1\leq k \leq T} \hat \Gamma_N^{\beta,\star}(k)
    \]
\EndFor
\State Compute $\hat q^N_{1-\alpha}$ as the $\lceil B(1-\alpha)\rceil$-th order statistic of 
$\{\hat \Gamma_N^{\beta,(b),\star}\}_{b=1}^B$
\end{algorithmic}
\end{algorithm}
\noindent
In computation, the open-ended supremum is approximated by a supremum over a large finite simulation horizon $T$. In theory, the open-ended quantile is consistently approximated by choosing $T=T_N$ such that $T_N/N$ grows at a  slow polynomial rate; this follows from the proof of Theorem \ref{thm:main}.\\
Traditionally, sequential change point tests are formulated for a univariate time series and stop after detecting a change, i.e. after rejecting $H_0$. However, in a multivariate setup, it can be natural, after identifying a change in component $j_1$, to continue monitoring the remaining components $\{1,...,d\}\setminus \{j_1\}$ and determine whether further changes occur. Indeed, such a setup may have the interpretation of a change spreading across components. For this reason, we define in the next step a set estimator for $ \mathcal{C}(\infty)$. The estimator is defined sequentially at time $k$ in the monitoring period as follows:
\begin{align}\label{e:setest}
 \widehat{\mathcal{C}}_N(k) := &\big\{j : \max_{1 \le r \le k}\widehat{\Gamma}_N^{\beta}(j,r)>q \big\}, \qquad \textnormal{where}\\
\widehat{\Gamma}_N^{\beta}(j,k):= & \max_{0 \leq \ell < k} \bigg\{w_{\beta}(\ell, k, N) \cdot \Big|\frac{k-\ell}{N} \sum_{n=1}^{N} X_{n,j}-\sum_{n=N+\ell+1}^{N+k}X_{n,j} \Big| \bigg\}.   \label{e:comp}
\end{align}
The set estimator $ \widehat{\mathcal{C}}_N(k)$ can be shown to satisfy asymptotic false detection guarantees analogous to those of a sequential hypothesis test; details are given in the next section. We already point out that $ \widehat{\mathcal{C}}_N(k) $ is by construction increasing in $k$, i.e.
\[
 \widehat{\mathcal{C}}_N(k) \subset  \widehat{\mathcal{C}}_N(k+1).
\]
This property is desirable, because it means that a change declared in some component $j$ at time $k$ will not be revoked at a later time $k'$.

\subsection{Theoretical analysis} \label{sec:theo}

We now proceed to the mathematical analysis of the methodology described in Section \ref{sec:21}. 
Under the model of that section, we analyze the behavior of the statistic \eqref{e:statmain} in an asymptotic framework where $N \to \infty$, corresponding to an increasing size of the training sample. The dimension $d=d^{(N)}$ is a non-decreasing sequence of natural numbers and may grow at a stretched-exponential rate in $N$, with the growth rate specified in Assumption \ref{ass:main} below. For the study of alternatives, all model objects may depend on $N$, including the dimension $d$, the noise distribution, the mean vectors, the change locations, and the set $\mathcal C(\infty) = \mathcal C^{(N)}(\infty)$ defined in \eqref{eq_def_C}. To avoid changing alternatives along subsequences, we assume formally that
\[
\mathcal C^{(N)}(\infty)\subset \mathcal C^{(N+1)}(\infty), \qquad N\ge 1.
\]
This assumption has no practical meaning for a fixed data set; it only ensures that, asymptotically, either the null holds for all $N$ or the alternative holds eventually. We suppress the index $N$ whenever no confusion can arise.
To state our assumptions, we need the definition of the $\Psi_2$-Orlicz norm of a real-valued random variable $Y$, which we recall as
\begin{align*}
    \|Y\|_{\Psi_2} 
        := \inf\Bigl\{ t > 0 : \; 
\mathbb{E}\bigl[\exp\!\bigl((Y/t)^2\bigr)\bigr] \leq 2 \Bigr\}.
\end{align*}
Note that $\|Y\|_{\Psi_2}$ is also called the subgaussian norm, as subgaussian random variables can be characterized as having a finite $\Psi_2$-Orlicz norm. 
\newpage 
\begin{ass} \label{ass:main}
    \begin{enumerate}[label=(A-\arabic*)] $ $
           \item \label{ass1} There exists a positive constant 
        \[
       0 <C_1< \frac{1-2\beta}{12-18 \beta}\qquad \textnormal{such that}\qquad   0 \le \limsup_N \frac{d}{\exp(N^{C_1})} <\infty.
        \]
        \item \label{ass2} For any fixed $N$, the sequence of model errors $(\varepsilon_i)_{i \in \mathbb{N}}$ consists of i.i.d. centered random vectors in $\mathbb{R}^d$.
        \item \label{ass3}
        For some $C_2>0$, it holds for $\varepsilon_1 = (\varepsilon_{1,1}, \ldots, \varepsilon_{1,d})^\top$ that 
        \[
        \sup_N \sup_{j=1,\ldots,d} \|\varepsilon_{1,j}\|_{\Psi_2} 
        \le C_2.
        \]
        \item \label{ass4}
       It holds for some constants $0< \underline{\sigma} \le \bar \sigma<\infty$ that
        \[
        \underline{\sigma}^2 \le \inf_N \min_{j} \mathbb{E}\varepsilon_{1,j}^2 \le \sup_N \max_{j} \mathbb{E}\varepsilon_{1,j}^2 \le \bar \sigma^2.
        \]
    \end{enumerate}
\end{ass}

\noindent The main step of our analysis is to prove the following Hölder-Gauss approximation, which states that, under $H_0$, the distribution of the detector $\sup_{k \ge 1} \hat\Gamma_N^{\beta}(k)$ can be approximated by the supremum of a high-dimensional Gaussian process. The name of this result is derived from the fact that it generalizes Hölderian invariance principles that are well explored in finite dimensions; see the literature review in Section \ref{sec:lit}. To this end, let $Z_1,Z_2,\ldots$ denote i.i.d. centered Gaussian vectors in $\mathbb{R}^d$ which have the same covariance matrix as $\varepsilon_1$. Also define the Gaussian detector
\begin{align} \label{e:fullcwZ}
\widetilde{\Gamma}_N^{\beta}(k):=  &\max_{0 \leq \ell < k} \bigg\{w_{\beta}(\ell, k, N) \cdot \Big\|\frac{k-\ell}{N} \sum_{n=1}^{N} Z_{n}-\sum_{n=N+\ell+1}^{N+k}Z_{n} \Big\| \bigg\}\, .  
\end{align} 

\begin{theo}[Hölder-Gauss Approximation]
\label{thm:main}
    Suppose that $H_0$ holds true and that Assumption \ref{ass:main} is satisfied. Then, for any $x_0>0$, it holds that 
\[
\lim_{N\to\infty} \sup_{x>x_0}\big| \mathbb{P}\big(\sup_{k\ge 1}\widehat{\Gamma}_N^{\beta}(k) \le x\big)-\mathbb{P}\big(\sup_{k\ge 1}\widetilde{\Gamma}_N^{\beta}(k)\le x\big)\big| =0.
\]
\end{theo}

Theorem \ref{thm:main} provides theoretical justification for the level approximation of the statistical procedures developed in Section \ref{sec:21}. We state this precisely in the next corollary. In the second part, we also derive a false positive guarantee for the set estimator $\widehat{\mathcal{C}}_N$, controlling the risk of erroneously including a component into our support estimate where no change has occurred. To this end, denote by $q^{(N)}_{1-\alpha}$ the  $(1-\alpha)$-quantile of $\sup_{k\ge 1}\widetilde{\Gamma}_N^{\beta}(k)$.

\begin{cor} \label{cor:main:level}
    Suppose that Assumption \ref{ass:main} holds and let $\alpha \in (0,1/2)$ be fixed. 
    Then, supposing $H_0$, we have
    \[
    \lim_{N \to \infty} \mathbb{P}_{H_0}\big(\sup_{k\ge 1}\widehat{\Gamma}_N^{\beta}(k) >q^{(N)}_{1-\alpha}\big) = \alpha.
    \]
    Moreover, with $q=q^{(N)}_{1-\alpha}$ in \eqref{e:setest},
    \[
    \limsup_{N \to \infty} \mathbb{P}\Big(\exists j \in \bigcup_{k\ge 1}\widehat{\mathcal{C}}_N(k) \,\,\, \textnormal{with} \,\,\, j \not \in \mathcal{C}(\infty)\Big) \le \alpha.
    \]
\end{cor}
The level guarantees in Corollary \ref{cor:main:level} are formulated in terms of the true quantile $q^{(N)}_{1-\alpha}$, which still depends on the unknown covariance matrix $\Sigma$ of the error distribution. However, one may use the empirical covariance matrix $\hat \Sigma$ of the training sample $X_1, \ldots X_N$ to simulate quantiles; see Algorithm \ref{alg:bootstrap_quantile}. This replacement is justified by Lemma \ref{lem:compar_gaus} in the Appendix, together with the fact that under our model assumptions it is not hard to show that $\max_{1 \leq i,j \leq d}|\Sigma_{ij}-\hat \Sigma_{ij}| = o_P(1)$.\\
Next, we characterize the behavior of the sequential test and set estimator in the presence of changes. For this purpose, we show two separate results. First, we analyze the deterministic drift of the test statistic in the presence of changes. Since we consider a max-norm statistic, we can express this most naturally using the componentwise statistic \eqref{e:comp}. Second, we analyze the size of the random part, i.e. the order of the test statistic under $H_0$. This follows by using an exponential moment result for suprema of Gaussian processes. Together we obtain:

\begin{theo} \label{thm:alt}
    Suppose that Assumption \ref{ass:main} holds. Moreover, suppose that there exists a change in component $j$ at location $N+k_j^{\star}$ such that $|\mu_j^{(1)}-\mu_j^{(2)}|=:\Delta_j>0$. Then, it holds for $k>k_j^{\star}$ that
    \begin{align} \label{e:dec:main}
        \widehat{\Gamma}_N^{\beta}(j,k) \ge
        \bigg(\frac{k-k^\star_j}{N+k}\bigg)^{1-\beta} \big(\Delta_j \sqrt{N}\,\big)
        - \mathcal{O}_P\Big(\sup_{k\ge 1}\widetilde{\Gamma}_N^{\beta}(k)\Big).
    \end{align}
    The constants implied by the $\mathcal{O}_P$-term are independent of $j$ and a further bound for it is implied by the fact that
    \begin{align} \label{e:bound:WC}
    \sup_{k\ge 1}\widetilde{\Gamma}_N^{\beta}(k) = \mathcal{O}_P\big( \log^{1/2}(ed)\big).
    \end{align}
\end{theo}

The bound in \eqref{e:bound:WC} is a worst-case upper bound, which is sharp, for example, in the case of i.i.d. noise components. In other cases, such as for discretized functional time series, it can be true that the growth is much slower, or even that the term is $\mathcal{O}_P(1)$. The same bound obviously holds when the dimension $d$ is fixed.
For our analysis, Theorem \ref{thm:alt} has two distinct and related uses: It can be used to analyze detection delays and consistency of the test statistic. We begin with consistency, both for the sequential test and for the set estimator.

\begin{cor} \label{cor:alt1}
Suppose that Assumption \ref{ass:main} holds. 
Then, if there exists a $j$ which may depend on $N$ such that
\[
\lim_{N \to \infty} \frac{\Delta_j \sqrt{N}}{\sqrt{\log(ed)}} = \infty, \quad \textnormal{then} \quad 
\lim_{N \to \infty} \mathbb{P}\big(\sup_{k\ge 1}\widehat{\Gamma}_N^{\beta}(k) >q_{1-\alpha}^{(N)}\big)=1.
\]
Similarly, if 
\[
\lim_{N \to \infty} \min_{j \in \mathcal{C}(\infty)}\frac{\Delta_j \sqrt{N}}{\sqrt{\log(ed)}} = \infty, \quad \textnormal{then} \quad 
\lim_{N \to \infty} \mathbb{P}\bigg(\mathcal{C}(\infty)\subset \bigcup_{k\ge 1}\widehat{\mathcal{C}}_N(k)\bigg)=1.
\]
\end{cor}

Finally, we come to the topic of detection delays and we focus again on individual components to make the results more precise. Suppose there exists a change in some component $j$ which may depend on $N$. The detection delay in component $j$ is the time elapsing between the occurrence of the change $k^\star_j$ and the first time that $\widehat{\Gamma}_N^{\beta}(j,k) >q_{1-\alpha}^{(N)}$. Denoting this first crossing time by $\hat k_j$, we define the detection delay as
\[
\tau_j := \max(0, \hat k_j-k_j^\star).
\]
Taking the positive part prevents us from having negative detection delays that result from erroneous and premature rejections. 

\begin{cor} \label{cor:alt2}
Suppose that Assumption \ref{ass:main} holds and that there exists a change in component $j$, with $k_j^\star  =  \lfloor N \theta \rfloor $ for some $\theta>0$. Then
\[
\textnormal{if} \quad \lim_{N \to \infty} \frac{\Delta_j \sqrt{N}}{\sqrt{\log(ed)}} = \infty, \quad\textnormal{it follows that} \quad
\tau_j = \mathcal{O}_P \bigg( \bigg(\frac{\log(ed) N^{1-2\beta}}{\Delta_j^2 }\bigg)^{1/[2(1-\beta)]}\bigg).
\]
\end{cor}

If $\Delta_j$ is bounded away from zero, this becomes
\[
\tau_j
=
\mathcal{O}_P\left[
(\log (ed))^{1/[2(1-\beta)]}
N^{(1-2\beta)/[2(1-\beta)]}
\right].
\]
This corollary demonstrates that delay times are generally shorter for larger $\beta \in [0,1/2)$, i.e. when putting more weight on the recent past. Recall that there exists an interaction between $\beta$ and the dimension $d$ in our theory, where larger $\beta$ requires smaller $d$. In the case where $d$ grows polynomially in $N$, the setup of \cite{gosmann:stoehr:heiny:dette:2022}, any value for $\beta$ is permitted and we can ensure very fast detection. Indeed, in that case, we obtain delays of size
\[
\mathcal{O}_P\left[
(\log N)^{1/[2(1-\beta)]}N^{\zeta}
\right],
\qquad
\zeta =\frac{1-2\beta}{2(1-\beta)},
\]
where $\zeta$ can be made arbitrarily small by choosing $\beta$ close to $1/2$. In contrast to our above results, analyzing the method by \cite{gosmann:stoehr:heiny:dette:2022} shows that their delays are of size $\mathcal{O}_P(\log^{1/2}(N)N^{1/2})$, which is much longer. Another advantage over \cite{gosmann:stoehr:heiny:dette:2022} of our methodology is that we justify it for open-ended monitoring, which is the gold-standard in the framework by \cite{chu:stinchcombe:white:1996}; see the recent review by \cite{aue:kirch:2024}.

\section{Statistical methodology: Segmentation} \label{sec:seg}

In this section, we develop a retrospective segmentation method for detecting and localizing component-specific mean changes in a high-dimensional time series. As in the previous section, theory under the null distribution relies on a Hölder-Gauss approximation for a multiscale scan statistic. The output of the method is a collection of (component-indexed) confidence intervals for change point locations.

\subsection{Model and method}\label{sec:seg:model}

Consider a sample of random vectors \( X_1, \dots, X_N \in \mathbb{R}^d \), where as before the $i$th component of $X_n$ is denoted by $X_{n,i}$. We decompose
\begin{align} \label{e:seg:model}
X_{n,i} = \mu_{n,i} + \varepsilon_{n,i}, \qquad n=1,\ldots,N,\quad i=1,\ldots,d,
\end{align}
where \( \mu_n \in \mathbb{R}^d \) is a deterministic mean vector and \( \varepsilon_n \in \mathbb{R}^d \) is a centered noise vector. The mean sequence is assumed to be piecewise constant in each component. More precisely, for every component $i$, let
\begin{align} \label{e:seg:cp}
c_{0,i} := 0 < c_{1,i} < \dots < c_{K_i,i} < N = c_{K_i+1,i}
\end{align}
denote the component-specific change point locations. We write
\begin{align} \label{e:seg:mean}
\mu_{n,i} = \mu_i^{(k)}, \qquad c_{k-1,i} < n \le c_{k,i},\qquad k=1,\ldots,K_i+1.
\end{align}
Thus the $k$th change in component $i$ occurs at location $c_{k,i}$ and has jump size
\begin{align} \label{e:seg:jump}
\Delta_{k,i}:=
|\mu_i^{(k+1)}-\mu_i^{(k)}|,
\qquad k=1,\ldots,K_i.
\end{align}
We also define the local separation of the $k$th change by
\begin{align} \label{e:seg:sep}
\delta_{k,i}
:=
\min\big(c_{k+1,i}-c_{k,i}, c_{k,i}-c_{k-1,i}\big).
\end{align}
The number of changes $K_i$ is unknown, and our goal is to construct confidence intervals for the associated component-specific change locations \(c_{k,i}\).
For detecting changes in component $i$, we use the symmetric scan statistic
\begin{align}
\label{eq:def:segmentation:scan}
    \gamma_i(n,h)
    =
    \frac{\sum_{m=n-h+1}^{n}X_{m,i}-\sum_{m=n+1}^{n+h}X_{m,i}}{N^{1/2-\beta}h^{\beta}},
\end{align}
where, as before, \(\beta\in[0,1/2)\) is a weight parameter and the pair $(n,h)$ is chosen such that
$1\le n-h+1<n+h\le N$.
For fixed $(n,h)$, the statistic compares the block immediately to the left of $n$ with the block immediately to the right of $n$. Large values of $|\gamma_i(n,h)|$ provide evidence for a change inside the discrete interval
\[
I(n,h):=\{n-h+1,\ldots,n+h-1\}.
\]
Our statistical method is multiscale, because it scans over many choices of $h$. Small values of $h$ are useful for localizing large and closely spaced changes, whereas large values of $h$ are needed for smaller but well-separated changes. The precise signal-to-noise condition under which all changes are detected is stated in Section~\ref{sec:seg:theory}. The use of the scan statistic \eqref{eq:def:segmentation:scan} with Hölder weights is inspired by \cite{kutta:dette:wang:2025}, where such statistics were considered for Banach space-valued data, including functional data, based on corresponding Hölderian invariance principles. That framework, however, does not cover the high-dimensional setting considered here.
Now, let
\begin{align}  \label{eq_def_mathcal_S}
\mathcal{S}\subset \{(n,h): 1 \le n-h+1<n+h\le N\}
\end{align}
be the set of indices over which scans are performed and, for later reference, define its $h$-slice by
\[
\mathcal{S}_h:= \{n: (n,h) \in \mathcal{S}\}.
\]
Different choices for $\mathcal{S}$ are possible. The exhaustive choice is
\begin{align} \label{e:seg:Sstar}
    \mathcal{S}^\star
    :=
    \{(n,h): 1 \le n-h+1<n+h\le N\},
\end{align}
which is statistically natural but computationally expensive. A cheaper alternative is given by a geometrically thinned grid,
\begin{align} \label{e:seg:Stheta}
    \mathcal{S}^\theta
    :=
    \{(n,h) \in \mathcal{S}^\star: \exists m \in \mathbb{N}_0
    \textnormal{ such that } h = \lfloor \theta^m\rfloor \},
    \qquad \theta>1.
\end{align}
Using $\mathcal{S}^\theta$ reduces the number of scanned pairs from order $\mathcal{O}(N^2)$ to order $\mathcal{O}(N\log N)$, an important improvement since scans need to be performed across all $d$ components. For notational simplicity in the proofs we will assume that $(\lfloor N/2\rfloor,\lfloor N/2\rfloor-1) \in \mathcal S$. Our basic scan procedure along one fixed component is described in Algorithm \ref{alg:seg:scan}. Subsequently, componentwise scans are combined to COMBI-SCAN in Algorithm \ref{alg:seg:combi}. The outputs are confidence intervals for component-specific changes.  To obtain confidence intervals that are as small as possible, Algorithm \ref{alg:seg:scan} begins by scanning small scales $h$, for which the interval length is \(2h-1\), and then proceeds to larger scales. Changes are detected, and confidence intervals created, whenever the scan statistic $|\gamma(n,h)|$ exceeds an input threshold $q$, analogous to the threshold from the sequential monitoring section. Once an interval is reported, future accepted intervals are required to be disjoint from it, which avoids double detection of the same change. 

\begin{algorithm}[h]
\caption{SCAN: Scan across a single scalar time series} \label{alg:seg:scan}
\begin{algorithmic}[1]
    \State \textbf{Input:} data $x_1,\ldots,x_N \in \mathbb{R}$, index set $\mathcal{S}$, threshold $q>0$
    \State \textbf{Output:} $\hat{\mathcal{I}}$, a collection of discrete confidence intervals
    \State \textbf{Initialize} $I_{\mathrm{all}}=\emptyset$ and $\hat{\mathcal{I}}=\emptyset$
    \For{$h=1,\ldots,\lfloor N/2\rfloor $}
        \If{$\mathcal{S}_h \neq \emptyset$}
            \For{$n \in \mathcal{S}_h$ in increasing order}
                \State Compute $\gamma(n,h)$ and set
$I:=\{n-h+1,\ldots,n+h-1\}$
\If{$|\gamma(n,h)|>q$ and $I\cap I_{\mathrm{all}}=\emptyset$}
    \State Update
    $I_{\mathrm{all}}\to I_{\mathrm{all}}\cup I$
    and
    $\widehat{\mathcal I}\to\widehat{\mathcal I}\cup\{I\}$
\EndIf
            \EndFor
        \EndIf
    \EndFor 
    \State \Return $\hat{\mathcal{I}}$
\end{algorithmic}
\end{algorithm}

\begin{algorithm}[h]
\caption{COMBI-SCAN: Scan across a multivariate time series} \label{alg:seg:combi}
\begin{algorithmic}[1]
    \State \textbf{Input:} data $X_1,\ldots,X_N \in \mathbb{R}^d$, index set $\mathcal{S}$, threshold $q>0$
    \State \textbf{Output:} $\hat{\mathcal{I}}$, a collection of component-indexed confidence intervals
    \State \textbf{Initialize} $\hat{\mathcal{I}}=\emptyset$
    \For{$i=1,\ldots,d$}
        \State Compute $ \hat{\mathcal{I}}_i := \textnormal{SCAN}((X_{1,i},...,X_{N,i}), \mathcal{S},q) $
        \State Update  $\hat{\mathcal{I}} \to \hat{\mathcal{I}} \cup \{(I,i): I\in \hat{\mathcal{I}}_i\}$
    \EndFor 
    \State \Return $\hat{\mathcal{I}}$
\end{algorithmic}
\end{algorithm}

\begin{algorithm}[h]
\caption{Bootstrap estimation of $\hat q^N_{1-\alpha}$}
\label{alg:seg_bootstrap_quantile}
\begin{algorithmic}[1]
\Require Index set $\mathcal S$, number of bootstrap samples $B$, level $\alpha$

\Ensure Estimated quantile $\hat q^N_{1-\alpha}$

\State Compute first differences $D_n=X_{n+1}-X_n$, $n=1,\ldots,N-1$, and set
\[
    \hat\Sigma_D
    =
    \frac{1}{2(N-1)}
    \sum_{n=1}^{N-1}D_nD_n^\top .
\]

\For{$b = 1,\dots,B$}
    \State Generate independent samples
    $Z_1^{(b),\star},\ldots,Z_N^{(b),\star}\sim \mathcal N(0,\hat\Sigma_D)$.
    \State Compute $T_N^{(b),\star}
        =
        \max_{(n,h)\in\mathcal S}
        \max_{1\le i\le d}
        |\gamma_i^{(b),\star}(n,h)|,$
    
    where $\gamma_i^{(b),\star}(n,h)$ is defined as in \eqref{eq:def:segmentation:scan}, with $X$ replaced by $Z^{(b),\star}$.
\EndFor

\State Compute $\hat q^N_{1-\alpha}$ as the $\lceil B(1-\alpha)\rceil$-th order statistic of
$\{T_N^{(1),\star},\ldots,T_N^{(B),\star}\}$.
\end{algorithmic}
\end{algorithm}

If the threshold $q$ in COMBI-SCAN is properly calibrated, each reported interval contains at least one true change point (in the corresponding component), with asymptotic probability \(\ge 1-\alpha\). 
For calibration, we suggest again using a parametric bootstrap based on Gaussian data.  A version for segmentation is given in Algorithm \ref{alg:seg_bootstrap_quantile}. Notice that, for sequential monitoring, estimating the covariance of the noise was relatively simple because it could be based on the uncontaminated training sample. Since no such sample is available in the segmentation problem, we estimate the covariance matrix using first-order differences of the data, to minimize the impact of potential changes. This estimator is appropriate when the contribution of mean changes to the empirical covariance is negligible, i.e. if mathematically speaking
\[
\left\|
\frac{1}{N}
\sum_{n=1}^{N-1}
(\mu_{n+1}-\mu_n)(\mu_{n+1}-\mu_n)^\top
\right\|_{\max}
=o(1).
\]
This condition is of course satisfied if there are not too many changes that are not too large. The theoretical foundation of the COMBI-SCAN with bootstrap calibration relies once more on a Hölder-Gauss approximation, which is given in the next section.

\subsection{Theoretical analysis} \label{sec:seg:theory}

We now analyze the theoretical properties of the segmentation procedure. The main task is once more the derivation of a Hölder-Gauss approximation to justify the quantile calibration in Algorithm \ref{alg:seg_bootstrap_quantile}. For our theory in this section, we will always confine ourselves to the case where the scan goes over all indices \( \mathcal{S}^\star \) in \eqref{e:seg:Sstar}. Still, similar results can be formulated for sparse versions such as \(\mathcal{S}^\theta\).
To describe the approximating null distribution, let $Z_1,Z_2,\ldots$ denote i.i.d. centered Gaussian vectors in $\mathbb{R}^d$ which have the same covariance matrix as $\varepsilon_1$ and define the corresponding scan statistic
\begin{align} \label{e:seg:tildegamma}
\tilde \gamma_i(n,h)
=
    \frac{\sum_{m=n-h+1}^{n}Z_{m,i}
    -
    \sum_{m=n+1}^{n+h}Z_{m,i}}
    {N^{1/2-\beta}h^{\beta}}.
\end{align}
The following result is the segmentation analogue of Theorem~\ref{thm:main}. It is stated under the no-change model, which is needed for subsequent inference.

\begin{theo}
\label{thm:seg:gauss}
Suppose that Assumption~\ref{ass:main} holds and that no changes occur, i.e. \(K_i=0\) for all \(i=1,\ldots,d\). Then, for any $x_0>0$, it follows that
\begin{align}
\label{e:seg:HG}
\sup_{t \ge x_0}
\bigg|
\mathbb{P}\bigg(
\max_{1 \leq i \leq d,\,(n,h) \in \mathcal{S}^\star}
|\gamma_i(n,h)|
\le t
\bigg)
-
\mathbb{P}\bigg(
\max_{1 \leq i \leq d,\,(n,h) \in \mathcal{S}^\star}
|\tilde \gamma_i(n,h)|
\le t
\bigg)
\bigg|
=o(1).
\end{align}
\end{theo}

Before proceeding to the analysis of our scan procedures, we record a small result of mathematical interest. In our Hölder-Gauss approximation, we approximate the distribution of a scan statistic, i.e. increments of a partial sum process over a grid, by the corresponding scan over Gaussian data. It is interesting to note that the grid is not needed in the Gaussian case, at least in theory. One may instead take all weighted increments of the associated $d$-dimensional Brownian motion, which brings the result even closer to classical Hölderian invariance principles.

\begin{cor}
\label{cor:BrownianApprox:main}
Let $B$ be a $d$-dimensional centered Brownian motion, determined by the covariance matrix of $\varepsilon_1$. Suppose that Assumption \ref{ass:main}, parts \ref{ass1} and \ref{ass4}, are satisfied. Define
\[
    D^\beta
    :=
    \sup_{0<y<x<x+y<1}
    \frac{
    \big\|\{B(x)-B(x-y)\}-\{B(x+y)-B(x)\}\big\|
    }{y^\beta},
    \qquad \beta \in [0,1/2).
\]
Then, for any $x_0>0$,
\begin{align}
\sup_{t \geq x_0}
\bigg|
\mathbb{P}\big(D^\beta \le t\big)
-
\mathbb{P}\bigg(
\max_{1 \leq i \leq d,\,(n,h) \in \mathcal{S}^\star}
|\tilde \gamma_i(n,h)|
\le t
\bigg)
\bigg|
= o(1).
\end{align}
\end{cor}

Corollary \ref{cor:BrownianApprox:main} is mainly of conceptual interest, since quantile calibration will in practice rely on the discrete Gaussian statistic in Theorem~\ref{thm:seg:gauss}. Let \(q^N_{1-\alpha}\) denote the $(1-\alpha)$ quantile of
\[
\max_{1 \leq i \leq d,\,(n,h) \in \mathcal{S}^\star}
|\tilde \gamma_i(n,h)|.
\]
We can now turn to the statistical properties of COMBI-SCAN. 

\begin{cor}
\label{cor:seg:false}
Suppose that Assumption~\ref{ass:main} holds and that COMBI-SCAN is run with threshold $q^N_{1-\alpha}$. Then
\begin{align} \label{e:wlp}
\limsup_{N\to\infty}
\mathbb{P}\left(
\exists (I,i)\in\widehat{\mathcal I}:
I\cap\{c_{1,i},\ldots,c_{K_i,i}\}=\emptyset
\right)
\le \alpha.  
\end{align}
Now suppose that, additionally,
\begin{align} 
\label{e:msp:cond}
\min_{\substack{1\le i\le d\\1\le k\le K_i}}
\frac{
\delta_{k,i}^{\,2(1-\beta)}
\Delta_{k,i}^2
}{
N^{1-2\beta}\log(ed)}
\to \infty.
\end{align}
Then
\begin{align} \label{e:slp}
\mathbb{P}\left(
\forall i=1,\ldots,d,\ \forall k=1,\ldots,K_i,\
\exists (I,i)\in\widehat{\mathcal I}
\textnormal{ such that } 
I \cap \{c_{1,i},\ldots, c_{K_i,i}\} =\{c_{k,i}\}
\right)
\to 1.
\end{align}
\end{cor}

In \cite{kutta:dette:wang:2025}, the bound on the false discovery rate \eqref{e:wlp} was called the "weak localization property". It states that there are no wrongly detected changes except with probability asymptotically bounded by $\alpha$. The combination of \eqref{e:wlp} and \eqref{e:slp} was called the "strong localization property". It states, additionally to \eqref{e:wlp}, that each change point is detected and contained in a unique output interval. Strong localization requires a multiscale condition on the interplay of separation and height of the $k$th change in the $i$th component. All changes that are high enough and well separated enough will be detected with asymptotic certainty.
It is even possible to quantify the interval length $|I|$ needed to detect a change. Under \eqref{e:msp:cond}, it holds uniformly in $k,i$ and with probability tending to one that
\[
|I| \le C_d \bigg(\frac{\log^{1/2}(ed)N^{1/2-\beta}}{\Delta_{k,i}}\bigg)^{\frac{1}{1-\beta}}.
\]
Here $C_d>0$ is a constant that may be written as $C_d= C C_d'$, where $C$ depends only on the constant $C_2$ from Assumption \ref{ass:main}, part \ref{ass3}, and $(C_d')_{d\ge 1}$ is any monotonically increasing sequence going to $\infty$ with $C_1' \ge 1$. For polynomial dimensions, where $\beta$ can be chosen arbitrarily close to $1/2$, and for $\Delta_{k,i}$ uniformly bounded away from $0$, the interval width grows uniformly at a rate $\mathcal{O}_P(N^\zeta)$, with $\zeta=\zeta(\beta)$ arbitrarily close to $0$. It turns out that in the case of polynomially increasing dimensions, even sharper results are available; this is the subject of the next section.

\subsection{Optimality and dependence} \label{sec:opt:dep}

We want to conclude our theoretical section with some discussion on the optimality of the proposed procedures and on their applicability under dependence. We begin by commenting on optimality, and we focus for illustration on the high-dimensional segmentation task.  Suppose for this section that $d$ is only polynomially growing in $N$. Moreover, to motivate our following discussion, let us focus  on the case where all changes are relatively large, i.e. $\min_{k,i}\Delta_{k,i} \ge \Delta>0$ for $\Delta$ fixed. Then, the condition \eqref{e:msp:cond}, which is needed for reliable detection \eqref{e:slp} becomes
\begin{align} 
\label{e:msp:cond:2}
\min_{\substack{1\le i\le d\\1\le k\le K_i}}
\frac{
\delta_{k,i}^{\,2(1-\beta)}
}{
N^{1-2\beta}\log(ed)}
\to \infty \quad \Leftrightarrow \quad \min_{\substack{1\le i\le d\\1\le k\le K_i}}
\delta_{k,i} \gg 
\big(N^{1/2-\beta}\log^{1/2}(ed)\big)^{1/(1-\beta)}.
\end{align}
This means that (heuristically) for $\beta \uparrow 1/2$, we seem to approach a condition where even changes that persist only for some logarithm of $d$ (and thus $N$)  become detectable. 
The formal side of this argument is slightly more involved. Hölderian invariance principles were traditionally shown on the space of functions $f:[0,1] \to \mathbb{R}$ that are Hölder continuous in the sense that
\[
\lim_{h \downarrow 0}\sup_{0<|y-x|<h}\frac{|f(y)-f(x)|}{\rho(|y-x|)}=0,
\]
where $\rho(x)=x^\beta$; see also our Corollary \ref{cor:BrownianApprox:main}. However, it has been shown that for any appropriate modulus function $\rho(x)$ Hölderian invariance principles can be shown (under appropriate moment conditions), provided that $\rho$ decays slower as $x \downarrow 0$ than the function $\rho_{Levy}$ in Lévy's modulus of continuity given by $\rho_{Levy}(x) = x^{1/2} \log^{1/2}(1/x)$ . For example \cite{rackauskas:suquet:2009} suggest using $\rho_\beta(x) = x^{1/2}\log^\beta(e/x)$ for $\beta>1/2$. As shown in \cite{kutta:dette:wang:2025}, using $\rho_\beta(x) $ for weighting  yields narrower confidence intervals  and better detection properties against shortly persistent changes (if, of course, the Hölderian invariance principle can be proved).  This motivates the use of such "logarithmic weight functions" (terminology as in \cite{kutta:dette:wang:2025}) for our high-dimensional setting, with scan statistic
\begin{align}
\label{eq:def:segmentation:scan:2}
    \gamma_i(n,h)
    =
    \frac{\sum_{m=n-h+1}^{n}X_{m,i}-\sum_{m=n+1}^{n+h}X_{m,i}}{h^{1/2}\log^\beta(eN/h)}.
\end{align}
Notice that this motivation seems to only make sense when $d$ is only polynomial; if however $d$ grows with $N$ at stretched exponential rate both the justification and benefits seem unclear because $\delta_{k,i}\gg \log(ed) \sim N^c$ then mechanically requires at least polynomial separation. In Appendix \ref{appen:segment}, we have worked out in detail the Hölder-Gauss approximation for the more strongly weighted scan statistic from \eqref{eq:def:segmentation:scan:2}. It is not difficult to show that with this weight (and again under $\log(d) \le C \log(N)$), condition \eqref{e:msp:cond} can be weakened to  \begin{align} 
\label{e:msp:cond:3}
\min_{\substack{1\le i\le d\\1\le k\le K_i}}
\frac{
\delta_{k,i}
\Delta_{k,i}^2
}{
\log^{2\beta}(N)\log(ed)}
\to \infty,
\end{align}
and in the special case where change point heights are bounded away from $0$ (which before gave us \eqref{e:msp:cond:2})
we get
\[
\min_{\substack{1\le i\le d\\1\le k\le K_i}}
\delta_{k,i} \gg \log^{2\beta}(N)\log(ed).
\]
The right side is bounded by $\log^{2\beta+1}(N)$ for polynomially bounded dimension $d$. For fixed $d$, our condition $\min_{k,i} \delta_{k,i} \gg \log^{2\beta}(N)$ for some $\beta>1/2$ is close to the (essentially) optimal one of $\min_{k,i}  \gg \log(N)$ . Even though we have not formally proved it, it seems clear that using the suggested weighting schemes based on $\rho_\beta$ could also be used for the task of sequential monitoring. Here, stronger weighting yields  shorter detection delays, and heuristically, for polynomially growing dimension it seems that delay times of logarithmic orders are then possible, which again is close to the optimum, even for univariate data \cite{yu:padilla:wang:rinaldo:2023}.\\
Finally, we turn to the issue of dependence. For most of this work, we have assumed that noise vectors are independent, to allow for a clean presentation of our theory. Yet, Hölder-Gauss approximations remain valid even under weak dependence. To illustrate this fact, the Hölder-Gauss approximations for the strongly weighted scan \eqref{eq:def:segmentation:scan:2}  are justified in Appendix \ref{appen:segment} under $\beta$-mixing, which is a popular dependence concept in time series analysis; see \cite{bradley:2007}. The statements are only given in the Appendix, because they are slightly more technical and do require the formal introduction of mixing coefficients that might be distracting here. Our approximations would likely also remain true under different dependence concepts such as physical dependence, even though we do not work it out here. 
We point out that a practical problem implied by dependence is, of course, the estimation of the long-run covariance matrix, but we do not discuss this issue here any further and refer the interested reader to \cite{li:chen:wang:wu:2024} for an estimator that is suitable in our setting.

\section{Finite-sample performance} \label{sec:3}

In this section, we investigate finite-sample performance of our new methodology. To avoid redundancy, we restrict this investigation to the sequential change point detection part of this work, where there exists an established benchmark by \cite{gosmann:stoehr:heiny:dette:2022} to which we can compare our results. \\

\subsection{Simulations}

\noindent \textbf{Model.}
In our simulation study data are generated according to the following scheme:
\begin{align*}
    X_{ij}=\varepsilon_{ij}+\delta\mathbb{I}\{i> N+k^\star, j\leq s\}, \quad 1 \leq j \leq d, 1 \leq i \leq N+T_0,
\end{align*}
i.e. we have a reference period length of $N$, a monitoring horizon of $T_0$ with $s$ changes of size $\delta$ at time $k^\star$. Recall that our methods are justified for an open-end model (theoretically $T_0=\infty$ is justified), but for our simulations monitoring has to end at some point, which is specified by $T_0$.
The mean zero error vectors $\varepsilon_{i}$ are i.i.d. normal random variables such that
\begin{align*}
    \text{Cov}(\epsilon_i)=\Sigma_1=(\sigma^{(1)}_{ij})_{1 \leq i,j \leq d}=(\frac{1}{|i-j|+1})_{1 \leq i,j \leq d}. 
\end{align*}
In the following, we always perform 1000 tests to determine the empirical rejection rates with a nominal level of $\alpha=0.05$. The computation of the quantiles via Algorithm \ref{alg:bootstrap_quantile} is detailed \hyperref[par:computation-quantiles]{below}. 
Our numerical results under the null hypothesis support the choice $\beta = 0.4$, which is kept fixed throughout all simulations under the alternative.

\noindent \textbf{Nominal approximation.} We first study the impact of the dimension $d \in \{100,200,400,800,1600\}$ on the approximation of the nominal level for different choices of $\beta \in \{0.1,0.2,0.3,0.4\}$. To investigate the impact of non-normality of the data, we also include results where $\epsilon_{ij}\overset{iid}{\sim}U([-1,1])$. The reference sample has a size of $N=100$ with monitoring horizon $T_0=100$.  
The results are recorded in Table \ref{tab:level} below and show that the choice of $\beta$ has no discernible impact on the quality of the level approximation. Consequently we recommend choosing a $\beta\geq 0.4$ as our theoretical results suggest that $\beta$ closer to 1/2 has smaller detection delays and better power. For the remainder of the section we will therefore apply the proposed tests with $\beta=0.4$. 

\begin{table}[ht]
\centering
\begin{tabular}{c|ccccc|ccccc}
\hline
& \multicolumn{5}{c|}{Normal data} & \multicolumn{5}{c}{Uniform data} \\
\hline
$\beta \backslash d$
& 100 & 200 & 400 & 800 & 1600
& 100 & 200 & 400 & 800 & 1600 \\
\hline
0.1 & 0.043 & 0.022 & 0.039 & 0.036 & 0.022
    & 0.048 & 0.042 & 0.044 & 0.041 & 0.027 \\
0.2 & 0.037 & 0.031 & 0.026 & 0.019 & 0.027
    & 0.039 & 0.042 & 0.045 & 0.048 & 0.040 \\
0.3 & 0.040 & 0.030 & 0.030 & 0.028 & 0.027
    & 0.051 & 0.042 & 0.040 & 0.038 & 0.032 \\
0.4 & 0.038 & 0.032 & 0.023 & 0.027 & 0.025
    & 0.041 & 0.034 & 0.031 & 0.042 & 0.026 \\
\hline
\end{tabular}
\caption{Rejection rates under the null for different choices of $\beta$ and $d$ 
for normal and uniform data. }
\label{tab:level}
\end{table}

Before turning to the power analysis, we first describe how we obtain the quantiles for each method.

\phantomsection
\noindent\textbf{Computation of quantiles.}\label{par:computation-quantiles}
To calculate the bootstrap quantiles via Algorithm \ref{alg:bootstrap_quantile} (both for our and for the method from  \cite{gosmann:stoehr:heiny:dette:2022}), we set the number of bootstrap samples as $B=200$. That algorithm requires a choice of the monitoring horizon $T$, and we always set $T=T_0$ here.
 As we will see in the results reported in Table \ref{tab:sim_power_combined_short}, the proposed method has a significantly higher power for late changes (for which only a short post-change segment is observed) than the one proposed in \cite{gosmann:stoehr:heiny:dette:2022}.
The fact that our procedure has higher power even when little post-change data is available will be discussed also in the below data application, because it translates into higher power of the proposed method against transient changes. \\
The computational complexity of  calculating the test statistic for the whole time period from $N+1$ to $N+T$ for either method is of order $dT^2$.
One may reduce computational complexity at the expense of less power against subtle changes, by scanning for changes only in recent data (say using only the last $B$ observations). 
This lowers the computational complexity to $\mathcal{O}(dTN)$, which scales only linearly in the length of the monitoring time $T$. \\  \cite{gosmann:stoehr:heiny:dette:2022} provide three different ways to obtain the necessary quantiles. Two of them are based on an extreme value approach that requires weak cross-sectional dependence (and a corresponding ordering of the components) and are observed to perform comparatively poorly to the third choice in their simulation studies. The third choice is a bootstrap procedure that the authors justify theoretically only in the case of cross-sectional independence. Its theoretical development is very similar to choosing $\beta=0$ for the method proposed in this paper and thus our theory suggests its validity even in the presence of cross-sectional dependence. The following comparisons are always based on the third approach. We also note that the choice of covariance matrix has little impact on the qualitative conclusions of the comparison and only influences the relationship between the dimension and the power for both methods. 
\\

\noindent\textbf{Power.}
Again we consider a moderate reference sample size $N=100$ together with short to moderate monitoring periods $T_0 \in \{100,200,400\}$. The change takes place in $s \in \{1,50\}$ components to model either a sparse or moderately dense setting. The size $\delta$ of the change will be specified in the tables below and is tailored towards choosing regimes with non-trivial power for the sake of comparison. We consider up to two different choices for the dimension $d \in \{100,200\}$ to elucidate its impact.  For the change location $k^\star\in\{T_0/2,(3/4)T_0\}$ we consider a change in the middle and a change towards the end of the monitoring period.\\
We report the results in Tables \ref{tab:sim_power_combined} and \ref{tab:sim_power_combined_short} below. We observe that the proposed method outperforms the one from \cite{gosmann:stoehr:heiny:dette:2022} for the sparse setting with $s=1$, particularly when the change is located towards the end of the monitoring period. For the denser setting with $s=50$ the proposed method performs better for later changes, while  \cite{gosmann:stoehr:heiny:dette:2022} have a slight edge for earlier changes, particularly when $\delta=0.375$. Simulations for $T_0>4N$ indicate that our method performs better (in comparison) the later the change and the longer the time horizon is, which is a result of the weight $\omega_\beta(\ell,k,N)$ decaying more quickly the smaller $\beta$ is. We elected to present the setting most favorable to \cite{gosmann:stoehr:heiny:dette:2022} to make the comparison as fair as possible.

\begin{table}[ht]
\centering
\small
\begin{tabular}{ccc|ccc|ccc}
\toprule
& & & \multicolumn{3}{c|}{s = 1} & \multicolumn{3}{c}{s = 50} \\
& & & \multicolumn{3}{c|}{$\delta$} & \multicolumn{3}{c}{$\delta$} \\
\cmidrule(lr){4-6} \cmidrule(l){7-9}
$T_0$ & Method & $d$ & 0.600 & 0.700 & 0.800 & 0.250 & 0.375 & 0.500 \\
\midrule
\multirow{6}{*}{100}
& \cite{gosmann:stoehr:heiny:dette:2022} & 100 & 0.156 & 0.243 & 0.392 & 0.173 & 0.455 & 0.753\\
& Proposed & 100 & 0.152 & 0.290 & 0.486 & 0.151 & 0.371 & 0.790 \\
& \cite{gosmann:stoehr:heiny:dette:2022} & 200 & 0.106 & 0.184 & 0.272 & 0.128 & 0.314 & 0.641 \\
& Proposed & 200 & 0.111 & 0.235 & 0.406 & 0.084 & 0.269 & 0.677 \\
\midrule
\multirow{6}{*}{200}
& \cite{gosmann:stoehr:heiny:dette:2022} & 100 & 0.218 & 0.343 & 0.595 & 0.220& 0.590 & 0.900\\
& Proposed & 100 & 0.259 & 0.450 & 0.684 & 0.168 & 0.430 & 0.851 \\
& \cite{gosmann:stoehr:heiny:dette:2022} & 200 & 0.142& 0.306 &0.453 & 0.161 &0.420 &  0.828\\
& Proposed & 200 & 0.171 & 0.371 & 0.611 & 0.118 & 0.321 & 0.768 \\
\midrule
\multirow{6}{*}{400}
& \cite{gosmann:stoehr:heiny:dette:2022} & 100 & 0.280 & 0.493 & 0.790 & 0.243& 0.661 & 0.945 \\
& Proposed & 100 & 0.314 & 0.628 & 0.852 & 0.221 & 0.531 & 0.932 \\
& \cite{gosmann:stoehr:heiny:dette:2022} & 200 &  0.231& 0.412 &0.657 &0.175  & 0.524&0.885  \\
& Proposed & 200 & 0.235& 0.515 & 0.791 & 0.147 & 0.389 & 0.865 \\
\bottomrule
\end{tabular}
\caption{Empirical rejection rates for the proposed method compared to that presented in \cite{gosmann:stoehr:heiny:dette:2022}.  Change is located at $k^\star=1/2T_0$, with a reference sample size of $N=100$.  }
\label{tab:sim_power_combined}
\end{table}

\begin{table}[ht]
\centering
\small
\begin{tabular}{ccc|ccc|ccc}
\toprule
& & & \multicolumn{3}{c}{s = 1}  & \multicolumn{3}{c}{s = 50}  \\
& & & \multicolumn{3}{c}{$\delta$} & \multicolumn{3}{c}{$\delta$} \\
\cmidrule(lr){4-6} \cmidrule(l){7-9}
$T_0$ & Method & $d$ & 1.00 & 1.125 & 1.25 &  0.6 & 0.7 & 0.8\\
\midrule
\multirow{3}{*}{100}
& \cite{gosmann:stoehr:heiny:dette:2022} & 100 & 0.077 & 0.117  & 0.146 & 0.241 & 0.349 & 0.512  \\
& Proposed & 100 & 0.178 & 0.319 & 0.476 & 0.224 & 0.413 & 0.676\\
\midrule
\multirow{3}{*}{200}
& \cite{gosmann:stoehr:heiny:dette:2022} & 100 & 0.095 & 0.169 & 0.261 &0.364 & 0.468 & 0.657  \\
& Proposed & 100 & 0.306 & 0.575 & 0.784 & 0.266 & 0.507 & 0.798  \\
\midrule
\multirow{3}{*}{400}
& \cite{gosmann:stoehr:heiny:dette:2022} & 100 & 0.145 & 0.229& 0.335 & 0.380 & 0.564 & 0.709 \\
& Proposed & 100 & 0.461 & 0.763 & 0.947  & 0.252 & 0.531 & 0.887\\

\bottomrule
\end{tabular}
\caption{Empirical rejection rates for the proposed method compared to that presented in \cite{gosmann:stoehr:heiny:dette:2022}.  Change is located at $k^\star=3/4T_0$, with a reference sample size of $N=100$.  }
\label{tab:sim_power_combined_short}
\end{table}

\noindent \textbf{Detection delay evaluation.}
We again consider $N=100$, this time for $(\delta,s)=(2,1)$ and $T_0 \in \{100,200,400\}$; this means we consider a fairly large change in a single component. The reason for a large change is that for our consideration of delay times we want
to ensure that both methods always detect the change, to disentangle delay from power considerations. We consider the two change locations $k^\star\in \{1/2T_0,3/4T_0\}$ and fix $d=100$ (results for other dimensions are similar). \\
We observe that, as $T_0$ increases, the bulk of the delay distribution for the proposed method gradually separates from that of \cite{gosmann:stoehr:heiny:dette:2022}. This effect is also reflected in the growth of the average delay. For instance, for $k^\star=3/4T_0$, average delays of our method are roughly half that of the benchmark; for $T_0=(100,200,400)$  they are $(11.53,17.37,28.59)$ for the proposed method, and $(19.32,34.21,62.75)$ for the method from \cite{gosmann:stoehr:heiny:dette:2022}. The discrepancy becomes larger the later the change and the longer the whole monitoring period is.

\begin{figure}[H]
\begin{center}
    \includegraphics[scale=0.3]{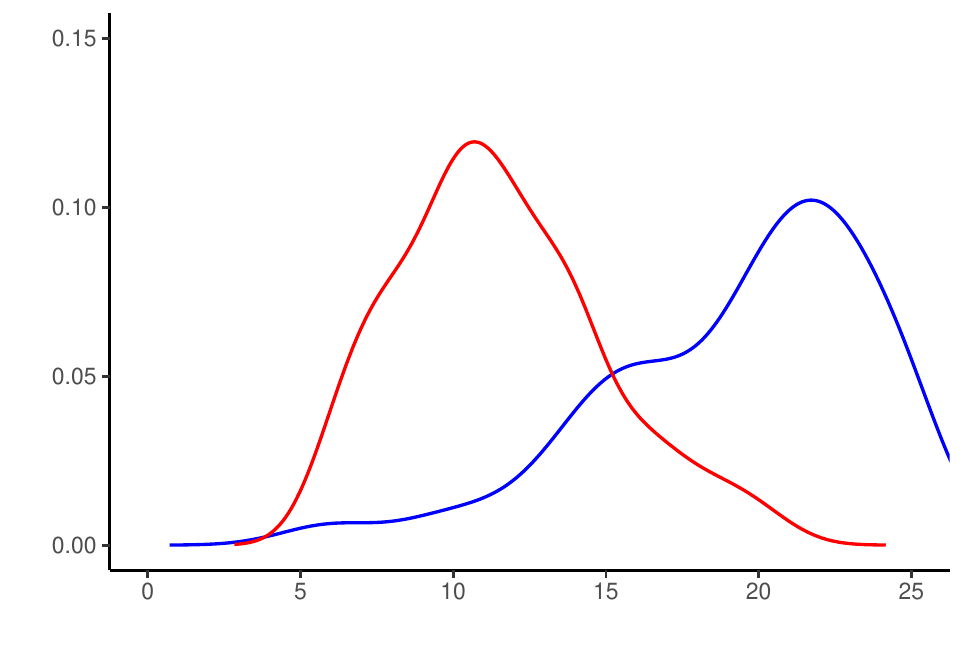}
    \includegraphics[scale=0.3]{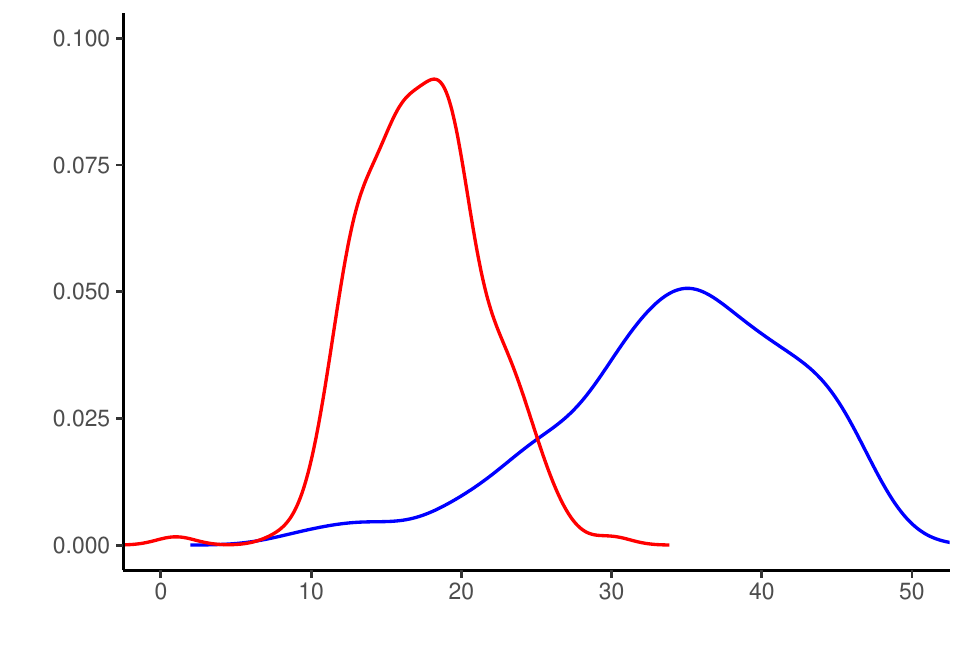}
    \includegraphics[scale=0.3]{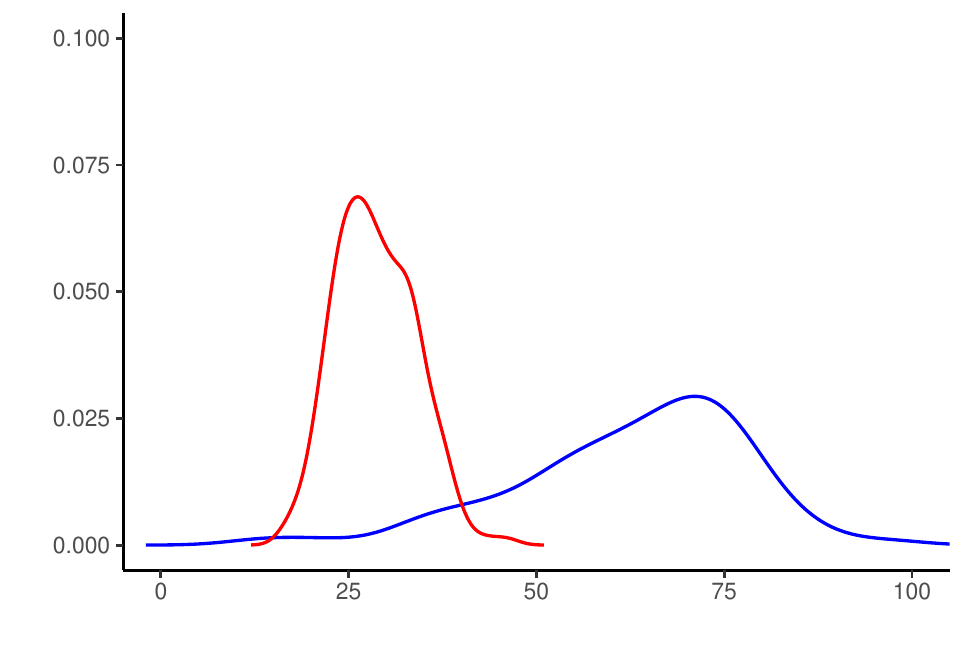}   
\end{center}
 \begin{center}
    \hspace{0.2cm}
    \includegraphics[scale=0.28]{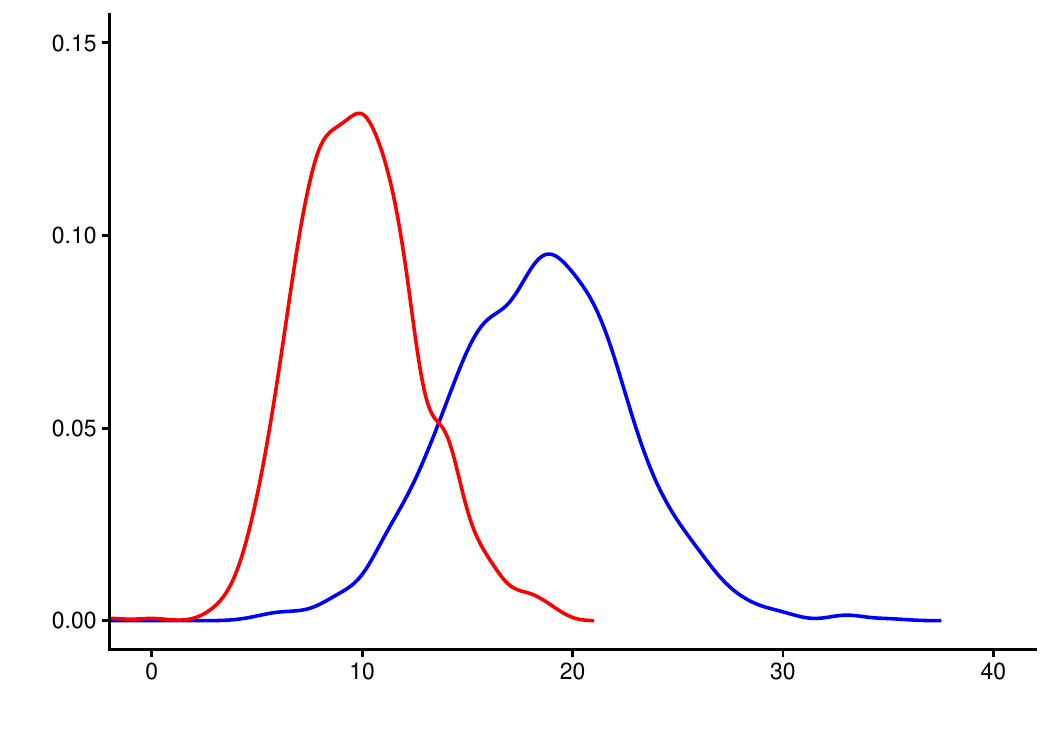}
    \includegraphics[scale=0.27]{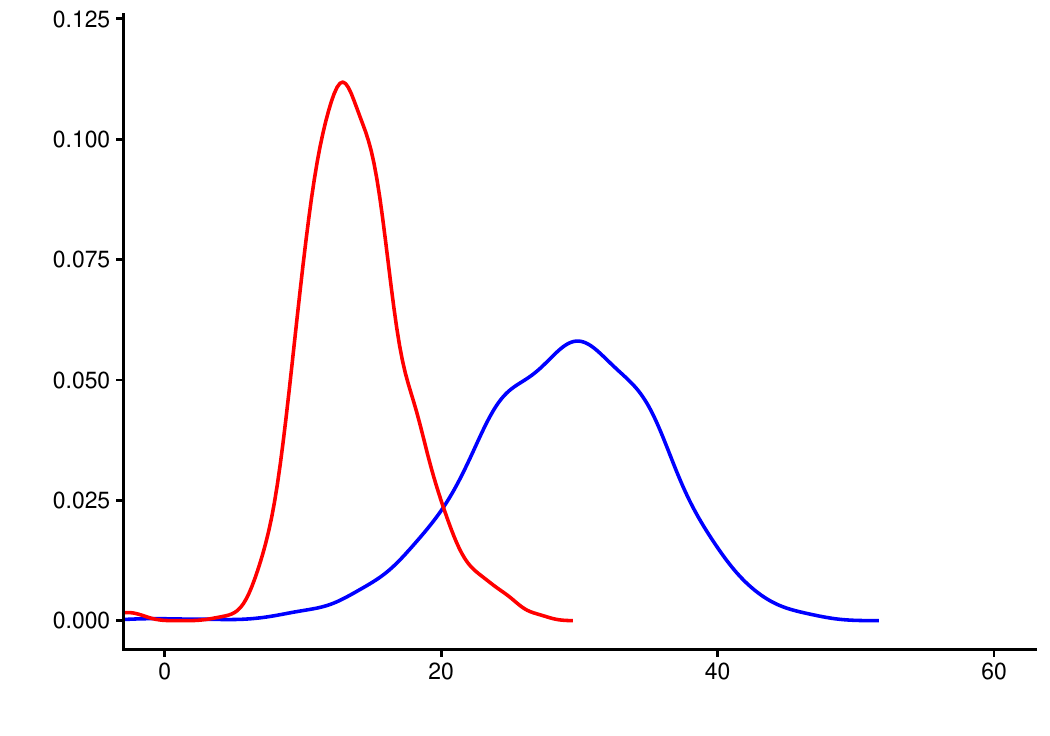} 
    \includegraphics[scale=0.28]{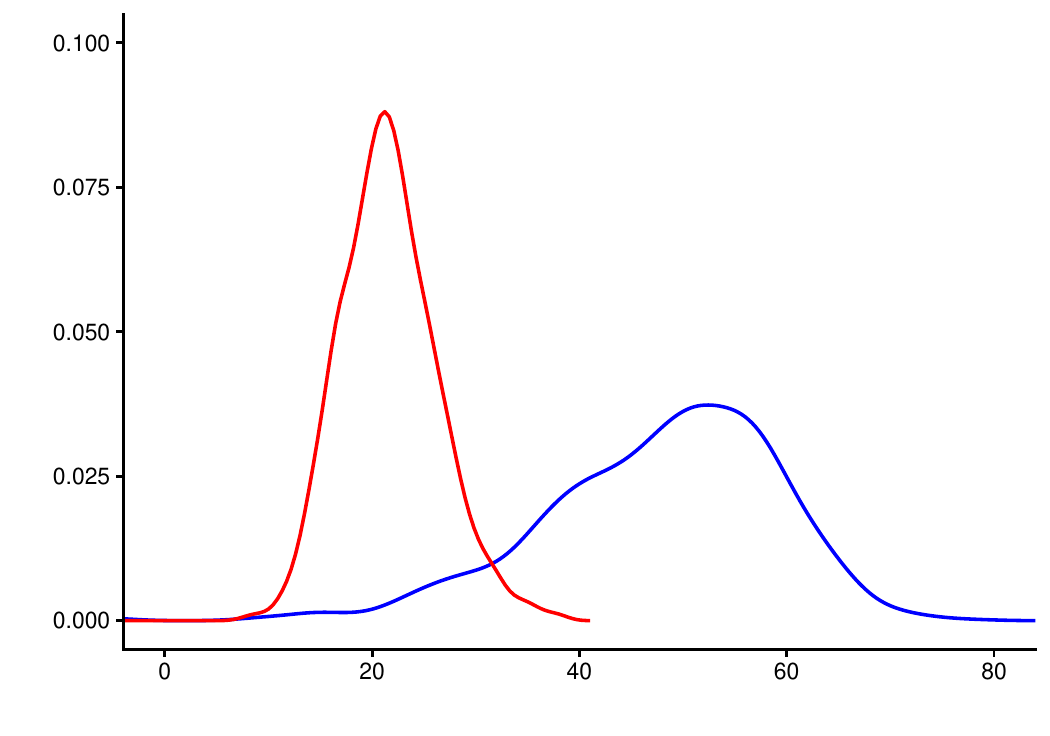}
\end{center}
    \caption{Density estimates of the delay times for $k^\star=3T_0/4$ (upper three) and $k^\star=T_0/2$ (lower three) for the method from \cite{gosmann:stoehr:heiny:dette:2022} (blue) and the method proposed in this paper with $\beta=0.4$ (red). The reference sample size is $N=100$ with the signal of size $\delta=2$ being present in $s=1$ of the $d=100$ components. The monitoring horizon $T_0$ is given by 100 (left), 200 (middle) and 400 (right).  }
    \label{fig:delays}
   
\end{figure}

\subsection{Data example}

We analyze hourly PM$2.5$ measurements from the U.S. Environmental Protection Agency (EPA) AirData database in California\footnote{The data set is available from \url{https://aqs.epa.gov/aqsweb/airdata/download_files.html}.} across different years. PM$2.5$ refers to tiny particles in the air that have a diameter of at most $2.5$ micrometers. PM$2.5$  concentration is a proxy for air pollution and can be used to detect wild-fires, which are associated with elevated  PM$2.5$ concentration. We analyze data from the years 2018, 2019 and 2020, which represent markedly different wildfire seasons. According to the California Department of Forestry and Fire Protection, 2019 was a comparatively mild wildfire season, whereas 2018
was highly destructive and 2020 was the largest wildfire season in
California's modern history
\citep{CALFIRE2018,CALFIRE2019,CALFIRE2020}. The purpose of this analysis is to apply our new sequential methodology to some real datasets, and to illustrate a hypothetical use-case. More precisely, we apply the statistical method developed in Section \ref{sec:2}, which is formulated for a data-stream with independent innovations. Notice that our model assumptions are somewhat stylized in this setting: In practice, PM$2.5$ values may not be perfectly independent across days, and may not have a perfectly piecewise constant mean. Nevertheless, wildfires are associated with sudden spikes in  PM$2.5$ values, much larger than typical deviations, making a  model with sudden jumps approximately plausible; the fact that changes usually do not persist, because wildfires are eventually put out, is a discrepancy to our theoretical model and will be discussed more below. Apart from model assumptions notice that, when analyzing datasets from past years, we are not genuinely conducting a real-time analysis. Rather, we conduct a "pseudo-real-time analysis", to study the behavior of our statistic. As in the previous section, we also compare our results to the established benchmark from \cite{gosmann:stoehr:heiny:dette:2022}.
We provide some additional details on the data: For preprocessing, we average the hourly PM${2.5}$ measurements
to daily values and keep monitoring sites for which at least \(80\%\) of
the hourly measurements are available over the respective year. We
remove days with more than ten missing observations and
impute the remaining missing values separately for each station by
linear interpolation. 
For each  year in $\{2018, 2019, 2020\}$, this gives a multivariate time series with $69$, $74$ and $78$ coordinates, respectively, each corresponding to one monitoring station. The geographical locations of the stations are displayed on the map in Figure \ref{fig:monitoring_stations}.
  We use the first $N=150$ days as the training period. 
  \begin{figure}[t]
    \centering
    \includegraphics[
        width=0.62\textwidth
    ]{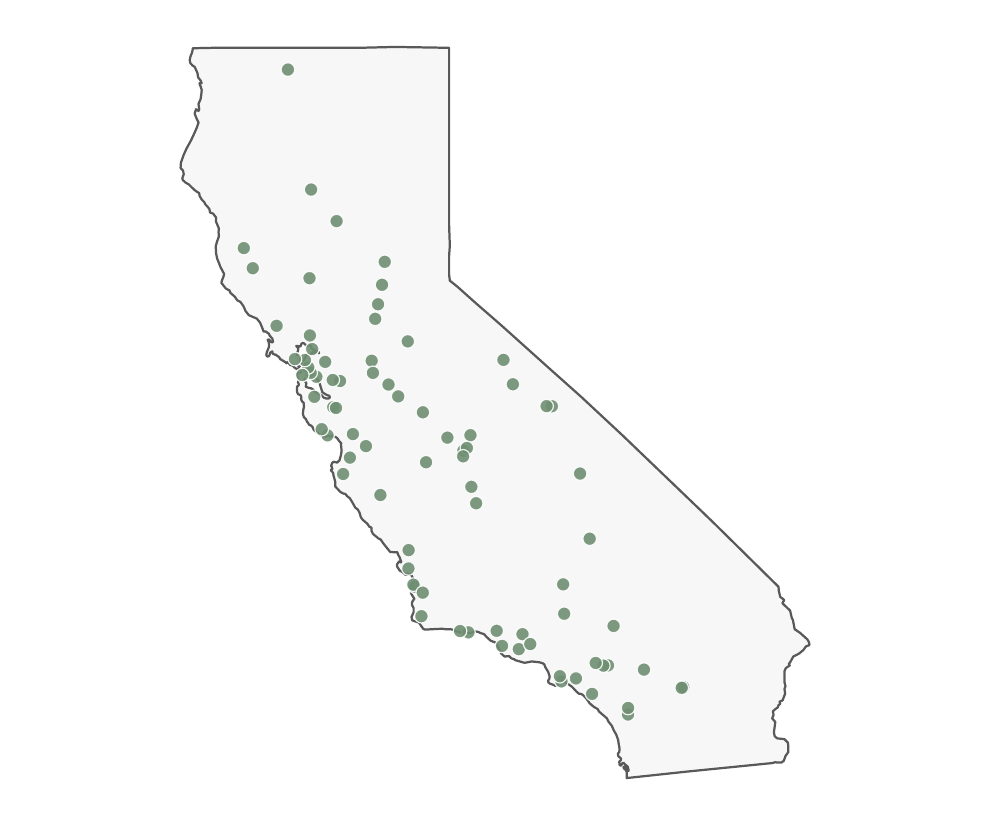}
    \caption{
        Locations of the PM${2.5}$ monitoring stations included in at
        least one of the analyses for 2018--2020.
    }
    \label{fig:monitoring_stations}
\end{figure}
To describe our statistical outcomes, we first notice that for most stations our proposed method and the benchmark from \cite{gosmann:stoehr:heiny:dette:2022} agree on the (non-)significance of a change. Indeed, both methods lead to the same conclusion for approximately $95\%$ of the stations considered in 2018, $91\%$ in 2019, and $100\%$ in 2020.  This congruence is a first plausibility check of our results. Next, we compare detection times. 
Figure \ref{fig:scatter} jointly reports the detection times of our proposed method with those of the competing method \cite{gosmann:stoehr:heiny:dette:2022} for all stations where both methods reject. This means that each point corresponds to one monitoring station. Red points indicate that our method detects the change earlier, blue points indicate that the method of \cite{gosmann:stoehr:heiny:dette:2022} detects earlier, and gray points correspond to equal detection times. We observe that our method delivers earlier detection times in an overwhelming majority of the cases. This effect is particularly pronounced in 2018 and 2020 (the stronger wild-fire seasons), where almost all points lie above the diagonal and several differences are substantial. The results for 2019 are slightly more mixed but still our method detects more than 2/3 of all changes earlier. To interpret Figure \ref{fig:scatter}, it is important to recall that it represents detection times and not delays (which are unknown without the true change locations). Accordingly, for any single change, it is unclear whether a value slightly above the line represents a relevant shortening of detection delay or not.
Still, qualitatively the earlier detections suggest generally shorter delays by our new method; the gap is often more than a week, which would be highly relevant for wildfire detection. Shorter delays fit the findings from our simulation study, see Figure \ref{fig:delays}. We also notice that there exists a non-negligible minority of changes with very large detection gap; this is most easily visible in 2018, but some cases also exist in 2020. These gaps can be up to 100 days long, and we believe they might have a distinct interpretation beyond just generally "shorter delays of the proposed method". Indeed, we believe that stations where the gap is very large might have experienced two separate spikes of PM${2.5}$ levels, with the first one detected by our proposed method and the second (much later one) by the benchmark method of \cite{gosmann:stoehr:heiny:dette:2022}. This is not really a difference in detection delay; rather it points to a different aspect related to power. Our theoretical model \eqref{e:amoc} is very stylized and postulates at most one change per component. A change then persists forever, making it relatively easy to detect. In reality, there often exist multiple transient changes, and wildfire data  is a good example of this, because every fire is eventually put out. 

\begin{figure}[h]
    \centering

    \begin{minipage}{0.32\textwidth}
        \centering
        \includegraphics[width=\linewidth]{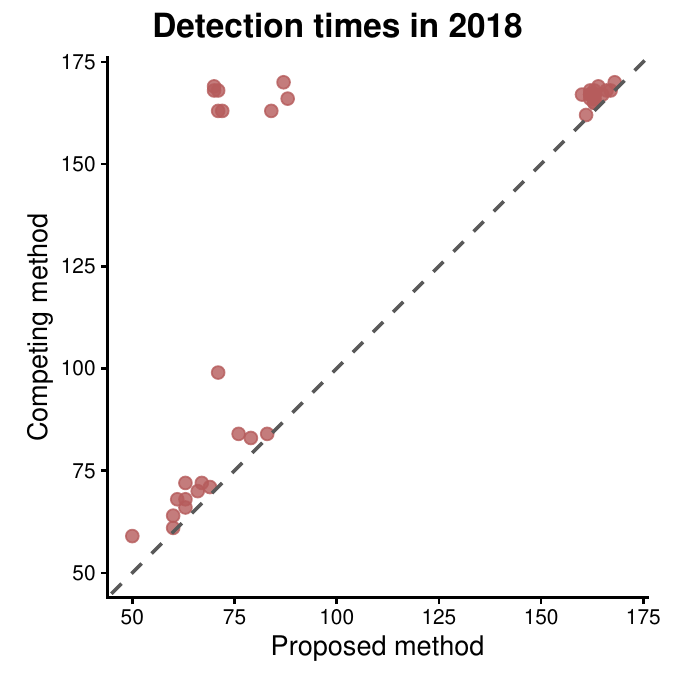}
    \end{minipage}
    \hfill
    \begin{minipage}{0.32\textwidth}
        \centering
        \includegraphics[width=\linewidth]{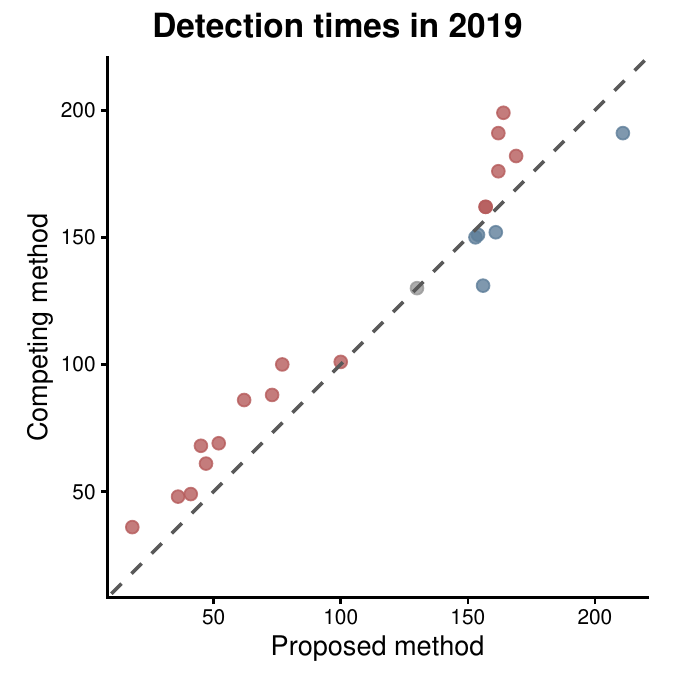}
    \end{minipage}
    \hfill
    \begin{minipage}{0.32\textwidth}
        \centering
        \includegraphics[width=\linewidth]{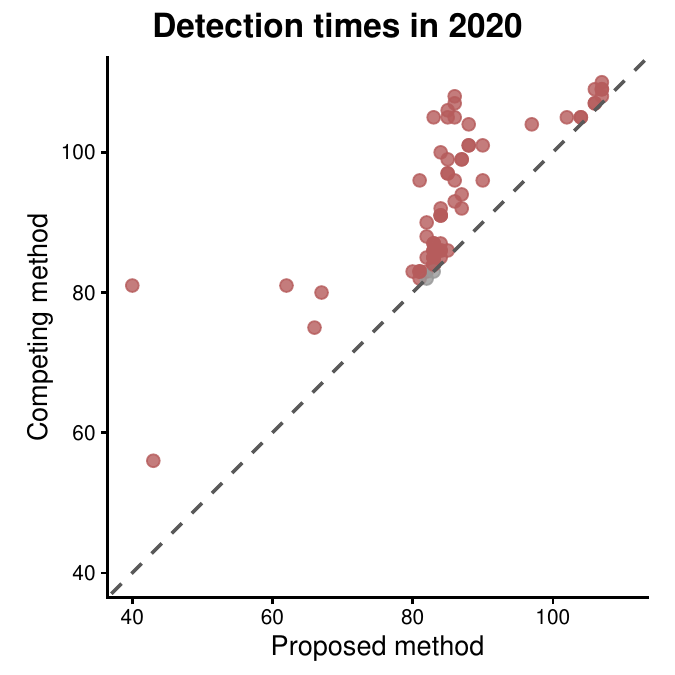}
    \end{minipage}

    \caption{
        Detection times of the proposed and competing methods for the California PM${2.5}$ data. Each point corresponds to one monitoring station. Red points indicate earlier detection by the proposed method, blue points indicate earlier detection by the method of \cite{gosmann:stoehr:heiny:dette:2022}, and grey points indicate equal detection times. The dashed line corresponds to equal detection times.
    }
    \label{fig:scatter}
\end{figure}

Against such transient changes Hölder statistics are distinctly more powerful than traditional methods such as  \cite{gosmann:stoehr:heiny:dette:2022} (see for a discussion of retrospective detection of transient changes eg. \cite{RackauskasSuquet2004}).
If, however, \cite{gosmann:stoehr:heiny:dette:2022} misses the first burst of PM${2.5}$ levels, but catches a second burst (due to some later wildfires) this would show in Figure \ref{fig:scatter} as large detection gaps. What it really demonstrates is an advantage of power of Hölder-type statistics against short transient changes. To evaluate this intuition, we consider the year 2018 and take the three measurement stations with the largest detection gap. Our first observation is that all of these stations were geographically quite close together (at most about 30 kilometers). Next, in  Figure \ref{fig:three:paths} we display the data from the monitoring period of these three stations, together with the first anomaly detection from our method and from the method by \cite{gosmann:stoehr:heiny:dette:2022}.  Figure \ref{fig:three:paths} confirms our theoretical intuition: There exist two spikes of PM$2.5$ that year, one in the summer (main wildfire season) detected by our proposed method but missed by \cite{gosmann:stoehr:heiny:dette:2022}. Their method then detects a later spike in the winter of 2018, likely due to a late wave of fires that year. This case reminds us that changes in applications are often transient and methods that are powerful against such transient changes, like Hölder-type statistics presented in this work, may be practically more powerful than the theoretical monitoring model suggests.

\begin{figure}[H]
    \centering
    \includegraphics[width=0.7\textwidth]{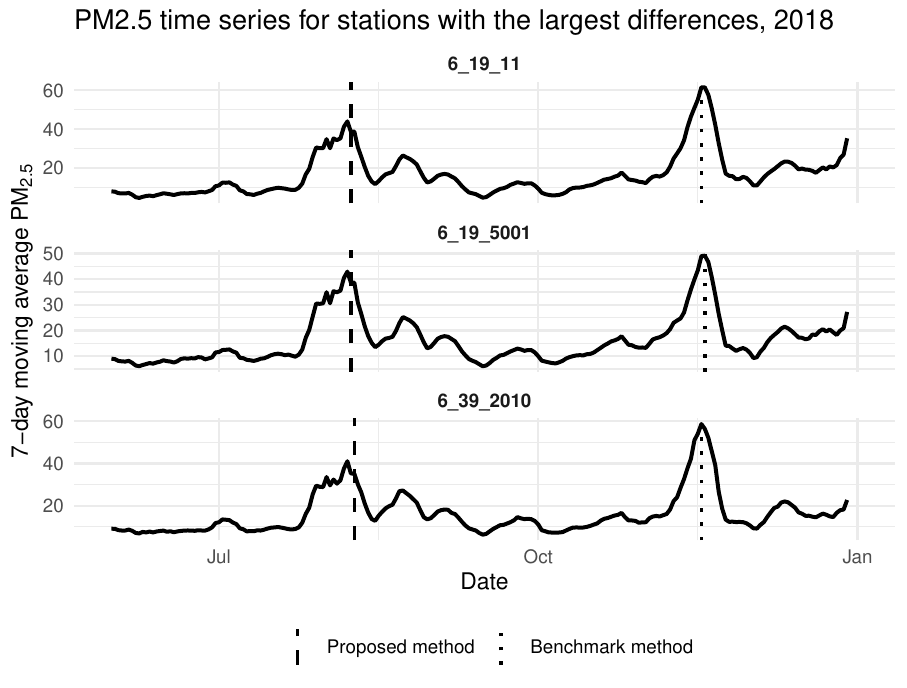}
 \caption{Seven-day moving averages of PM$_{2.5}$ concentrations in 2018 at the three stations with the largest differences between the detection times of the proposed and competing method \cite{gosmann:stoehr:heiny:dette:2022}. }
    \label{fig:three:paths}
\end{figure}

\section*{Acknowledgements}
The authors were partially supported by the Aarhus University Research Foundation (AUFF), project numbers 47331, 47222,  47221 and 47388.

\bibliographystyle{abbrv}
\bibliography{reference}

\newpage 
\appendix
\begin{center}
{\LARGE\bfseries Supplement}
\end{center}
$ $\\
The supplement is dedicated to the proofs of our main theoretical results. It is structured as follows: Section \ref{appen:A} contains the proof for the Hölder-Gauss approximations for the monitoring. Section \ref{appen:B} contains the proofs of the other statement regarding the monitoring. Section \ref{appen:segment} contains the proofs for the Hölder-Gauss approximation for the segmentation procedure under dependence, divided into subsections for the logarithmi and the polynomail weight. A third subsection contains concentration results for small scales. Section \ref{appen:GAP-dependent} contains high dimensional Gaussian approximations for dependent data, which are used in the proofs of Section \ref{appen:segment}. Section \ref{appen:brownian-approx} refines the Hölder-Gauss approximations from Section \ref{appen:segment}, showing that they can be represented via Hölder norms of Brownian motions. This is also used to obtain the precise presentation used in the statement of Theorem \ref{thm:seg:gauss} in the main text. Finally, Section \ref{appen:aux_res} contains auxiliary results used in all other sections of the Appendix.

\section{Proof of Theorem \ref{thm:main}}
\label{appen:A}
We begin by noting that Assumption \ref{ass:main}, part \ref{ass1} can be weakened in this proof to  $C_1\leq \frac{1-2\beta}{12-22\beta}$. The stricter bound in the main article is chosen to give a single condition for sequential monitoring and segmentation.\\

\noindent Central to our proof is a decomposition of $\sup \widehat{\Gamma}_N^{\beta}(k)$ over \textit{detectors corresponding to short- and long-time horizons}. To this end, choose fixed constants $\delta$ and $\zeta$, sufficiently
close to $1/(11-20\beta)$ and $2C_1$, respectively, such that
\begin{align}\label{ass:delta_eta}
    &\frac{1}{11-20\beta}<\delta<1,
    \qquad 2C_1<\zeta<\frac12,\\
    &C_1<\delta(1-2\beta),
    \qquad 2\beta\delta+11C_1<1,\nonumber\\
    &(1-\beta)(\delta+\zeta)+\frac{C_1}{2}<1+\zeta,\nonumber\\
    &\left(\frac23-\beta\right)(\delta+\zeta)+C_1<\frac23.\nonumber
\end{align}
Such a choice exists under the bound on $C_1$ stated above. Then, we define the detectors corresponding to short- and long-time horizons as $A_{1,N}$ and $A_{2,N}$, respectively, where 
\begin{align*}
    A_{1,N} := & \sup_{1 \le k\le N^{1+\zeta}}\widehat{\Gamma}_N^{\beta}(k), \quad 
    A_{2,N} :=  \sup_{k> N^{1+\zeta}}\widehat{\Gamma}_N^{\beta}(k).
\end{align*}
This allows us to decompose the detector as 
\begin{align}\label{eq:def_A1}
\sup_{k\ge 1}\widehat{\Gamma}_N^{\beta}(k) =\max\{A_{1,N},A_{2,N}\}.
\end{align}
Starting from this decomposition, we proceed in two steps. 
\begin{enumerate}
    \item[(1)] In the first step, we provide a further approximation for long-time horizons. In particular, Lemma \ref{lem1} stated below will show that $A_{2,N}$ can be approximated by some $A_{2,1,N}$ defined below.
    \item[(2)] By applying step (1) to $\widehat{\Gamma}_N^{\beta}(k)$ in \eqref{eq:def_A1} and its respective Gaussian counterpart $\tilde{\Gamma}_N^{\beta}(k)$, we will see that it is sufficient to show a Hölder-Gauss approximation for $\max(A_{1,N}, A_{2,1,N}).$  The latter will be provided in Lemma \ref{lemma:main_gauss_apprx}.
\end{enumerate}

To formalize this approach, we define
\begin{align}\label{e:def:A21N}
    A_{2,1,N}:=
    \sup_{k> N^{1+\zeta}}
    \max_{0 \leq \ell < k} \bigg\{w_{\beta}(\ell, k, N) \cdot \frac{k-\ell}{N}\bigg\} \Big\|\sum_{n=1}^{N} \varepsilon_{n} \Big\| \, .  
    \end{align}
We note that under the null hypothesis $H_0$, the change point detector $\widehat{\Gamma}_N^{\beta}(k)$ in \eqref{e:statmain} can be expressed as
\begin{align} \label{e:fullcwH0}
\widehat{\Gamma}_N^{\beta}(k):=  &\max_{0 \leq \ell < k} \bigg\{w_{\beta}(\ell, k, N) \cdot \Big\|\frac{k-\ell}{N} \sum_{n=1}^{N} \varepsilon_{n}-\sum_{n=N+\ell+1}^{N+k}\varepsilon_{n} \Big\| \bigg\}\, .  
\end{align}

\begin{lemma} \label{lem1}
    Suppose that the assumptions of Theorem \ref{thm:main} are satisfied. Then, it holds
    \[
    A_{2,N} = A_{2,1,N}+o_P(1).
    \]
\end{lemma}
The proof is deferred to Section \ref{sec_proof_lem1}.
Next, note that Lemma \ref{lem1} implies that
\[
\sup_{k\ge 1}\widehat{\Gamma}_N^{\beta}(k)  = \max\{A_{1,N}, A_{2,1,N}\} +o_P(1).
\]
Applying the same results for the Gaussian case, we obtain that
\[
\sup_{k\ge 1}\widetilde{\Gamma}_N^{\beta}(k) =  \max\{\tilde A_{1,N}, \tilde A_{2,1,N}\} +o_P(1),
\]
where $\tilde A_{1,N}$, $\tilde A_{2,1,N}$ are the Gaussian counterparts of $A_{1,N}$ (defined in \eqref{eq:def_A1}), $A_{2,1,N}$ (defined in \eqref{e:def:A21N}), that is, with $\varepsilon_i$ replaced by $Z_i$ everywhere. We recall that by $Z_1,Z_2,...$ we denote  i.i.d. centered Gaussian vectors in $\mathbb{R}^d$ which have the same covariance matrix as $\varepsilon_1$. 
Our key technical lemma entails the following Hölder-Gauss approximation.
\begin{lemma} \label{lemma:main_gauss_apprx}
    Suppose that the assumptions of Theorem \ref{thm:main} are satisfied.
 
    Then, it holds for any $x_0>0$
\[
\lim_{N\to\infty} \sup_{x > x_0 }\big| \mathbb{P}\big(\max\{ A_{1,N},  A_{2,1,N}\}\le x\big)-\mathbb{P}\big(\max\{ \tilde A_{1,N},  \tilde A_{2,1,N}\}\le x\big)\big| =0.
\]
\end{lemma}
This result will be proven in Section \ref{sec_proof_lem_gauss_apprx}.

\subsection{Proof of Lemma \ref{lemma:main_gauss_apprx}} \label{sec_proof_lem_gauss_apprx}

We start by simplifying the supremum in $A_{2,1,N}$; for a definition recall equation \eqref{e:def:A21N}. 
For this, we substitute in the definition for the weight function $w_\beta(\ell,k,N)$, yielding
\begin{align}
    A_{2,1,N}
    =\sup_{k\ge N^{1+\zeta}}\max_{0\le\ell<k}\bigg\{
   \frac{(k-\ell)^{1-\beta}}{N^{1/2}(N+k)^{1-\beta}}\bigg\}\bigg\|\sum_{n=1}^N\varepsilon_n\bigg\|.
\end{align}
To find the maximizing pair $(\ell, k)$, the norm of the random variables is irrelevant, because it is independent of both $\ell$ and $k$. To analyze the supremum then, we first focus on the inner maximum.
Since $\beta\in[0,1/2)$, it is easy to see that the maximum is attained when $\ell=0$.  Fixing $\ell=0$, the remaining function is monotonously increasing in $k$, hence its supremum is attained as $k\rightarrow+\infty$, which gives
\begin{align} \label{e:der:A21N}
    A_{2,1,N}=N^{-1/2}\bigg\|\sum_{n=1}^N\varepsilon_n\bigg\|.
\end{align}
We will below only work with the simplified form of $A_{2,1,N}$ derived in \eqref{e:der:A21N}.
In what follows, we use $c>0$ as a generic constant that may have different values from line to line and that does not depend on parameters of the model, i.e., $N,d, \beta,k$, etc.
Armed with these preparations, we are now in a position to prove Lemma \ref{lemma:main_gauss_apprx}. Some auxiliary results are placed in Section \ref{sec_aux_gauss_apprx}, while other, more general background results are presented in Section \ref{appen:aux_res}.

\begin{proof}[Proof of Lemma \ref{lemma:main_gauss_apprx}]
    
    We recall that $A_{1,N}$ is the maximum of absolute values of terms of the form
    \begin{align}\label{eq:terms_A_1}
      w_{\beta}(\ell, k, N) \cdot \Big(\frac{k-\ell}{N} \sum_{n=1}^{N} \varepsilon_{n,j}-\sum_{n=N+\ell+1}^{N+k}\varepsilon_{n,j}   \Big)
      =\frac{N^{1/2}\lb\frac{k-\ell}{N} \sum_{n=1}^{N} \varepsilon_{n,j}-\sum_{n=N+\ell+1}^{N+k}\varepsilon_{n,j}\rb }{(k-\ell)^{\beta} \big( N+k\big)^{1-\beta} }
    \end{align}
    where triplets $(k,\ell,j)$ satisfy $1\le k\le N^{1+\zeta}, 0\le \ell < k, 1\le j\le d$, and $A_{2,1,N}$ is the maximum of absolute values of terms of the form
    \begin{align}\label{eq:terms_A_21}
          N^{-1/2}\sum_{n=1}^N\varepsilon_{n,j},
    \end{align}
    which are indexed by just $j=1,\ldots,d$. In order to unite these elements, we define the set $\mathcal{S}$ as the union of the corresponding indices, i.e., $\mathcal{S}=\{(k,\ell,j)|1\le k\le N^{1+\zeta}, 0\le \ell < k, 1\le j\le d\}\cup\{j^\prime|1\le j^\prime\le d\}$. The cardinality of this set is  $|\mathcal{S}|=d\lfloor N^{1+\zeta}\rfloor (1+\lfloor N^{1+\zeta}\rfloor )/2+d$.    
    For further convenience we define a vector $\gamma_N=(\gamma_{N,1},\ldots, \gamma_{N,|\mathcal{S}|})^\top$ of dimension $|\mathcal{S}|$, which will contain all the elements of form \eqref{eq:terms_A_1} and \eqref{eq:terms_A_21}.  In order to address these elements we let $f:\{1,\ldots,|\mathcal{S}|\}\rightarrow \mathcal{S}$ be an arbitrary bijection. Thus we can rewrite 
    \begin{align}
       \max\{ A_{1,N},  A_{2,1,N}\}\!
        =\max_{m\in\{1,\ldots,|\mathcal S|\}}\{|\gamma_{N,m}|\}=\|\gamma_N\|.
    \end{align}
Next, we divide the coordinates of $\gamma_N$ into two parts:  \begin{itemize}
    \item (Small distance indices) Triplets $(k,\ell,j)$ for which $k-\ell\le(N+k)^{1-\delta}$ and $1\leq j \leq d$; the entire set of named triplets is $\mathcal{S}^-$. We remark that the number of such triplets is bounded by the total number of triplets, i.e.
    \begin{align}\label{eq:bound_S-}
        |\mathcal{S}^-|\leq d\lfloor N^{1+\zeta}\rfloor (1+\lfloor N^{1+\zeta}\rfloor )/2
    \end{align}
    \item (Large distance indices) All remaining indices are gathered in the set $\mathcal{S}^+$, that is, $\mathcal{S}^+=\mathcal{S}\setminus \mathcal{S}^-$. The number of its elements can be bounded with
    \begin{align}\label{eq:bound_S+}
        d\leq |\mathcal{S}^+|\leq |\mathcal{S}|= d\lfloor N^{1+\zeta}\rfloor (1+\lfloor N^{1+\zeta}\rfloor )/2+d
    \end{align}
\end{itemize} 
The rough idea is that the maximum over $\mathcal{S}^-$  will be negligible, while the maximum over the second part will be comparable with the respective Gaussian version. More precisely, let $x_0$ be a fixed, arbitrary small number.
     We denote by $\tilde{\gamma}_N$  the respective Gaussian counterpart of $\gamma_N$, i.e., in which we replaced all $\varepsilon_{n,j}$ by Gaussian $Z_{n,j}$.
Then, Lemma \ref{lemma:main_gauss_apprx} is implied by the following assertions: 
        \begin{align} \label{eq_step1}
        & \sup_{t>x_0} \mathbb{P}\Big(\max_{ f(m)\in\mathcal{\mathcal{S}}^-}|\gamma_{N,f(m)}|\ge t\Big) = o(1), \\
              \label{eq_step1_gauss}
        & \sup_{t>x_0} \mathbb{P}\Big(\max_{ f(m)\in\mathcal{\mathcal{S}}^-}|\tilde \gamma_{N,f(m)}|\ge t\Big) = o(1), \\
        \label{eq_step2}
        &  \sup_{t\ge x_0}\Big|\mathbb{P}\Big(\max_{f(m)\in\mathcal{S}^+}|\gamma_{N,f(m)}|\ge t\Big)-\mathbb{P}\Big(\max_{f(m)\in\mathcal{S}^+}|\tilde\gamma_{N,f(m)}|\ge t\Big)\Big|  =o(1).  
    \end{align}
    Note that \eqref{eq_step1_gauss} is an application of \eqref{eq_step1}. The rest of this section is devoted to the proofs of \eqref{eq_step1} and \eqref{eq_step2}. 
    \end{proof}

\begin{proof}[Proof of \eqref{eq_step1}]
    As a preparatory step, we bound the Orlicz norm by using \eqref{eq:sum_subgauss} and Assumption~\ref{ass3} 
     \begin{align} \label{e:SG:b1}
    \max_{f(m)\in \mathcal{S}^-}\left\| \gamma_{N,f(m)} \right\|_{\Psi_2}^2 \leq &\max_{(k,\ell,j)\in\mathcal{S}^-} \Big\{\frac{C_2^2(k-\ell)^{2(1-\beta)}}{(N+k)^{2(1-\beta)}}+\frac{C_2^2(k-\ell)^{1-2\beta}N}{(N+k)^{2(1-\beta)}}\Big\}\\
      &\le C_2^2\max_{(k,\ell,j)\in\mathcal{S}^-} \Big\{(N+k)^{-2(1-\beta)\delta}+N(N+k)^{-(1+\delta(1-2\beta))}\Big\}\\
      &\le C_2^2\Big(N^{-2(1-\beta)\delta}+N^{-\delta(1-2\beta)}\Big)
      \le2C_2^2N^{-\delta(1-2\beta)}:=K^2. \label{eq1}
    \end{align}
     In the last inequality, we controlled the distance $k-\ell$, using the definition of $\mathcal{S}^-$, and in the third step, we used $N\leq N+k$.
    Notice that the number $K$ depends on $N$, which is not reflected in the definition.
    Then,  the standard concentration bound \eqref{eq:max_subgauss} for a maximum of sub-Gaussian variables gives us  for any $t>0$
    \begin{align} \label{eq2}
        \mathbb{P}\Big(\max_{ f(m)\in\mathcal{S}^-}|\gamma_{N,f(m)}|\ge t\Big)\le |\mathcal{S}^-|\exp\Big(-\frac{t^2}{2cK^2}\Big),
    \end{align}
   where $c>0$ is a constant independent of $N$.

The cardinality of the set $\mathcal{S}^-$ is bounded as in \eqref{eq:bound_S-}, furthermore, Assumption~\ref{ass1}  states that $d$ is bounded by $\exp(N^{C_1})$ up to a constant, and we thus get that 
\begin{align} \label{eq3}
    |\mathcal{S}^-| \le  c \exp(N^{C_1})  N^{1+\zeta} (1+ N^{1+\zeta})/2.
\end{align}
Thus, combining \eqref{eq1}, \eqref{eq2} and \eqref{eq3} with a union bound
    \begin{align} 
        \mathbb{P}\Big(\max_{f(m)\in\mathcal{S}^-}|\gamma_{N,f(m)}|\ge t\Big)
        \le c (1+N^{1+\zeta})N^{1+\zeta}\exp\Big(N^{C_1}-\frac{ct^2N^{\delta(1-2\beta)}}{2C_2^2}\Big).
    \end{align}
    With this we obtained that for any $t> x_{0,N}:= \sqrt{cN^{C_1-\delta(1-2\beta)}}$, we have
        \begin{align} 
        \mathbb{P}\Big(\max_{ f(m)\in\mathcal{\mathcal{S}}^-}|\gamma_{N,f(m)}|\ge t\Big)= o(1).
    \end{align}
    Assumption~\ref{ass1} with the bound on $\delta$ from \eqref{ass:delta_eta} imply that     $C_1<\delta(1-2\beta)$, which assures that $x_{0,N}$ tends to 0 with the growth of $N$. This means that for $N$ big enough we  have $x_{0,N}<x_0$, which completes the proof of \eqref{eq_step1}. 
\end{proof}
To complete the proof of Lemma \ref{lemma:main_gauss_apprx}, it is left to prove assertion \eqref{eq_step2}. 
\begin{proof}[Proof of \eqref{eq_step2}]
    We recall that for  all elements of $\gamma_N$, that correspond to $\mathcal{S}^+$, we have $k\le  N^{1+\zeta}$ and $k-\ell>N^{1-\delta}$. From this follows that any $\gamma_{N,f(m)}$ with $f(m) \in \mathcal{S}^+$ is a sum of no more than $N+N^{1+\zeta}$ terms. Thus, if we denoted for brevity 
    \begin{align} \label{def_N_bar}
        \bar{N}=N+\lfloor N^{1+\zeta} \rfloor,
    \end{align}
    we can rewrite $\gamma_{N,f(m)}$ as the normalized sums of simpler elements $B_{n;f(m)}$, defined as
     \begin{align} \label{e:def:Bf}
         B_{n;f(m)}=\begin{cases}\frac{(k-\ell)^{1-\beta}\bar N^{1/2}}{N^{1/2}(N+k)^{1-\beta}}\varepsilon_{n,j}, & 1\le n \le N \text{ and } f(m)=(k,\ell,j),\\
         -\frac{N^{1/2}\bar N^{1/2}}{(k-\ell)^\beta(N+k)^{1-\beta}}\varepsilon_{n,j}, & N+\ell+1\le n \le N+k  \text{ and } f(m)=(k,\ell,j),\\
         \frac{\bar N^{1/2}}{N^{1/2}}\varepsilon_{n,j}, & 1\le n\le N  \text{ and } f(m)=j,\\
         0 ,& \text{ otherwise}.\end{cases}
     \end{align}
    One can check then 
     \begin{align}\label{eq:gamma_sum}
         \gamma_{N,f(m)}=\frac{1}{\bar{N}^{1/2}}\sum_{n=1}^{\bar{N}}B_{n;f(m)}.
     \end{align}
 Before proceeding, we remark that $B_{n;f(m)}$'s are still sub-Gaussian random variables whose Orlicz norms can be uniformly bounded for all indexes. To show this, we distinguish the cases as in \eqref{e:def:Bf}(except for the zero terms, as they are obviously bounded):
     \begin{align}
         \|B_{n;f(m)}\|_{\Psi_2}
         &\le C_2\begin{cases}\frac{(k-\ell)^{1-\beta}\bar N^{1/2}}{N^{1/2}(N+k)^{1-\beta}}\\
         \frac{N^{1/2}\bar N^{1/2}}{(k-\ell)^\beta(N+k)^{1-\beta}}\\
         \bar N^{1/2}/N^{1/2}
         \end{cases}
               \le  C_2\max\Bigg\{\frac{N(1+N^{\zeta})^{1/2}}{N^{1-\beta\delta}},(1+N^{\zeta})^{1/2}\Bigg\}
         \\
         &\le c\max\{ N^{\zeta/2+\beta\delta},  N^{\zeta/2}\}
         \le c N^{\zeta/2+\beta\delta}:=b_N.
         \label{eq_def_b_N}
     \end{align}
For this derivation we  used in the first inequality Assumption \ref{ass3}. Noting that $(k-\ell)/(N+k)<1$, the first and the third case both can be bounded with $(1+N^{\zeta})^{1/2}$, since $\bar N\le N+N^{1+\zeta}$. The second case we bound by using the latter inequality and  $N^{1-\delta} < k-\ell \le N^{1+\zeta}$. Finally, we use  $1+N^{1+\zeta} \le 2 N^{1+\zeta}$.
     
    For the essential part of the Gaussian approximation we turn to 
    Lemma~\ref{lem:th_chetverikov} (formally we apply it to the vector containing all $\pm B_{n;f(m)}$, for the sake of brevity some of the following calculations are only done for $B_{n;f(m)}$), for which we need to verify conditions~\ref{cond1_var}-\ref{cond3_orlizs}. 
Unfortunately,  condition  \ref{cond2_momen} of Lemma~\ref{lem:th_chetverikov} is not  generally satisfied. As a remedy, we   truncate  each element $B_{n;f(m)}$ at a level $u$, which will be specified below.  More precisely, we set 
\begin{align}
   B_{n;f(m)}^{\le u}&=B_{n;f(m)}1\{|B_{n;f(m)}|\le u\}, \label{e:def:B_le_u}\\
   B_{n;f(m)}^{> u}&=B_{n;f(m)}1\{|B_{n;f(m)}|>u\},
\end{align}
where 
\begin{align} \label{e:def_u}
    u=b_N\sqrt{8c\log(|\mathcal{S}^+|\bar {N)}}.
\end{align}
Note that Lemma \ref{lemma:negl_more_u} in Section \ref{sec_aux_gauss_apprx} says that the maximum over the centered version of  $\sum_nB_{n;f(m)}^{>u}$ is negligible, while we check that they satisfy conditions \ref{cond1_var}-\ref{cond3_orlizs} in Lemma \ref{lem:verf_trunc} of the same section. 

Returning to the initial $\gamma_{N,f(m)}$, since $B_{n;f(m)}$ are centered, we can write
\begin{align}
    B_{n;f(m)}=( B_{n;f(m)}^{\le u}-\E( B_{n;f(m)}^{\le u}))+( B_{n;f(m)}^{>u}-\E( B_{n;{f(m)}}^{>u})).
\end{align}
Using the simple inequality for the maximum of  the difference for two sequences $\{a_i\}_{i=1}^n$ and $\{b_i\}_{i=1}^n$, which states that $|\max_{i}a_i-\max_ib_i|\le\max_i|a_i-b_i|$, and, as consequence, $\max_i|b_i|-\max_i|a_i-b_i|\le\max_ia_i\le\max_ib_i+\max_i|a_i-b_i|$, we can obtain
\begin{align}
    &\max_{f(m)\in\mathcal{S}^+}\Big|\sum_{n=1}^{\bar{N}}B_{n;f(m)}\Big|
    \ge\max_{f(m)\in\mathcal{S}^+}\Big|\sum_{n=1}^{\bar N} B_{n;f(m)}^{\le u}-\E( B_{n;f(m)}^{\le u})\Big|-\max_{f(m)\in\mathcal{S}^+}\Big|\sum_{n=1}^{\bar N} B_{n;f(m)}^{>u}-\E( B_{n;{f(m)}}^{>u})\Big|,\\
&\max_{f(m)\in\mathcal{S}^+}\Big|\sum_{n=1}^{\bar{N}}B_{n;f(m)}\Big|
\le\max_{f(m)\in\mathcal{S}^+}\Big|\sum_{n=1}^{\bar N} B_{n;f(m)}^{\le u}-\E( B_{n;f(m)}^{\le u})\Big|+\max_{f(m)\in\mathcal{S}^+}\Big|\sum_{n=1}^{\bar N} B_{n;f(m)}^{>u}-\E( B_{n;{f(m)}}^{>u})\Big|.\label{eq:max_B_le}
\end{align}
From the first inequality we get
\begin{align}\label{eq:tranc_ge}
   & \PR\lb\max_{f(m)\in\mathcal{S}^+}\bar N^{-1/2}\Big|\sum_{n=1}^{\bar{N}}B_{n;f(m)}\Big|\ge t\rb\\
    &\ge\PR\lb\max_{f(m)\in\mathcal{S}^+}\bar N^{-1/2}\Big|\sum_{n=1}^{\bar N} B_{n;f(m)}^{\le u}-\E( B_{n;f(m)}^{\le u})\Big|\ge t+\max_{f(m)\in\mathcal{S}^+}\bar N^{-1/2}\Big|\sum_{n=1}^{\bar N} B_{n;f(m)}^{>u}-\E( B_{n;{f(m)}}^{>u})\Big|\rb\\
    &\ge\PR\lb\max_{f(m)\in\mathcal{S}^+}\bar N^{-1/2}\Big|\sum_{n=1}^{\bar N} B_{n;f(m)}^{\le u}-\E( B_{n;f(m)}^{\le u})\Big|\ge t+\bar N^{-1}\rb+o(\bar N^{-1}).
\end{align}
The last inequality follows from the law of full probability decomposition with respect to the event \\$\{\max_{f(m)\in\mathcal{S}^+}\bar N^{-1/2}\Big|\sum_{n=1}^{\bar N} B_{n;f(m)}^{>u}-\E( B_{n;{f(m)}}^{>u})\Big|>\bar N^{-1}\}$ and Lemma~\ref{lemma:negl_more_u}.

Let $\tilde B^{\le u}_n$ be the corresponding Gaussian versions of $ B^{\le u}_n$. Then, an application of Lemma~\ref{lem:th_chetverikov} and Lemma \ref{lem:verf_trunc}, where for  $a_M$ we take $u=b_N\sqrt{8c\log(|\mathcal{S}^+|\bar {N)}}$, gives
\begin{align}\label{eq:gaus_bounded_ge2}
   &\Bigg| \PR\lb\max_{f(m)\in\mathcal{S}^+}\bar N^{-1/2}\Big|\sum_{n=1}^{\bar N} B_{n;f(m)}^{\le u}-\E(  B_{n;f(m)}^{\le u})\Big|\ge t+\bar N^{-1}\rb\\
   &-\PR\lb\max_{f(m)\in\mathcal{S}^+}\bar N^{-1/2}\Big|\sum_{n=1}^{\bar N}\tilde B_{n;f(m)}^{\le u}-\E( \tilde B_{n;f(m)}^{\le u})\Big|\ge t+\bar N^{-1}\rb\Bigg|\\
   &\le \lb\frac{cb_N^2\log^{11}(|\mathcal{S}^+|\bar N)}{\bar N}\rb^{1/6}.
\end{align}

To accommodate for  "$+\bar N^{-1}$" we use an anti-concentration inequality as in Lemma~\ref{lem:anti_conc} for the centered  Gaussian vectors $\pm\bar N^{-1/2}\sum_n(\tilde B_{n;\cdot}^{\le u}-\E( \tilde B_{n;\cdot}^{\le u}))$.   Since the terms are independent along the index $n$, the variance of each element of the mentioned Gaussian vector can be  rewritten as the sum of variances and lower bounded by 
\begin{align}\label{eq:var_sum_B}
    &\Var\lb\bar N^{-1/2}\sum_n\tilde B_{n;f(m)}^{\le u}\rb
    =\bar N^{-1}\sum_{n=1}^{\bar N}\Var(\tilde B_{n;f(m)}^{\le u})
    =\bar N^{-1}\sum_{n=1}^{\bar N}\Var( B_{n;f(m)})+O(|\mathcal{S}^+|^{-1})\\
    &\ge\begin{cases}
         \underline{\sigma}^2\lb\frac{(k-\ell)^{2(1-\beta)}}{(N+k)^{2(1-\beta})}+\frac{N(k-\ell)^{1-2\beta}}{(N+k)^{2(1-\beta})}\rb+O(|\mathcal{S}^+|^{-1}), &    f(m)=(k,\ell,j),\\
         \underline{\sigma}^2+O(|\mathcal{S}^+|^{-1}), &  f(m)=j\end{cases}\\
         &\ge\underline{\sigma}^2\lb\frac{(k-\ell)^{2(1-\beta)}}{(N+k)^{2(1-\beta)}}+\frac{N(k-\ell)^{1-2\beta}}{(N+k)^{2(1-\beta)}}\rb+O(|\mathcal{S}^+|^{-1})
         \ge \frac{c}{N^{2(1-\beta)(\delta+\zeta)}} =: v_N^2,
\end{align}
here we used Lemma~\ref{lem:momem_B_u} and the fact that $k-\ell>N^{1-\delta}$ and $k\le N^{1+\zeta}$. Then, from Lemma~\ref{lem:anti_conc} with $a=1/\bar N$ and $b=v_N$, we have

\begin{align}
    & \PR\lb\max_{f(m)\in\mathcal{S}^+}\bar N^{-1/2}\Big|\sum_{n=1}^{\bar N}(\tilde B_{n;f(m)}^{\le u}-\E( \tilde B_{n;f(m)}^{\le u}))\Big|\ge t\rb\\
   & -\PR\lb\max_{f(m)\in\mathcal{S}^+}\bar N^{-1/2}\Big|\sum_{n=1}^{\bar N}(\tilde B_{n;f(m)}^{\le u}-\E( \tilde B_{n;f(m)}^{\le u}))\Big|\ge t+\bar N^{-1}\rb\\
    &\le \frac{1}{\bar Nv_N} (\sqrt{2\log |\mathcal{S}^+|}+2).
\end{align}
 Combining the last  inequality with \eqref{eq:tranc_ge} and \eqref{eq:gaus_bounded_ge2} gives
\begin{align}
   & \PR\lb\max_{f(m)\in\mathcal{S}^+}\bar N^{-1/2}\Big|\sum_{n=1}^{\bar{N}}B_{n;f(m)}\Big|\ge t\rb\\
    &\ge\PR\lb\max_{f(m)\in\mathcal{S}^+}\bar N^{-1/2}\Big|\sum_{n=1}^{\bar N}\tilde B_{n;f(m)}^{\le u}-\E( \tilde B_{n;f(m)}^{\le u})\Big|\ge t\rb-\lb\frac{cb_N^2\log^{11}(|\mathcal{S}^+|\bar N)}{\bar N}\rb^{1/6}\\
    &+o(\bar N^{-1})-\frac{1}{\bar Nv_N} (\sqrt{2\log |\mathcal{S}^+|}+2).
\end{align}
Following the same path but starting with \eqref{eq:max_B_le} we can bound the initial expression from above
\begin{align}
   & \PR\lb\max_{f(m)\in\mathcal{S}^+}\bar N^{-1/2}\Big|\sum_{n=1}^{\bar{N}}B_{n;f(m)}\Big|\ge t\rb\\
    &\le\PR\lb\max_{f(m)\in\mathcal{S}^+}\bar N^{-1/2}\Big|\sum_{n=1}^{\bar N}\tilde B_{n;f(m)}^{\le u}-\E( \tilde B_{n;f(m)}^{\le u})\Big|\ge t\rb+\lb\frac{cb_N^2\log^{11}(|\mathcal{S}^+|\bar N)}{\bar N}\rb^{1/6}\\
    &-o(\bar N^{-1})+\frac{1}{\bar Nv_N} (\sqrt{2\log |\mathcal{S}^+|}+2).
\end{align}
These give
\begin{align}
   & \Bigg|\PR\lb\max_{f(m)\in\mathcal{S}^+}\bar N^{-1/2}\Big|\sum_{n=1}^{\bar{N}}B_{n;f(m)}\Big|\ge t\rb
    -\PR\lb\max_{f(m)\in\mathcal{S}^+}\bar N^{-1/2}\Big|\sum_{n=1}^{\bar N}\tilde B_{n;f(m)}^{\le u}-\E( \tilde B_{n;f(m)}^{\le u})\Big|\ge t\rb\Bigg|\\
    &\le\lb\frac{cb_N^2\log^{11}(|\mathcal{S}^+|\bar N)}{\bar N}\rb^{1/6}
    +o(\bar N^{-1})+\frac{1}{\bar Nv_N} (\sqrt{2\log |\mathcal{S}^+|}+2).
\end{align}
Bounds \eqref{ass:delta_eta} together with Assumption~\ref{ass1} make sure that obtained quantities tend to 0. 

At this point we obtained a preliminary version of the final result, but  the variance profile of the involved Gaussian variables coincides  with the one of truncated version of $B_{n;f(m)}$. So the last step is to compare $\max_{f(m)\in\mathcal{S}^+}\bar N^{-1/2}\Big|\sum_{n=1}^{\bar N}(\tilde B_{n;f(m)}^{\le u}-\E( \tilde B_{n;f(m)}^{\le u}))\Big|$ with  $\max_{f(m)\in\mathcal{S}^+}\bar N^{-1/2}\Big|\sum_{n=1}^{\bar N}(\tilde B_{n;f(m)}-\E( \tilde B_{n;f(m)}))\Big|$, where  $\tilde B_{n;f(m)}$ is the Gaussian version of $ B_{n;f(m)}$ in \eqref{e:def:Bf}.
 For this  we use Lemma~\ref{lem:compar_gaus}. To  calculate the corresponding $\Delta$, we take arbitrary indexes $m,r$ such that $f(m),f(r)\in\mathcal{S}^+$, and write
\begin{align}   
\label{e:truncated-covariance-bound}
    &\Bigg|\cov\lb\bar N^{-1/2}\sum_{n=1}^{\bar N}\tilde B_{n;f(m)}^{\le u},\bar N^{-1/2}\sum_{n'=1}^{\bar N}\tilde B_{n';f(r)}^{\le u}\rb-\cov\lb\bar N^{-1/2}\sum_{n=1}^{\bar N}\tilde B_{n;f(m)},\bar N^{-1/2}\sum_{n'=1}^{\bar N}\tilde B_{n';f(r)}\rb\Bigg|\\
    &\le \bar N^{-1}\sum_{n=1}^{\bar N}\Bigg|\cov\lb \tilde B_{n;f(m)}^{\le u},\tilde B_{n;f(r)}^{\le u}\rb-\cov\lb\tilde B_{n;f(m)},\tilde B_{n;f(r)}\rb\Bigg|\\
    &=\bar N^{-1}\sum_{n=1}^{\bar N}\Bigg|\cov\lb  B_{n;f(m)}^{\le u}, B_{n;f(r)}^{\le u}\rb-\cov\lb B_{n;f(m)}, B_{n;f(r)}\rb\Bigg|\\
    &\le \max_{1 \leq n \leq \bar N}\Bigg|\cov\lb  B_{1;f(m)}^{\le u}, B_{1;f(r)}^{\le u}\rb-\cov\lb B_{1;f(m)}, B_{1;f(r)}\rb\Bigg|.
\end{align}
Due to Lemma~\ref{lem:momem_B_u} we know that the latter expression is of order $  \mathcal{O}\lb(|\mathcal{S}^+|\bar N)^{-1}\rb$, which means that we can take $\Delta=c(|\mathcal{S}^+|\bar N)^{-1}$. To find the value of $\sigma$, i.e. lower bound of $\Var(\bar N^{-1/2}\sum_{n=1}^{\bar{N}}\tilde B_{n;f(m)})$, we perform a very similar calculation to \eqref{eq:var_sum_B} and get that for any  $f(m)\in\mathcal{S}^+ $ $\Var(\bar N^{-1/2}\sum_{n=1}^{\bar{N}}\tilde B_{n;f(m)})\ge cN^{-2(1-\beta)(\delta+\zeta)}$. This immediately gives us
\begin{align}
   & \Bigg|\PR\lb\max_{f(m)\in\mathcal{S}^+}\bar N^{-1/2}\Big|\sum_{n=1}^{\bar N}(\tilde B_{n;f(m)}^{\le u}-\E( \tilde B_{n;f(m)}^{\le u}))\Big|\ge t\rb-\PR\lb\max_{f(m)\in\mathcal{S}^+}\bar N^{-1/2}\Big|\sum_{n=1}^{\bar{N}}\tilde B_{n;f(m)}\Big|\ge t\rb
    \Bigg|\\
    &\le \frac{cN^{(1-\beta)(\delta+\zeta)}\log |\mathcal{S}^+|\log^{1/3}\bar N}{(|\mathcal{S}^+|\bar N)^{1/3}}.
\end{align}
Since $N^{1+\delta}d\leq|\mathcal{S}^+|\leq 3dN^{2+2\zeta}$, and (by Assumption~\ref{ass1}) $d$ is upper bounded by a constant times $\exp(N^{C_1})$, we obtain convergence to 0 of the right hand side.
Recalling \eqref{eq:gamma_sum}, the proof of \eqref{eq_step2} is completed. 

\end{proof}

\subsubsection{Auxiliary results for the proof of Lemma \ref{lemma:main_gauss_apprx}} \label{sec_aux_gauss_apprx}

The first lemma ensure the validity of our truncation argument in the proof of Lemma \ref{lemma:main_gauss_apprx}. Recall the definition of $u=b_N\sqrt{8c\log(|\mathcal{S}^+|\bar {N)}}$ in \eqref{e:def_u}.

\begin{lemma}\label{lemma:negl_more_u}
  Suppose that the assumptions of Theorem \ref{thm:main} are satisfied.
Then, we have
    \begin{align}
        \PR\lb\bar{N}^{-1/2}\max_{f(m)\in\mathcal{S}^+}\Bigg|\sum_{n=1}^{\bar{N}}\lb B_{n;f(m)}^{>u}-\E(B_{n;f(m)}^{>u})\rb\Bigg|>\bar{N}^{-1}\rb=o(\bar{N}^{-1}).
    \end{align}
\end{lemma}
\begin{proof}[Proof of Lemma \ref{lemma:negl_more_u}]
First we check that with $u$ as above, the maximum of $\sum_nB_{n;f(m)}^{>u}$ is with high probability equal to $0$.
To this end, we write 
    \begin{align}\label{eq:prob_A}
    \PR(A):=\PR\lb\max_{f(m)\in\mathcal{S}^+}\Bigg|\sum_{n=1}^{\bar{N}}B_{n;f
    (m)}^{>u}\Bigg|\ne 0\rb
    \le |\mathcal{S}^+|\bar{N}\max_{f(m)\in\mathcal{S}^{+}}\mathbb{P}(B_{n;f(m)}^{>u}\ne 0)\\\nonumber
    =|\mathcal{S}^+|\bar{N}\max_{f(m)\in\mathcal{S}^{+}}\mathbb{P}(|B_{n;f(m)}|>u)
    \le 2|\mathcal{S}^+|\bar{N}\exp\lb-\frac{u^2}{2cb_N^2 }\rb,
\end{align}
where $b_N$ was defined in \eqref{eq_def_b_N} and  the last inequality is the standard sub-Gaussian bound for a tail  \eqref{eq:tail_subgauss}. By plugging into the latter expression the definition of $u$ given in the lemma, we obtain
\begin{align}
\label{eq_prob_bound_A}
\PR (A) 
    \le\frac{2}{(|\mathcal{S}^+|\bar{N})^3}.
\end{align}
For future use we explicitly write the bound of $\mathbb{P}(|B_{n;f(m)}|>u) $
\begin{align}\label{eq:ineq_prob_B}
    \mathbb{P}(|B_{n;f(m)}|>u) \le\frac{2}{(|\mathcal{S}^+|\bar{N})^4}.
\end{align}
The expectation of $B_{n;f(m)}^{>u}$ can be bounded using Hölder's inequality and the just used sub-Gaussian tail bound, which leads to
\begin{align}
    \Big|\E(B_{n;f(m)}^{>u})\Big|\le \sqrt{\Var(B_{n;f(m)})\PR(|B_{n;f(m)}|>u)}\le\sqrt{\frac{2\Var(B_{n;f(m)})}{(|\mathcal{S}^+|\bar{N})^{4}}}.
    \end{align}
One can check that $\Var(B_{n;f(m)})\le c\bar{N}$, so we can bound the total expression by   
\begin{align}\label{eq:bound_exp_B}
    \Big|\E(B_{n;f(m)}^{>u})\Big|\le \frac{c}{(|\mathcal{S}^+|\bar{N})^{3/2}}. 
    \end{align}
    Then we rewrite  
    \begin{align}
        &\PR\lb\max_{f(m)\in\mathcal{S}^+}\Big|\sum_{n=1}^{\bar{N}}\lb B_{n;f(m)}^{>u}-\E(B_{n;f(m)}^{>u})\rb\Big|>\bar{N}^{-1/2}\rb\\
       & =    \PR\lb\max_{f(m)\in\mathcal{S}^+}\Big|\sum_{n=1}^{\bar{N}}\lb B_{n;f(m)}^{>u}-\E(B_{n;f(m)}^{>u})\rb\Big|>\bar{N}^{-1/2}\Bigg|A\rb\PR(A)\\
       &+\PR\lb\max_{f(m)\in\mathcal{S}^+}\Big|\sum_{n=1}^{\bar{N}}\lb B_{n;f(m)}^{>u}-\E(B_{n;f(m)}^{>u})\rb\Big|>\bar{N}^{-1/2}\Bigg|A^c\rb\PR(A^c)
       \\
       &\le\frac{2}{(|\mathcal{S}^+|\bar{N})^3}+\PR\lb\max_{f(m)\in\mathcal{S}^+}\Big|\sum_{n=1}^{\bar{N}} \E(B_{n;f(m)}^{>u})\Big|>\bar{N}^{-1/2}\rb\\
       &= \frac{2}{(|\mathcal{S}^+|\bar{N})^3}+1\{\max_{f(m)\in\mathcal{S^+}}\Big|\sum_{n=1}^{\bar{N}}\E(B_{n;f(m)}^{>u})\Big|>\bar{N}^{-1/2}\}\\
       &\le\frac{2}{(|\mathcal{S}^+|\bar{N})^3}+1\Big\{\bar{N}\max_{f(m)\in\mathcal{S^+}}\Big|\E(B_{n;f(m)}^{>u})\Big|>\bar{N}^{-1/2}\Big\}
       =\frac{2}{(|\mathcal{S}^+|\bar{N})^3}
    \end{align}
   The first inequality is true due to \eqref{eq_prob_bound_A} and the fact that on the event 
   \[
        A^c=\Bigg\{\max_{f(m)\in\mathcal{S}^+}\Big|\sum_{n=1}^{\bar{N}}B_{n;f
        (m)}^{>u}\Big|= 0\Bigg\},
   \] 
   we have $\sum_{n=1}^{\bar{N}}B_{n;f
    (m)}^{>u}= 0$ for all $f(m)\in\mathcal{S^+}$. 
    For the last inequality we used  \eqref{eq:bound_exp_B}. This concludes the proof of Lemma~\ref{lemma:negl_more_u}.
\end{proof}
The following simple lemma describes the relationship between the truncated random variables $B_{n;f(m)}^{\le u}$ and their original random variables $B_{n;f(m)}$ in terms of their  covariances. Recall the definition of $\bar N$ in \eqref{def_N_bar}. 
\begin{lemma}\label{lem:momem_B_u} 
Suppose that the assumptions of Theorem \ref{thm:main} are satisfied. 
For $u$ defined in \eqref{e:def_u} and $f(m),f(r)\in\mathcal{S}^+$, we have
    \begin{align}
        \cov(B_{n;f(m)}^{\le u} , B_{n;f(r)}^{\le u})=\cov(B_{n;f(m)}, B_{n;f(r)})+\mathcal{O}\lb(|\mathcal{S}^+|\bar{N})^{-1}\rb,
    \end{align}
    where the implied constants do not depend on $n$.
\end{lemma}

\begin{proof}
    We  write
    \begin{align*}
        &\cov(B_{n;f(m)}^{\le u},B_{n;f(r)}^{\le u})=\E \left[ \lb B_{n;f(m)}- (B_{n;f(m)}^{> u}-\E B_{n;f(m)}^{> u})\rb\lb B_{n;f(r)}- (B_{n;f(r)}^{> u}-\E B_{n;f(r)}^{> u})\rb \right] \\
       & =\cov(B_{n;f(m)},B_{n;f(r)})-\E(B_{n;f(r)}B_{n;f(m)}^{> u})-\E(B_{n;f(m)}B_{n;f(r)}^{> u})\\
       &+\E(B_{n;f(m)}^{> u}B_{n;f(r)}^{> u})-\E(B_{n;f(m)}^{> u})\E(B_{n;f(r)}^{> u})
    \end{align*}
  The last term with first moments can be bounded using   \eqref{eq:bound_exp_B}, terms with  second mixed moments are  bounded similarly, for example
    \begin{align}
    & \E(B_{n;f(m)}B_{n;f(r)}^{> u})\le \sqrt{\E(B_{n;f(m)}^2B_{n;f(r)}^2)\PR(|B_{n;f(r)}|>u)}\\
    \le & \frac{c\lb\E(B_{n;f(m)}^4)\E(B_{n;f(r)}^4)\rb^{1/4}}{(|\mathcal{S}^+|\bar{N})^{2}}\le\frac{c}{|\mathcal{S}^+|\bar{N}},
    \end{align}
   where we bound $\mathbb{P}(|B_{n;f(m)}|>u)$ with \eqref{eq:ineq_prob_B} and  taking into the account that the fourth moment of sub-Gaussian random variables is bounded by the 4th power of $\|\cdot\|_{\Psi_2}$-norm.
\end{proof}
Next, let us check that the centered truncated variables $B_{n;f(m)}^{\le u}-\E(B_{n;f(m)}^{\le u})$ defined in \eqref{e:def:B_le_u} satisfy conditions \ref{cond1_var}-\ref{cond3_orlizs}. Recall the definition of $u$ in \eqref{e:def_u}.

\begin{lemma}\label{lem:verf_trunc}
Suppose that the assumptions of Theorem \ref{thm:main} are satisfied. Then,
the random variables $B_{n;f(m)}^{\le u}-\E(B_{n;f(m)}^{\le u})$, $f(m) \in \mathcal{S}^+$,  satisfy conditions \ref{cond1_var}-\ref{cond3_orlizs} from Lemma \ref{lem:th_chetverikov} with the choice $a_M= u$.
\end{lemma}
\begin{proof}[Proof of Lemma \ref{lem:verf_trunc}]
\textit{Verification of \ref{cond1_var}} It is easy to check that condition~\ref{cond1_var} of Lemma~\ref{lem:th_chetverikov} is satisfied for many elements that are in the third line of the distinction by cases in the definition of $B_{n;f(m)}$ in  \eqref{e:def:Bf},  i.e. for $(\bar N/N)^{1/2}\varepsilon_{n,j}$, according to Assumption \ref{ass:main}, part (A-4). 
Indeed, let us take $m$ such that $f(m)=j$, for some $j\in\{1,\ldots,d\}$, then using Lemma~\ref{lem:momem_B_u} we get 
\begin{align}
    \frac{1}{\bar N}\sum_{n=1}^{\bar N}\Var (B_{n;f(m)}^{\le u})
   &= \frac{1}{\bar N}\sum_{n=1}^{\bar N}\E ( B_{n;f(m)}^2)+\frac{1}{\bar N}\sum_{n=1}^{\bar N}\mathcal{O}\lb(|\mathcal{S}^+|\bar{N})^{-1}\rb\\
   &\geq\E(\varepsilon_{1,j}^2)>c,
\end{align}
since $\E ( B_{n;f(m)}^2)=\bar N N^{-1}\E[\varepsilon_{n,j}^2]$ for $n \leq N$. Such independent constant $c$, which will play the role of $\sigma^2$ in \ref{cond1_var}, exists due to \ref{ass4}. This is enough to satisfy \ref{cond1_var}.

\textit{Verification of \ref{cond2_momen}} Using again Lemma~\ref{lem:momem_B_u} we write for $r=1,2$
\begin{align} 
    &\frac{1}{\bar N}\sum_{n=1}^{\bar N}\E(|B_{n;f(m)}^{\le u}-\E( B_{n;f(m)}^{\le u})|^{2+r})\le  \frac{2^ru^r}{\bar N}\sum_{n=1}^{\bar N}\Var(B_{n;f(m)}^{\le u})\\
    &\le 2^ru^r\begin{cases}
        \bar\sigma^2\lb\frac{(k-\ell)^{2(1-\beta)}}{(N+k)^{2(1-\beta)}}+\frac{N(k-\ell)^{1-2\beta}}{(N+k)^{2(1-\beta)}}+o(1)\rb,\; & f(m)=(k,\ell,j)\\
        \bar\sigma^2(1+o(1)),\; &f(m)=j
    \end{cases} \\ & \le c\bar\sigma^2u^r.
\end{align}
At the last step we bounded ratios by 1, since $k-\ell\le N+k$ and $N\le N+k$.  
If we choose $a_M$ from the \ref{cond2_momen} and as $\max\{u,b_N\}=u$, we see that \ref{cond2_momen} holds. 

\textit{Verification of \ref{cond3_orlizs}} Here we want to show that for the constant $b_N$ from \eqref{eq_def_b_N} the  condition ~\ref{cond3_orlizs} is still satisfied, i.e. $\E(\exp(|B_{n;f(m)}^{\le u}-\E(B_{n;f(m)}^{\le u})|/b_N))\le 2$ for all $n=1,\ldots,\bar N$ and $f(m)\in\mathcal{S}^+$. To be precise we will show it for $2b_N$ instead of $b_N$, but this only changes the resulting bounds by a multiplicative factor. To that end note that by monotonicity we have
\[
    \|B_{n;f(m)}^{\le u}\|_{\Psi_1}\leq \|B_{n;f(m)}\|_{\Psi_1}\leq b_N.
\]
Additionally, using that the expected value can be expressed as a conditional expectation and the fact that the conditional expectation is a contraction for Orlicz norms, we also have
\[
    \|\E[B_{n;f(m)}^{\le u}]\|_{\Psi_1}\leq \|B_{n;f(m)}\|_{\Psi_1}\leq b_N.
\]
Using the triangle inequality yields the desired bound.

     \end{proof}

\subsection{Proof of Lemma \ref{lem1} } \label{sec_proof_lem1}

\begin{proof}[Proof of Lemma \ref{lem1}]
    To prove the lemma, it suffices to show that (by the triangle inequality)
    \begin{align} \label{goal_A22}
    A_{2,2,N}:=\sup_{k \ge N^{1+\zeta}}\max_{0 \leq \ell < k} \bigg\{w_{\beta}(\ell, k, N) \cdot \Big\|\sum_{n=N+\ell+1}^{N+k}\varepsilon_{n} \Big\| \bigg\}=o_P(1).
    \end{align}
    Consider some (arbitrarily small) $q>0$. Then, we aim to show that the following probability converges to 0:
\begin{align}
    & \mathbb{P}\big(A_{2,2,N}>q\big) \nonumber \\
    \overset{max- norm}{=} & \mathbb{P}\bigg(\max_{j=1,...,d}\sup_{k \ge N^{1+\zeta}}\max_{0 \leq \ell < k} \bigg\{w_{\beta}(\ell, k, N) \cdot \Big|\sum_{n=N+\ell+1}^{N+k}\varepsilon_{n,j} \Big| \bigg\}>q\bigg) \nonumber \\
    \overset{union\, bound}{\le} & d \max_j \mathbb{P}\bigg(\sup_{k \ge N^{1+\zeta}}\max_{0 \leq \ell < k} \bigg\{w_{\beta}(\ell, k, N) \cdot \Big|\sum_{n=\ell+1}^{k}\varepsilon_{n,j} \Big| \bigg\}>q\bigg).
    \label{eq_union}
\end{align}

We begin by noting that Lemma \ref{lem:psum_conc} implies the existence of random variables $M_1,...,M_d$ that satisfy 
\begin{align}
    \label{eq:moc:bound}
\Big|\sum_{n=\ell +1}^{k}\varepsilon_{n,j}\Big|&\leq M_j\sqrt{(k-\ell)\log(e(k \lor N))}, \quad \ell,k\geq 1,\\
    \label{eq:moc:subgaussian}
    \|M_j\|_{\Psi_2}&\leq c, \quad 1 \leq j \leq d,
\end{align}
for some $c>0$ that depends only on $C_2$. Moreover, we note that, using $\beta \in [0,1/2)$ and the definition of $w_\beta$ in  \eqref{eq_def_omega},
\begin{align} \label{a3}
    w_\beta(\ell,k,N)\sqrt{k-\ell}=\frac{\sqrt{N}(k-\ell)^{1/2-\beta}}{(N+k)^{1-\beta}}\leq \sqrt{\frac{N}{k}},
\end{align}
where we bounded $k-\ell\leq k$ and $N
+k\geq k$.
Consequently, combining \eqref{eq:moc:bound} and \eqref{a3}, we may bound
\begin{align*}
    \sup_{k \geq N^{1+\zeta}}\max_{1 \leq \ell \leq k}w_\beta(\ell,k,N)\Big|\sum_{n=\ell}^k\varepsilon_{n,j}\Big|\leq  \max_{k \geq N^{1+\zeta}}M_j\frac{\sqrt{N\log(ek)}}{\sqrt{k}}\leq DM_jN^{-\zeta'}
\end{align*}
for any $\zeta'<\zeta/2$, where $D>0$ is a constant that depends on $\zeta'$. Using \eqref{eq:moc:subgaussian}, Lemma \ref{lem:lem_sub_gaussian} and Assumption \ref{ass1} then yields
\begin{align}
    d \max_j \mathbb{P}\bigg(\sup_{k \ge N^{1+\zeta}}\max_{0 \leq \ell < k} \bigg\{w_{\beta}(\ell, k, N) \cdot \Big|\sum_{n=\ell+1}^{k}\varepsilon_{n,j} \Big| \bigg\}>q\bigg) &\leq 2d\exp\Big(-\frac{D^2q^2N^{2\zeta'}}{c^2}\Big)\\
    &\leq 2\exp\Big(-\frac{D^2q^2N^{2\zeta'}}{c^2}+N^{C_1}\Big) \\
    & = o(1), \label{a2}
\end{align}
where we used that we can choose $\zeta'$ such that $\zeta'>C_1/2$. Thus, combining \eqref{eq_union} and \eqref{a2} yields the desired result \eqref{goal_A22}.

\end{proof}

\section{Proofs of Theorem \ref{thm:alt}, Corollaries \ref{cor:main:level}, \ref{cor:alt1}, \ref{cor:alt2}} \label{sec_proof_alt}
\label{appen:B}
\begin{proof}[Proof of Corollary \ref{cor:main:level}]
The proof of the Corollary follows directly from the Gaussian approximation in Theorem \ref{thm:main} and the fact that convex functionals of Gaussian processes have continuous distribution functions (see Corollary 4.42 in \cite{Bogachev:1998}).
\end{proof}
 
\begin{proof}[Proof of Theorem \ref{thm:alt}]
First, we notice that, if there is a change in component $j$, we have the lower bound
\begin{align*}
    \widehat{\Gamma}_N^{\beta}(j,k)= & \max_{0 \leq \ell < k} \bigg\{w_{\beta}(\ell, k, N) \cdot \Big|\frac{k-\ell}{N} \sum_{n=1}^{N} X_{n,j}-\sum_{n=N+\ell+1}^{N+k}X_{n,j} \Big| \bigg\}\\
    \ge & w_\beta(k_j^\star,k,N)\Big|\frac{k-k_j^\star}{N} \sum_{n=1}^{N} \varepsilon_{n,j}-\sum_{n=N+k_j^\star+1}^{N+k}\varepsilon_{n,j} + (k-k_j^\star) (\mu^{(1)}_j-\mu^{(2)}_j)\Big| \\
    \ge & w_\beta(k_j^\star,k,N) |(k-k_j^\star) (\mu^{(1)}_j-\mu^{(2)}_j)|\\
    & - w_\beta(k_j^\star,k,N)\Big|\frac{k-k_j^\star}{N} \sum_{n=1}^{N} \varepsilon_{n,j}-\sum_{n=N+k_j^\star+1}^{N+k}\varepsilon_{n,j} \Big|\\
    \ge &  w_\beta(k_j^\star,k,N) (k-k_j^\star) \Delta_j - {O}_P\Big(\sup_{k\ge 1}\widetilde{\Gamma}_N^{\beta}(k)\Big).
\end{align*}
Focusing on the deterministic part, we get that  
\[
 w_\beta(k_j^\star,k,N) (k-k_j^\star) \Delta_j = \bigg(\frac{k-k^\star_j}{N+k}\bigg)^{1-\beta}(\Delta_j\sqrt{N}).
\]
This shows Theorem \ref{thm:alt}. 
\end{proof} 
\begin{proof}[Proof of Corollary \ref{cor:alt1}] We start with the first part. Choosing $k\gg N$, the above representation shows that the deterministic part is $\sim (\Delta_j\sqrt{N})$. If now 
\[
\lim_{N \to \infty} \frac{\Delta_j \sqrt{N}}{\sqrt{\log(ed)}} = \infty
\]
as posited in the corollary, then it follows from above that
\[
\sup_k\widehat{\Gamma}_N^{\beta}(k) -q_{1-\alpha}^{(N)}\ge (\Delta_j\sqrt{N}) - \mathcal{O}_P\big(\sqrt{\log(ed)}\big),
\]
where we have also used that $q_{1-\alpha}^{(N)} = \mathcal{O}(\sqrt{\log(ed)})$. The right side tends to infinity and thus the first part of Corollary \ref{cor:alt1} is proved. The second part follows by similar computations and is therefore omitted.
\end{proof}
Finally, we come to the proof of Corollary \ref{cor:alt2}. 
\begin{proof}[Proof of Corollary \ref{cor:alt2}]
We know that rejection occurs asymptotically in component $j$ with probability going to $1$.  
Rejection occurs with probability (asymptotically) $\ge 1-\epsilon$, if for some constant $C_\epsilon>0$ and some $k>k^\star_j$ it holds that
\[
\bigg(\frac{k-k^\star_j}{N+k}\bigg)^{1-\beta}(\Delta_j\sqrt{N})\ge C_\epsilon \sqrt{\log(ed)}.
\]
Rearranging terms, we get this if  $k-k^\star_j$ satisfies 
\[
k-k^\star_j \ge \bigg(\frac{C_\epsilon^2\log(ed)}{\Delta_j^2 N}\bigg)^{1/[2(1-\beta)]}(N+k). 
\]
Now, let us without loss of generality assume that $k_j^\star \asymp  N$. Then, for $k_j^\star \le k \le C k_j^\star$ we can simplify the above relation to 
\[
\bigg(\frac{\log(ed)}{\Delta_j^2 N}\bigg)^{1/[2(1-\beta)]}N =\mathcal{O}( k-k^\star_j).
\]
This implies that the delay time satisfies
\[
\tau_j = \mathcal{O}_P \bigg( \bigg(\frac{\log(ed)}{\Delta_j^2 N}\bigg)^{1/[2(1-\beta)]}N\bigg) = \mathcal{O}_P \bigg( \bigg(\frac{\log(ed)N^{1-2\beta}}{\Delta_j^2 }\bigg)^{1/[2(1-\beta)]}\bigg).
\]
This completes the proof.
\end{proof}

\section{Proofs for Segmentation}
\label{appen:segment}
The proof of Theorem \ref{thm:seg:gauss} works exactly along the same lines as that of Theorem \ref{thm:main} and is therefore omitted. In Appendix~\ref{app:log} we prove the analogue of the Theorem~\ref{thm:seg:gauss} for the statistic  $\gamma^{\log}(n,h)$ with a logarithmic weight defined in \eqref{eq:def:segmentation:scan:2}, while in Appendix~\ref{app:poly} we consider the statistic $\gamma(n,h)$ introduced in \eqref{eq:def:segmentation:scan}.

In order to formally introduce  dependence in our model, we consider so called $\beta$-mixing coefficients. For this, we start with a probability space $(\Omega, \mathcal{F},\mathbb{P})$, then, for any two $\sigma-$algebras $\mathcal{F}_1, \mathcal{F}_2 \subset \mathcal{F}$ we define 
\[
\beta(\mathcal F_1,\mathcal F_2)
=
\frac{1}{2}
\sup_{\substack{(A_i)\text{ partition of }\Omega,\; A_i\in\mathcal F_1\\
(B_j)\text{ partition of }\Omega,\; B_j\in\mathcal F_2}}
\sum_{i,j}
\left|
\mathbb P(A_i\cap B_j)-\mathbb P(A_i)\mathbb P(B_j)
\right|.
\]
For a general collection of random variables $(X_i)_{i \in \mathbb Z}$, the mixing at distance $k\in\mathbb{N}$ is defined by
\begin{align} \label{e:defbet}
\beta_{X_{1},...,X_{N}}(k)=\sup_{l \in \mathbb Z}\beta(\mathcal{F}_1^l,\mathcal{F}_{l+k}^N)~,
\end{align} 
where
$\mathcal{F}_i^j=\sigma(X_{i},...,X_{j})$. However, since in our model we have only $(X_i)_{i=1}^{N}$, 
we will typically use $\beta(k):=\beta_{X_{1},...,X_{N}}(k)=\max_{\ell \in \{1,\ldots,N-k\}}\beta(\mathcal{F}_1^\ell,\mathcal{F}_{\ell+k}^N)$ instead. If formally necessary to talk about indices of an entire time series, one may always complete the segment $X_1,..,X_N$ to a time series indexed in $\mathbb Z$ by adding further independent copies of $X_1$. In this manner, both definitions are equivalent. 

Further, let $q=q_N>r=r_N$ be two sequences of positive integers jointly satisfying 
\begin{align} \label{e:qr}
    2(q+r)<N.
\end{align}
Using $q$ and $r$, we define  index sets
\begin{align}
    &I_1=\{1,...,q\},\\
    &J_1=\{q+1,...,q+r\},\\
    &\quad \quad \quad \vdots \\
    &I_m=\{(m-1)(q+r)+1,...,(m-1)(q+r)+q\},\\
    &J_m=\{(m-1)(q+r)+q+1,...,m(q+r)\},\\
    &J_{m+1} = \{m(q+r)+1,...,N\},
\end{align} 
where we put $m=m_N=\lfloor N/(q+r)\rfloor$. In other words, we decompose the sequence $\{1,...,N\}$ into large blocks $I_1,...,I_m$ of length $q$ separated by small blocks $J_1,...,J_{m}$ of length $r$ and there exists (potentially) an "overhanging remainder" of terms called $J_{m+1}$.  We then define
\begin{align}
    \bar \sigma_q:=\max_{1 \leq j \leq d }\frac{1}{m}\sum_{I \in \{I_1,...,I_m\}}\text{Var}\left(q^{-1/2}\sum_{i \in I}\epsilon_{ij}\right) \\
    \bar{\sigma}_r:=\max_{1 \leq j \leq d }\frac{1}{m}\sum_{J \in \{J_1,...,J_m\}}\text{Var}\left(r^{-1/2}\sum_{i \in J}\epsilon_{ij}\right) 
\end{align}
and additionally set $\bar \sigma_0:=1$.

\subsection{ Logarithmic weighting}\label{app:log}

For  convenience we recall the statistic with the logarithmic weight, introduced in Section~\ref{sec:opt:dep}. To avoid confusion later in the section, we will denote it by $\gamma^{\log}(n,h)$
\begin{align*}
    \gamma_j^{\log}(n,h)=\frac{\sum_{m=n-h+1}^{n}X_{m,j}-\sum_{m=n+1}^{n+h}X_{m,j}}{h^{1/2}\log^\beta(eN/h)},\qquad j=1,\ldots,d,\; \text{ and }(n,h)\in\mathcal{S},
\end{align*}
where $\mathcal{S}$ satisfies \eqref{eq_def_mathcal_S}.

Before formulating the main result of the section, we list the assumptions on  $X_1,...,X_N$ and parameters of the statistic. 
\begin{ass} \label{ass:segm}
    \begin{enumerate}[label=(A-C\arabic*)] $ $
           \item \label{ass_dim_seg} There exists a positive constant $\gamma$ such that
        \[
        d\leq cN^{\gamma}.
        \]
        \item \label{ass_var_seg}
        \[
         \max_{1 \leq i \leq N, 1 \leq j \leq d}\|\epsilon_{i,j}\|_{\Psi_2}\leq K.
        \]
        \item \label{ass_mixing} For $\forall q,r$ we have
    \begin{align*} 
    \bar \sigma_q &\geq \underline{c} \\
       \beta(k)&\leq C \exp(-ck^a), \quad k \in \N   . 
    \end{align*}
    Note that the dependence on $r$ is implicit in the definition of $\bar \sigma_q$.
   
        \item \label{ass_par_beta_a}
The parameters $\beta$ and $a$ are such that
\begin{align*}
    \beta>1/2,\qquad 2\beta-1/a\geq 1.
\end{align*}
    \end{enumerate}
\end{ass}
\begin{theo} 
\label{theo1:dep}
    Let  $X_1,...,X_N$ have mean zero (i.e. $X_i=\epsilon_i)$) and satisfy Assumptions \ref{ass:segm} for some fixed $a, \beta, c, K, C,\gamma, \underline{c}>0$. We then have that there exists a Gaussian random variable $G \in \R^{d|\mathcal{S}|}$, indexed in $(j,n,h) \in  \{1,...,d\} \times \mathcal S$, such that
    \begin{align}
        \label{e:Hold:Gau:Log}
        \sup_{t \ge x_0} \big|\PR(\sup_{1 \leq j \leq d, (n,h) \in \mathcal{S}}|G_{j,n,h}|\le t)- \mathbb{P}(\sup_{1 \leq j \leq d, (n,h) \in \mathcal{S}}|\gamma^{\log}_{j}(n,h)| \le t)\big|=o(1)   
    \end{align}
    for some $x_0>0$ that depends only on $a,c, K, C,\gamma, \underline{c}>0$. The covariance of this Gaussian process is characterized by
     \begin{align}
        \text{Cov}(G_{j_1,n_1,h_1},G_{j_2,n_2,h_2})=\text{Cov}( \gamma^{\log}_{j_1}(n_1,h_1),  \gamma^{\log}_{j_2}(n_2,h_2))
    \end{align}
\end{theo}

\begin{proof}
    
In general, the proof has the same structure as for Theorem~\ref{thm:main}, however, because of the dependence in the data, some extra steps are required. The main part is to use the Gaussian approximation for dependent data, which is summarized in Appendix~\ref{Thm:GaussApprox}. For this, similar to \eqref{e:def:Bf}, we  rewrite $\gamma^{\log}$ as the normalized sum of simple elements $D_1,...,D_N \in \R^{d|\mathcal S|}$, indexed in $(j,n,h)\in \{1,...,d\}\times \mathcal S$, with components  
\begin{align}\label{def_D}
    D_{i,j,n,h}=
    \begin{cases}
        0    & \quad \begin{array}{c} 1 \leq i < n-h+1 \end{array}, \\
        \epsilon_{ij}/\Big((h/N)^{1/2}\log^\beta(eN/h)\Big) & \quad \begin{array}{c} n-h+1 \leq i \leq n \end{array}, \\
        -\epsilon_{ij}/\Big((h/N)^{1/2}\log^\beta(eN/h)\Big)& \quad \begin{array}{c} n < i \leq n+h \end{array}, \\
        0    & \quad \begin{array}{c} n+h < i \leq N \end{array}.
    \end{cases}
\end{align}
indexed in $1 \leq j \leq d$ and  $(n,h) \in \mathcal{S}$. The dimension of $D_i$'s we denote by $\tilde d=d|\mathcal{S}|\leq dN^2$. Consequently $\log(\tilde d)\leq  C$ under Assumption~\ref{ass_dim_seg}. It is easy to see that
\begin{align} \label{eq_gamma_D_ij}
    N^{-1/2}\sum_{i=1}^ND_{i,j,n,h}=\gamma^{\log}_j(n,h)~.
\end{align}
The desired claim will follow immediately by combining the following statements:
\begin{enumerate}
    \item[(I)] There exists a Gaussian process $\tilde G$ such that for any sufficiently small $\delta>0$ we can find an $x_0$ such that
        \begin{align} 
        \sup_{t \geq x_0}\Big|\PR(\max_{1 \leq j \leq d, (n,h) \in \mathcal{S}}|\gamma^{\log}_j(n,h)|\geq t)-\PR(\max_{\substack{1 \leq j \leq d, (n,h) \in \mathcal{S}\\ h >  N^{1-\delta}}}|\tilde G_{j,n,h}|\geq t)\Big|=o(1)
        \end{align}
    \item[(II)] We show that the coordinates with small bandwidths are negligible. In particular,
        \begin{align}
        \sup_{t \geq x_0}\bigg\{\PR(\max_{1 \leq j \leq d, (n,h) \in \mathcal{S}}|\tilde G_{j,n,h}|\geq t)-\PR(\max_{\substack{1 \leq j \leq d, (n,h) \in \mathcal{S}\\ h >  N^{1-\delta}}}|\tilde G_{j,n,h}|\geq t)\bigg\}=o(1)~.
        \end{align}
    \item[(III)] We show that $\tilde G$ is close to the  Gaussian process $G$ from Theorem \ref{theo1:dep}, i.e.
        \begin{align}
             \sup_{t \geq x_0}\Big|\PR(\max_{1 \leq j \leq d, (n,h) \in \mathcal{S}} |G_{j,n,h}|\geq t)-\PR(\max_{1 \leq j \leq d, (n,h) \in \mathcal{S}}|\tilde G_{j,n,h}|\geq t)\Big|=o(1)~.
        \end{align}
\end{enumerate}

\textbf{Proof of (I)}\\

First, we use Lemma \ref{Lemma:TruncateScalesLog} to obtain an $x_0>0$ that depends only on $\delta,c,a,\gamma$ and $K$ such that
\begin{align}
\label{pe:3}
        \PR\Big(\max_{\substack{1 \leq j \leq d, (n,h) \in \mathcal{S}\\ h \leq N^{1-\delta}}}|\gamma^{\log}_j(n,h)|\geq x_0\Big)=o(1).
    \end{align}   
For the main part, finding the appropriate Gaussian process $\tilde G$, we turn to  Theorem~\ref{Thm:GaussApprox}. In order to use its results, we need to check that $D_i$'s satisfy the assumptions~\eqref{e:vc}-\eqref{ass:th:seg:gaus}.   By definition of $D_{i,j,n,h}$, see \eqref{def_D}, and since  $h \geq 1$ we immediately observe that the properties 
\begin{align}
    \tilde d&\leq C N^{2+\gamma},\\
    \beta_{D_1,...,D_N}(k)&\leq C \exp(-ck^a),\\
    \max_{\substack{1 \leq j \leq d, (n,h) \in \mathcal{S},\\ h \geq N^{1-\delta}, 1\leq i \leq N}}\|D_{i,j,n,h}\|_{\Psi_2} &\leq CKN^{\delta/2}\log(N)^{-\beta}
\end{align}
are inherited from the Assumptions~\ref{ass:segm}. Thus we only need to verify the variance condition \eqref{e:vc} to apply Theorem \ref{Thm:GaussApprox} to the coordinates with $h>N^{1-\delta}$ (or, more specifically, to $\pm$ the coordinates to obtain the result for the maximum of the absolute values). We remark that assumption \eqref{ass:th:seg:gaus:dim} is satisfied for an arbitrarily small power, i.e.  for any $\theta>0$, we have $\log(\tilde dN)<N^\theta$, in the following $\theta$ will be fixed. To that end let $q, r$ be given as in Theorem \ref{Thm:GaussApprox}, that is 
\begin{align}
\label{eq:block-size-defin}
    q=\lfloor N^{4\theta+\delta}\rfloor;\qquad r=\lfloor N^\theta\rfloor.
\end{align}
In order to verify condition \eqref{e:vc} it is sufficient to lower-bound $\frac{1}{mq}\sum_{l=1}^m\E\Big[\sum_{i,k\in I_l}D_{i,j,n,h}D_{k,j,n,h}\Big]$ for one choice of $j,n,h$. We set   $n=\lfloor N/2\rfloor$, $h=\lfloor N/2\rfloor-1$ and choose the $j$ which maximizes the definition of $\bar \sigma_q$.  Then,
\begin{align}
    \frac{1}{mq}\sum_{l=1}^m\E\Big[\sum_{i,k\in I_l}D_{i,j,n,h}D_{k,j,n,h}\Big]
    =\frac{N}{mqh\log(eN/h)^{2\beta}}\sum_{1 \leq l \leq m}\E\Big[\sum_{i,k\in I_l}\epsilon_{ij}\epsilon_{kj}\Big]+o(1),
    \end{align}
where the $o(1)$-term corresponds to products of $D_{i,j,n,h}$ and $D_{k,j,n,h}$ that have different signs in their definitions, these can only appear when the block $I_l$ contains $n$. The maximum number of such terms, each of it is bounded by a constant, is of order $\mathcal{O}(q^2)$. For $\delta,\theta$ sufficiently small they will give a contribution of order $Nq^2/(mqh\log^{2\beta}(N/h))\approx(q+r)q/N=o(1)$, where we plug in our chosen values for $m,h,q,r$. We are continuing bounding  the expression above by noticing that for $h=\lfloor N/2\rfloor-1$ we have that $h\log^{2\beta}(eN/h)$ is of order $N$ , so that
    \begin{align}
    \frac{1}{mq}\sum_{l=1}^m\E\Big[\sum_{i,k\in I_l}D_{i,j,n,h}D_{k,j,n,h}\Big]\geq(1+o(1))\frac{2}{mq\log(2)^{2\beta}}\sum_{1 \leq l \leq m}\E\Big[\sum_{i,k\in I_l}\epsilon_{ij}\epsilon_{kj}\Big].
\end{align}  This quantity is bounded below by Assumption \ref{ass_mixing} so that \eqref{e:vc} is satisfied.

Theorem \ref{Thm:GaussApprox} and equation \eqref{pe:3} thus yield 
\begin{align} \label{e:x0app}
    \sup_{t \geq x_0}\Big|\PR(\max_{1 \leq j \leq d, (n,h) \in \mathcal{S}}| \gamma^{\log}_j(n,h)|\geq t)-\PR(\max_{\substack{1 \leq j \leq d, (n,h) \in \mathcal{S}\\ h >  N^{1-\delta}}}|\tilde G_{j,n,h}|\geq t)\Big|=o(1)
\end{align}
for a Gaussian $\tilde G_{j,n,h}$ with covariance structure
\begin{align}
    \text{Cov}(\tilde G_{j_1,n_1,h_1},\tilde G_{j_2,n_2,h_2})=\frac{1}{mq}\sum_{l=1}^m\E\Big[\sum_{i,k\in I_l}D_{i,j_1,n_1,h_1}D_{k,j_2,n_2,h_2}\Big] 
\end{align}
(the Gaussian is essentially the $Y$ defined following eq. \eqref{p400}; $Y$ depends on the objects $m,q,r$ and $I_l$ defined at the beginning of Section \eqref{e:secGA}). \\

\textbf{Proof of (II):}\\ 

We will first bound the variances of the entries at scale $h$ and then use a union bound argument in combination with a Gaussian tail bound to show that all entries with $h$ sufficiently small are negligible.\\

Note that
\begin{align}
    \text{Var}(\tilde G_{j,n,h})=\frac{1}{mq}\sum_{l=1}^m\E\Big[\sum_{i,k\in I_l}D_{i,j,n,h}D_{k,j,n,h}\Big]
\end{align}
and that $|I_l|=q$, by the definition of $D_1,...,D_N$ in \eqref{def_D},  we therefore have at most $Ch/q$ (choose $C$ so that this is an integer) many blocks $I_l$ for which $\E\Big[\sum_{i,k\in I_l}D_{i,j,n,h}D_{k,j,n,h}\Big]$ is not 0 (when $h \leq q$ there are at most two such blocks, we omit the structurally similar argument for it for the sake of brevity). By possibly reordering the blocks, we have 
\begin{align}
\label{e:Varbound}
     \text{Var}(\tilde G_{j,n,h})&=\frac{1}{mq}\sum_{l=1}^{Ch/q}\E\Big[\sum_{i,k\in I_l}D_{i,j,n,h}D_{k,j,n,h}\Big]\\
     &\leq C\frac{h}{qN}\max_{1 \leq l \leq Ch/q}\E\Big[\Big(\sum_{i \in I_l}D_{i,j,n,h}\Big)^2\Big]\\
     &\leq C\frac{h}{qN}qK^2N/h\log(eN/h)^{-2\beta},
\end{align}
where the last inequality follows by Theorem 3 from \cite{yoshihara:1978}. Note that some of the summands in the first line may well be zero and $Ch/q$ is merely an upper bound. Consequently we obtain
\begin{align} \label{e:Gaussex}
    \max_{\substack{1 \leq j \leq d, (n,h) \in \mathcal{S}\\ h \leq  N^{1-\delta}}}\text{Var}(\tilde G_{j,n,h})=:\sigma_G \le C_\sigma \log(N)^{-2\beta}~.
\end{align}
Here, $C_\sigma$ is some fixed, large enough constant. Now, for any $t >0$, we obtain using the union bound
\begin{align*}
   & \mathbb{P}(\max_{\substack{1 \leq j \leq d, (n,h) \in \mathcal{S}\\ h \leq  N^{1-\delta}}}|\tilde G_{j,n,h}|>t) \le d |\mathcal{S}| \max_{\substack{1 \leq j \leq d, (n,h) \in \mathcal{S}\\ h 
   \leq N^{1-\delta}}}\mathbb{P}(|\tilde G_{j,n,h}|>t) \\
   \le& C N^{2+\gamma} \mathbb{P}(|\mathcal{N}(0,1)|>t\sigma_G^{-1/2})\le C N^{2+\gamma} \mathbb{P}(|\mathcal{N}(0,1)|>tC_\sigma^{-1/2} \log^{\beta}(N)). 
\end{align*}
Using Mill's bound \cite{Gordon:1941}, we can bound the right side with
\begin{align}\label{e:mill}
\frac{C N^{2+\gamma} \exp\big(-t^2(2C_\sigma)^{-1}\log^{2\beta}(N)\big)}{t(C_\sigma)^{-1/2}\log^{\beta}(N)}.
\end{align}
We analyze the numerator further and first focus on the exponential. For some $N_0=N_0(t)$ sufficiently large, it follows for all $N>N_0$
that 
\[
t^2(2C_\sigma)^{-1}\log^{2\beta}(N)\ge (2+\gamma)\log(N)
\]
because $2\beta>1$ by assumption \ref{ass_par_beta_a}. Notice that we can choose $N_0(t)$  to be monotonically decreasing in $t$ here.
Thus,  for all $N>N_0$,
\begin{align*} 
    \frac{C N^{2+\gamma} \exp\big(-t^2(2C_\sigma)^{-1}\log^{2\beta}(N)\big)}{t(C_\sigma)^{-1/2}\log^{\beta}(N)}\le \frac{C}{t(C_\sigma)^{-1/2}\log^{\beta}(N)} 
\end{align*}
and the ratio in \eqref{e:mill} is $o(1)$. Due to monotonicity of $N_0(t)$, the decay is indeed uniform for all $t \in [t_0,\infty)$ where $t_0>0$ is an arbitrary constant. Consequently, we obtain
\[
\sup_{t \ge t_0}\mathbb{P}(\max_{\substack{1 \leq j \leq d, (n,h) \in \mathcal{S}\\ h \leq  N^{1-\delta}}}|\tilde G_{j,n,h}|>t) =o(1)
\]
and thus setting $t_0=x_0$ (the constant occurring in the inequality \eqref{e:x0app}) we get
\begin{align}
    \sup_{t \geq x_0}\bigg\{\PR(\max_{1 \leq j \leq d, (n,h) \in \mathcal{S}}|\tilde G_{j,n,h}|\geq t)-\PR(\max_{\substack{1 \leq j \leq d, (n,h) \in \mathcal{S}\\ h >  N^{1-\delta}}}|\tilde G_{j,n,h}|\geq t)\bigg\}=o(1)~.
\end{align}

\textbf{Proof of (III):}\\

Note that  $(G_{j,n,h})$ is a $\tilde d$ dimensional Gaussian random vector whose covariance matrix has the entries
\begin{align}
    \text{Cov}(G_{j_1,n_1,h_1}, G_{j_2,n_2,h_2})= \text{Cov}( \gamma^{\log}_{j_1}(n_1,h_1),  \gamma^{\log}_{j_2}(n_2,h_2))
    =
    \frac{1}{N}\E\Big[\sum_{1 \leq i,k \leq N}D_{i,j_1,n_1,h_1}D_{k,j_2,n_2,h_2}\Big],
\end{align}
where we used \eqref{e:Hold:Gau:Log} and \eqref{eq_gamma_D_ij}. 
Using the same arguments as before to remove coordinates with $h \leq N^{1-\delta}$ the proof is finished if we can establish that
\begin{align}
    \sup_{t \geq x_0}\Big|\PR(\max_{\substack{1 \leq j \leq d, (n,h) \in \mathcal{S}\\h\leq N^{1-\delta}}} |G_{j,n,h}|\geq t)-\PR(\max_{\substack{1 \leq j \leq d, (n,h) \in \mathcal{S}\\h\leq N^{1-\delta}}}|\tilde G_{j,n,h}|\geq t)\Big|=o(1)~.
\end{align}
Using Lemma~\ref{lem:compar_gaus} we can bound the l.h.s. by  $\Delta^{1/3}(\log N)^{2\beta}(\log\tilde d)^{2/3}$, where $\Delta$ is defined as 
\begin{align}
    \Delta=\max_{\substack{1 \leq j_1, j_2 \leq d \\ (n_1, h_1), (n_2, h_2) \in \mathcal{S}\\h_1,h_2 \leq N^{1-\delta}}}
\Big|\frac{1}{mq}\sum_{l=1}^m\E\Big[\sum_{i,k\in I_l}D_{i,j_1,n_1,h_1}D_{k,j_2,n_2,h_2}\Big]-\E\Big[\frac{1}{N}\sum_{1 \leq i,k \leq N}D_{i,j_1,n_1,h_1}D_{k,j_2,n_2,h_2}\Big] \Big|~.
\end{align}
Taking into account Assumption~\ref{ass_dim_seg}, it is enough to show that
\begin{align}
\label{pe:2log}
    \Delta^{1/3}(\log N)^2\leq C N^{-\delta_1} 
\end{align}
for some (arbitrarily small) constant $\delta_1>0$. To obtain a suitable bound for $\Delta$, consider the following calculations:
\begin{align}
\label{e:DelBound1}
 & \hspace{-4mm}\max_{\substack{1 \leq j_1, j_2 \leq d \\ (n_1, h_1), (n_2, h_2) \in \mathcal{S}}}\!
\Big|\frac{1}{mq}\sum_{l=1}^m\!\E\Big[\sum_{i,k\in I_l}D_{i,j_1,n_1,h_1}D_{k,j_2,n_2,h_2}\Big]\!-\!\frac{1}{N}\sum_{l=1}^m\E\Big[\sum_{i,k\in I_l}D_{i,j_1,n_1,h_1}D_{k,j_2,n_2,h_2}\Big]\Big|\\
&\leq  \Big|\frac{1}{mq}-\frac{1}{N}\Big| \max_{\substack{1 \leq j_1, j_2 \leq d \\ (n_1, h_1), (n_2, h_2) \in \mathcal{S}}}\Big|\sum_{l=1}^m\E\Big[\sum_{i,k\in I_l}D_{i,j_1,n_1,h_1}D_{k,j_2,n_2,h_2}\Big]\Big|\\
&\leq  C\Big|\frac{1}{mq}-\frac{1}{N}\Big| K^2N\leq CK^2(1-\frac{q}{q+r})\leq CK^2\frac{r}{q},
\end{align}
where we used similar arguments as those yielding \eqref{e:Varbound} to obtain the second to last inequality. Next we observe that

\begin{align}
\label{e:DelBound2}
   & 
\frac{1}{N}\E\Big[\sum_{1 \leq i,k \leq N}D_{i,j_1,n_1,h_1}D_{k,j_2,n_2,h_2}\Big]-\frac{1}{N}\sum_{l=1}^m\E\Big[\sum_{i,k\in I_l}D_{i,j_1,n_1,h_1}D_{k,j_2,n_2,h_2}\Big]\\
=&I_1+I_2+I_3+I_4+I_5+I_6 + I_7 + I_8,
\end{align}
where
\begin{align}
    I_1&=N^{-1}\sum_{1 \leq l_1\neq l_2 \leq m}\E\Big[\sum_{i \in I_{l_1},k\in I_{l_2}}D_{i,j_1,n_1,h_1}D_{k,j_2,n_2,h_2}\Big],\\
    I_2&=N^{-1}\sum_{1 \leq l_1\neq l_2 \leq m}\E\Big[\sum_{i \in J_{l_1},k\in J_{l_2}}D_{i,j_1,n_1,h_1}D_{k,j_2,n_2,h_2}\Big],\\
    I_3&=N^{-1}\sum_{1 \leq l_1,l_2 \leq m}\E\Big[\sum_{i \in I_{l_1},k\in J_{l_2}}D_{i,j_1,n_1,h_1}D_{k,j_2,n_2,h_2}\Big],\\
    I_4&=N^{-1}\sum_{1 \leq l_1,l_2 \leq m}\E\Big[\sum_{i \in J_{l_1},k\in I_{l_2}}D_{i,j_1,n_1,h_1}D_{k,j_2,n_2,h_2}\Big],\\
    I_5&=N^{-1}\sum_{l=1}^m\sum_{i \in J_{m+1}, k \in I_{l}}\E\Big[D_{i,j_1,n_1,h_1}D_{k,j_2,n_2,h_2}+D_{k,j_1,n_1,h_1}D_{i,j_2,n_2,h_2}\Big],\\
    I_6 &= N^{-1}\sum_{l=1}^m\sum_{i \in J_{m+1}, k \in J_{l}}\E\Big[D_{i,j_1,n_1,h_1}D_{k,j_2,n_2,h_2}+D_{k,j_1,n_1,h_1}D_{i,j_2,n_2,h_2}\Big],\\
    I_7&=N^{-1}\sum_{l=1}^m\sum_{i,k \in J_{l}}\E\Big[D_{i,j_1,n_1,h_1}D_{k,j_2,n_2,h_2}]~ , \\    
    I_8 &= N^{-1}\sum_{(i,k)\in J_{m+1}\times J_{m+1}}\E\Big[D_{i,j_1,n_1,h_1}D_{k,j_2,n_2,h_2}\Big]. 
\end{align}
The task is then to generate bounds for each of those terms (in absolute value).
For notational simplicity we consider only the cases where $n_1=n_2,h_1=h_2$, the cases with distinct indices are handled in the same manner. $|I_1|$ to $|I_5|$ can be bounded in similar ways, for instance
\begin{align}
     |I_1| &\leq CN^{-1}\Big( K^2(q^2+r^2+qr)m^2\beta(r)^{1/4}\Big)\\
     |I_3| &\leq  CN^{-1}K^2\frac{h}{q}\frac{N}{h}\sum_{1 \leq i \leq q}\sum_{1 \leq j \leq r}\beta(q+j-i)^{1/4}
\end{align}
follows from Lemma B in \cite{yoshihara:1978} in combination with similar arguments as in \eqref{e:Varbound}. For $|I_1|$ to $|I_6|$ we may upper bound the resulting quantities by $ C(K^2N^{-1}+\frac{r}{q})$. For $|I_7|$ we proceed similar as for \eqref{e:Varbound}: we have at most $C h/q$ many blocks which are not 0. By the same arguments as in \eqref{e:Varbound} we thus have
\[
    |I_7|\leq CN^{-1}\frac{h}{q} rK^2N/h\log(N/h)^{-2\beta}\leq C\frac{r}{q}.
\]
Using Lemma B in \cite{yoshihara:1978} we see that $|I_8|$ is $O(rN^{-1}N/h\log(N/h)^{-2\beta})=O(r/q)$ (here we use $h>q$). Combining \eqref{e:DelBound1} and \eqref{e:DelBound2} thus yields that $\Delta \leq C r/q$. Eventually, this then yields \eqref{pe:2log} due to \eqref{eq:block-size-defin}.
\end{proof}

\subsection{Polynomial weighting}\label{app:poly}

As in the previous subsection we also establish a more general result for the statistic with polynomial weighting: 
\begin{align}
    \gamma_i(n,h)
    =
    \frac{\sum_{m=n-h+1}^{n}X_{m,i}-\sum_{m=n+1}^{n+h}X_{m,i}}{N^{1/2-\beta}h^{\beta}}.
\end{align}
The assumptions necessary to establish the corresponding result are listed directly in the formulation of the theorem. Some of them, i.e. assumptions concerning the variances are identical to \ref{ass_var_seg} and \ref{ass_mixing}, however, because of the different weighting, assumptions on the dimensions are different.

\begin{theo}
\label{theo2:dep}
    Let $\beta<1/2$, let the sample $X_1,...,X_N$ have mean zero and satisfy
    \begin{align}
        \bar \sigma_q &\geq \underline{c}\\
       \beta(k)&\leq C \exp(-ck^a)\\              
       \max_{1 \leq i \leq N, 1 \leq j \leq d}\|\epsilon_{i,j}\|_{\Psi_2}&\leq K\\
       \log(d)&\leq CN^\gamma\\
       \frac{N^{\beta\gamma/(1/2-\beta)+6.5\gamma}}{\sqrt{N}}&\leq C_\gamma N^{-\tilde c}
    \end{align}
    for some fixed $a,c,\tilde c, \underline{c}, K, C, C_\gamma,\gamma>0$. We then have that there exists a Gaussian process $G \in \R^{d|\mathcal{S}|}$ such that
    \begin{align}
        \label{e:Hold:Gau:Poly}
        \sup_{t \ge x_0} \big|\PR(\sup_{1 \leq j \leq d, (n,h) \in \mathcal{S}}|G_{j,n,h}|\le t)- \mathbb{P}( \sup_{1 \leq j \leq d, (n,h) \in \mathcal{S}}|\gamma_{j}(n,h)|\le t)\big|=o(1),            
    \end{align}    
    for some $x_0>0$ that depends only on $a,c, K, C,C_\gamma,\gamma, \underline{c}$. The covariance  of this Gaussian process is characterized by
    \begin{align}
        \text{Cov}(G_{j_1,n_1,h_1},G_{j_2,n_2,h_2})=\text{Cov}( \gamma_{j_1}(n_1,h_1),  \gamma_{j_2}(n_2,h_2))
    \end{align}
\end{theo}

\begin{proof}
Consider the vectors $D_1,...,D_N$ with components
\begin{align}
    D_{i,j,n,h}=
    \begin{cases}
        0    & \quad \begin{array}{c} 1 \leq i < n-h+1 \end{array}, \\
        \epsilon_{ij}/\Big((h/N)^{\beta}\Big) & \quad \begin{array}{c} n-h+1 \leq i \leq n \end{array}, \\
        -\epsilon_{ij}/\Big((h/N)^{\beta}\Big)& \quad \begin{array}{c} n < i \leq n+h \end{array}, \\
        0    & \quad \begin{array}{c} n+h < i \leq N \end{array}.
    \end{cases}
\end{align}
indexed by $1 \leq j \leq d$ and  $(n,h) \in \mathcal{S}$ which has dimension $\tilde d=d|\mathcal{S}|\leq dN^2$. Consequently $\log(\tilde d)\leq C \log(d)$ under the assumption that $\log(d)\leq C N^\gamma$ for $\gamma>0$. It is easy to see that
\begin{align}
    N^{-1/2}\sum_{i=1}^ND_{i,j,n,h}=\gamma_j(n,h)~.
\end{align}
The desired claim will follow by combining the following statements:
\begin{enumerate}
    \item[(I)] Let $\lambda^\star=\frac{1}{a(1-2\beta)}+1$. There exists a Gaussian process $\tilde G$ such that we can find an $x_0$ such that
        \begin{align} 
        \sup_{t \geq x_0}\Big|\PR(\max_{1 \leq j \leq d, (n,h) \in \mathcal{S}}| \gamma_j(n,h)|\geq t)-\PR(\max_{\substack{1 \leq j \leq d, (n,h) \in \mathcal{S}\\ h >  N^{1-\delta^\star}/\log(N)^{\lambda^\star}}}|\tilde G_{j,n,h}|\geq t)\Big|=o(1)
        \end{align}
    \item[(II)] We show that the coordinates with small bandwidths are negligible. 
        \begin{align}
        \sup_{t \geq x_0}\bigg\{\PR(\max_{1 \leq j \leq d, (n,h) \in \mathcal{S}}|\tilde G_{j,n,h}|\geq t)-\PR(\max_{\substack{1 \leq j \leq d, (n,h) \in \mathcal{S}\\ h >  N^{1-\delta^\star}/\log(N)^{\lambda^\star}}}|\tilde G_{j,n,h}|\geq t)\bigg\}=o(1)~.
        \end{align}
    \item[(III)] We show that $\tilde G$ is close to the  Gaussian process $G$ from Theorem \ref{theo2:dep}, i.e.
        \begin{align}
             \sup_{t \geq x_0}\Big|\PR(\max_{1 \leq j \leq d, (n,h) \in \mathcal{S}} |G_{j,n,h}|\geq t)-\PR(\max_{1 \leq j \leq d, (n,h) \in \mathcal{S}}|\tilde G_{j,n,h}|\geq t)\Big|=o(1)~.
        \end{align}
\end{enumerate}
Combining the three steps immediately yields the desired result.

\textbf{Proof of (I)}
Note that by Lemma \ref{Lemma:TruncateScalesPoly}, we obtain that
\begin{align}
\label{pe:1}
        \PR\Big(\max_{\substack{1 \leq j \leq d, (n,h) \in \mathcal{S}\\ h \leq N^{1-\delta^\star}/\log(N)^{\lambda^\star}}}| \gamma_j(n,h)|\geq x_0\Big)=o(1)
    \end{align}
    for any $x_0>0$, where $\delta^\star=\gamma/(1-2\beta)$ in the assumptions of the theorem (they match those in the statement of the lemma).

Using the definition of $D_{i,j,n,h}$
we immediately observe that the  following key properties
\begin{align}
    \beta_{D_1,...,D_N}(k)&\leq C \exp(-ck^a)\\
    \max_{\substack{1 \leq j \leq d, (n,h) \in \mathcal{S}\\ h \geq N^{1-\delta^\star}/\log(N)^{1/a}\\1 \leq i \leq N}}\|D_{i,j,n,h}\|_{\Psi_2} &\leq C KN^{\beta\gamma/(1/2-\beta)}
\end{align}
are inherited from the assumptions on $(\epsilon_{ij})$. Moreover, since we have for the dimension $\log(d) \leq C N^\gamma$ we obtain
\[
\frac{N^{\beta\gamma/(1/2-\beta)}\log(\tilde dN)^{6.5}}{\sqrt{N}}\le C_q N^{-c}
\]
for all $N$ sufficiently large and with some large enough, fixed  constant $C_q$ and 
\[
c = \frac{1}{2}-\gamma \Big( \frac{13}{2} + \frac{\beta}{1/2-\beta}\Big).
\]
For the calculation we have used that $\tilde d\le N^2 d$ (as noticed above in this proof).

Thus we only need to verify the variance condition \eqref{e:vc} to apply Theorem \ref{Thm:GaussApprox} to the remaining coordinates (or, more specifically, to $\pm$ the coordinates to obtain the result for the maximum of the absolute values). Choosing $n=\lfloor N/2\rfloor$ and $h=\lfloor N/2\rfloor-1$ yields that this condition is also inherited from $(\epsilon_{i,j})$ (see the proof of Theorem \ref{theo1:dep} for a more detailed derivation). Now, combining Theorem \ref{Thm:GaussApprox} which yields
\begin{align}
    \sup_{t \geq x_0}\Big|\PR(\max_{\substack{1 \leq j \leq d, (n,h) \in \mathcal{S}\\ h >  N^{1-\delta^\star}/\log(N)^{\lambda^\star}}}| \gamma_j(n,h)|\geq t)-\PR(\max_{\substack{1 \leq j \leq d, (n,h) \in \mathcal{S}\\ h >  N^{1-\delta^\star}/\log(N)^{\lambda^\star}}}|\tilde G_{j,n,h}|\geq t)\Big|=o(1)
\end{align}
and equation \eqref{pe:1}, gives us
\begin{align}
    \sup_{t \geq x_0}\Big|\PR(\max_{1 \leq j \leq d, (n,h) \in \mathcal{S}}| \gamma_j(n,h)|\geq t)-\PR(\max_{\substack{1 \leq j \leq d, (n,h) \in \mathcal{S}\\ h >  N^{1-\delta^\star}/\log(N)^{\lambda^\star}}}|\tilde G_{j,n,h}|\geq t)\Big|=o(1).
\end{align}

Here, $\tilde G_{j,n,h}$ are centered Gaussian random variables with covariance structure
\begin{align}
    \text{Cov}(\tilde G_{j_1,n_1,h_1}, \tilde G_{j_2,n_2,h_2})=\frac{1}{mq}\sum_{l=1}^m\E\Big[\sum_{i,k\in I_l}D_{i,j_1,n_1,h_1}D_{k,j_2,n_2,h_2}\Big]
\end{align}
corresponding to the object $Y$ in Theorem \ref{Thm:GaussApprox}. Notice that this object $Y$ in turn has an involved covariance structure involving the above objects $m,q,I_l$. For definitions of these objects, we refer to the
 beginning of Section \ref{e:secGA}, as well as the definitions following eq. \eqref{p400}. \\

 \textbf{Proof of (II):}\\
 Elementary calculations involving counting the blocks where $D_{i,j,n,h}$ is 0 and Theorem 3 from \cite{yoshihara:1978} (see the arguments following equation \eqref{e:Varbound} for a detailed derivation in a similar setting) show that 
\begin{align}
    \max_{\substack{1 \leq j \leq d, (n,h) \in \mathcal{S}\\ h \leq  N^{1-\delta^\star}}}\text{Var}(\tilde G_{j,n,h})=O(N^{-\gamma})~.
\end{align}
Now, using similar arguments as those following eq. \eqref{e:Gaussex}, we can deduce that
\[
\PR(\max_{\substack{1 \leq j \leq d, (n,h) \in \mathcal{S}\\ h \le  N^{1-\delta^\star}/\log(N)^{\lambda^\star}}}|\tilde G_{j,n,h}|\geq t)=o(1),
\]
uniformly over all $t \ge x_0$, which in turn implies
\begin{align}
    \sup_{t \geq x_0}\Big\{\PR(\max_{1 \leq j \leq d, (n,h) \in \mathcal{S}}|\tilde G_{j,n,h}|\geq t)-\PR(\max_{\substack{1 \leq j \leq d, (n,h) \in \mathcal{S}\\ h >  N^{1-\delta^\star}/\log(N)^{\lambda^\star}}}|\tilde G_{j,n,h}|\geq t)\Big\}=o(1)~.
\end{align}

\textbf{Proof of (III):}\\
Let $(G_{j,n,h})$ be a $\tilde d$ dimensional Gaussian with covariance matrix given by 
\begin{align}
    \text{Cov}(G_{j_1,n_1,h_1}G_{j_2,n_2,h_2})&=\frac{1}{N}\E\Big[\sum_{1 \leq i,k \leq N}D_{i,j_1,n_1,h_1}D_{k,j_2,n_2,h_2}\Big]\\
    &=  \text{Cov}\Big(\gamma_{j_1}(n_1,h_1),\gamma_{j_2}(n_2,h_2)\Big)
\end{align}
The proof is finished if we can establish that
\begin{align}
    \sup_{t \geq x_0}\Big|\PR(\max_{1 \leq j \leq d, (n,h) \in \mathcal{S}} |G_{j,n,h}|\geq t)-\PR(\max_{1 \leq j \leq d, (n,h) \in \mathcal{S}}|\tilde G_{j,n,h}|\geq t)\Big|=o(1)~.
\end{align}
By Lemma \ref{lem:compar_gaus} (and again first removing the coordinates with small $h$) we therefore only need to show that
\begin{align}
\label{pe:2}
    \Delta^{1/3}\log(\tilde dN)\leq C N^{-\delta_1}
\end{align}
for some (arbitrarily small)  $\delta_1>0$, and
\begin{align}
    \Delta:=\max_{\substack{1 \leq j_1, j_2 \leq d \\ (n_1, h_1), (n_2, h_2) \in \mathcal{S}}}
\Big|\frac{1}{mq}\sum_{l=1}^m\E\Big[\sum_{i,k\in I_l}D_{i,j_1,n_1,h_1}D_{k,j_2,n_2,h_2}\Big]-\E\Big[\frac{1}{N}\sum_{1 \leq i,k \leq N}D_{i,j_1,n_1,h_1}D_{k,j_2,n_2,h_2}\Big] \Big|.
\end{align}
The same calculations as those following \eqref{pe:2log} yield the desired result.

\end{proof}

\subsection{Small scales do not matter}

\begin{lemma}
\label{Lemma:TruncateScalesLog}
  Under the assumptions \ref{ass_dim_seg}-\ref{ass_par_beta_a},
for any fixed $\delta>0$, there exists $x_0=x_0(K,a,\beta,\gamma,\delta,\overline{C})$   such that
                \begin{align}
                    \PR\Big(\max_{\substack{1 \leq j \leq d, (n,h) \in \mathcal{S}\\ h \leq N^{1-\delta}}}|\gamma^{\log}_j(n,h)|\geq x_0\Big)=o(1)~.
                \end{align}        
\end{lemma}
\begin{proof}
    Using the definition of $ \gamma^{\log}_j(n,h)$ and Lemma \ref{lem:SumSplit} with $m=(2c^{-1}\log(N))^{1/a}$  we obtain that for any $\delta>0$ 
    \begin{align}
        \PR\Big(\max_{\substack{1 \leq j \leq d, (n,h) \in \mathcal{S}\\ h \leq N^{1-\delta}}} |\gamma^{\log}_j(n,h)|\geq t\Big)\leq 2\overline{C}mN^{2+\gamma}\exp\Big(-\frac{Cc^{1/a}\delta^{2\beta}\log(N)^{2\beta-1/a}t^2 }{4K^2}\Big)+\overline{C}4N^{-1}
    \end{align}
   Using that $2\beta-1/a\geq 1$ we then have for any $x_0>0$ for which $x_0^2>\frac{4(2+\gamma)2K^2}{Cc^{1/a}\delta^{2\beta}}$ holds, that
    \begin{align}
        \PR\Big(\max_{\substack{1 \leq j \leq d, (n,h) \in \mathcal{S}\\ h \leq N^{1-\delta}}}|\gamma^{\log}_j(n,h)|\geq x_0\Big)=o(1)~.
    \end{align}        
\end{proof}

\begin{lemma}
\label{Lemma:TruncateScalesPoly}
  Under the assumptions of Theorem~\ref{theo2:dep}
    \begin{align}
        \max_{1 \leq i \leq N, 1 \leq j \leq d}\|\epsilon_{ij}\|_{\Psi_2}&\leq K  \\
        \beta(k)&\leq \overline{C} \exp(-ck^a)\\
        \log(d) &\leq \overline{C} N^\gamma.       
    \end{align} 
    We then have for $\beta<1/2$ that  
                \begin{align}
                   \PR\Big(\max_{\substack{1 \leq j \leq d, (n,h) \in \mathcal{S}\\ h \leq N^{1-\gamma/(1-2\beta)}/\log(N)^{\lambda^\star}}}| \gamma_j(n,h)|\geq x_0\Big)=o(1)
                \end{align}
   for any $x_0>0$ and $\lambda^\star=\frac{1}{a(1-2\beta)}+1$.

\end{lemma}
\begin{proof}    
    Using the definition of $ \gamma_j(n,h)$ and Lemma \ref{lem:SumSplit} with $m=(2c^{-1}\log(N))^{1/a}$  we obtain that
    \begin{align}
    \label{pe:11}
        \PR\Big(\max_{\substack{1 \leq j \leq d, (n,h) \in \mathcal{S}\\ h \leq N^{1-\delta^\star}/\log(N)^{\lambda^\star}}}| \gamma_j(n,h)|\geq x_0\Big)\leq  2mN^2d\exp\Big(-\frac{cx_0^2 N^{\gamma}\log(N)^{1-2\beta}}{4K^2}\Big)+C4N^{-1}
    \end{align}
    where we define $\delta^\star=\gamma/(1-2\beta)$.  Note that for $N$ sufficiently large we have $\log(mdN)\leq  2CN^\gamma$. We hence obtain that
    \begin{align}
        \PR\Big(\max_{\substack{1 \leq j \leq d, (n,h) \in \mathcal{S}\\ h \leq N^{1-\delta^\star}/\log(N)^{\lambda^\star}}}| \gamma_j(n,h)|\geq x_0\Big)=o(1)
    \end{align}
    as desired.    
\end{proof}

\section{Gaussian approximation for dependent data} \label{e:secGA}
\subsection{Approximation for bounded data}
\label{appen:GAP-dependent}
In this subsection we present a high dimensional Gaussian approximation result for dependent data that allows for the variances of all but one coordinate to degenerate. The proof is based on a (in other contexts well-known) large-block-small-block decomposition. The reader should recall the definition of the large and small blocks, their sizes and their number ($I_j, J_j, q,r, m$) given in Section \ref{appen:segment}. For our reference, we define the variable
\begin{align}
    \label{def:TN}
    T_N=\max_{1 \leq j \leq d}  N^{-1/2}\sum_{i=1}^{N}\eta_{ij}~.
\end{align}
In this section, we assume that there exists some constant $D_N$ such that 
\begin{align} 
\label{p400}
    |\eta_{ij}|\leq D_N \qquad 1 \leq i \leq N, \,\,1 \leq j \leq d.
\end{align}
Further let 
\begin{align}
    B_l=\sum_{ i \in I_l} \eta_i, \quad S_l= \sum_{ i \in J_l} \eta_i
\end{align}
and $\{ \tilde B_l , 1 \leq l \leq m\}$, $\{ \tilde S_l, 1\leq l \leq m\}$ be sequences of independent random vectors with marginal distributions given by those of $B_l$ and $S_l$. Additionally let 
\begin{align}
\label{def:Y}
    Y=(Y_1,...,Y_d)
\end{align} be a centered normal vector with covariance matrix $\E[YY^T]=(mq)^{-1}\sum_{i=1}^m\E[B_iB_i^T]$. \\
We also recall the definition of $\bar\sigma_r$ and $\bar\sigma_q$, wich in this section will correspond to the variables $\eta_i$:
\begin{align}
    \bar \sigma_q:=\max_{1 \leq j \leq d }\frac{1}{m}\sum_{I \in \{I_1,...,I_m\}}\text{Var}\left(q^{-1/2}\sum_{i \in I}\eta_{ij}\right) \\
    \bar{\sigma}_r:=\max_{1 \leq j \leq d }\frac{1}{m}\sum_{J \in \{J_1,...,J_m\}}\text{Var}\left(r^{-1/2}\sum_{i \in J}\eta_{ij}\right) 
\end{align}
\begin{theo}
\label{pt1}
   Suppose that \eqref{p400} holds and that for some $c_1,c_2>0$ we have
   \begin{align}
       \bar \sigma_q &\geq  \underline{c} \\
       \bar \sigma_q \lor  \bar {\sigma}_r &\leq \overline{c}\\
       (q+r\log(dN))\log(dN)^{1/2}D_NN^{-1/2}&\leq C N^{-c_1}\\
       r/q \log(d)^2 & \leq C N^{-c_1}\\
       q^{1/2}D_N\log(dN)^4N^{-1/2}&\leq C N^{-c_2}       ~.
   \end{align}
   Then
   \begin{align}
   \label{boundedGAP}
     \sup_{t >x_0}\left|\PR(T_N\leq t)-\PR\Big(\max_{1 \leq j \leq d}Y_j\leq t\Big)\right|\leq C N^{-c}+(m-1)\beta(r)
   \end{align}
for some $c>0$ only depending on $c_1,c_2$ and any $x_0>0$. In particular we also have that 
\begin{align}
\label{anticonc}
     \sup_{t >x_0}\left|\PR(T_N\leq t+\epsilon)-\PR(T_N\leq t)\right|\leq C \Big(N^{-c}+\epsilon\log(dN)\Big)+(m-1)\beta(r)
\end{align}
where $C$ does not depend on $\epsilon$.
\end{theo}
\begin{proof}
    The proof follows by exactly the same arguments as  in the proof of Theorem 5.1 in \cite{bastian:2024}. The improvement to $\log(dN)^4$ instead of $\log(dN)^5$ stems from an improvement of Theorem 4.1 from \cite{chetverikov:wilhelm:kim:2021} that is used in the course of the proof. For the convenience of the reader we include it here. \\
    For the remainder of this proof $c$ denotes a generic and sufficiently small positive constant that may change from line to line that ultimately only depends on $c_1,c_2$. Its final value is the one appearing in the theorem statement.  First we bound the error incurred by leaving out the small blocks, i.e. we have that
    \begin{align}
        |T_N-\max_{1 \leq j \leq d}N^{-1/2}\sum_{i=1}^mB_{ij}|\leq \max_{1 \leq j \leq d}\Big|N^{-1/2}\sum_{i=1}^mS_{ij}\Big| + \max_{1 \leq j \leq d}\Big|N^{-1/2}S_{(m+1)j}\Big|
    \end{align}
and note that
\begin{align}  
\label{b1}
    \max_{1 \leq j \leq d}\Big|N^{-1/2}S_{(m+1)j}\Big|&\leq C  qD_N/\sqrt{N}     
\end{align}

By Corollary 2.7 from \cite{Yu:1994} we have
\begin{align}
\label{b2}
   &\Big| \PR\Big(\max_{1 \leq j \leq d}N^{-1/2}\sum_{i=1}^mB_{ij}\leq t\Big)-\PR\Big(\max_{1 \leq j \leq d}N^{-1/2}\sum_{i=1}^m\tilde B_{ij}\leq t\Big)\Big| \leq (m-1)\beta(r)\\
   &\Big| \PR\Big(\max_{1 \leq j \leq d}N^{-1/2}\Big|\sum_{i=1}^mS_{ij}\Big|\leq t\Big)-\PR\Big(\max_{1 \leq j \leq d}N^{-1/2}\Big|\sum_{i=1}^m\tilde S_{ij}\Big|\leq t\Big)\Big| \leq (m-1)\beta(q)
\end{align}
Using independence and the fact that $\tilde S_{ij}$ is bounded by $rD_N$ we obtain by Lemma D.3 from \cite{chetverikov:wilhelm:kim:2021} that
\begin{align}
    \E\left[\max_{1 \leq j \leq d}\Big|N^{-1/2}\sum_{i=1}^m\tilde S_{ij}\Big|\right]\leq C\Big(\sqrt{r/q \bar \sigma_r\log(d)}+N^{-1/2}rD_N\log(d) \Big)
\end{align}
so that 
\begin{align}
\label{p3}
     \PR\Big(\max_{1 \leq j \leq d}N^{-1/2}\Big|\sum_{i=1}^m\tilde S_{ij}\Big|> t\Big) &\leq C \left(\sqrt{r/q \bar \sigma_r\log(d)}+N^{-1/2}rD_N\log(d)\right)/t   
\end{align}

Let $\delta_1=\sqrt{N^{-c_1}/\log(dN)}$. Combining \eqref{b1}, \eqref{b2}, \eqref{p3} (choose $t=\delta_1$ in \eqref{p3}) we obtain for $t>x_0$ that
\begin{align}
\label{p300}
    \PR(T_N\leq t)& \leq  \PR\Big(\max_{1 \leq j \leq d}N^{-1/2}\sum_{i=1}^mB_{ij}\leq t+\delta_1+qD_NN^{-1/2}\Big)+CN^{-c}+(m-1)\beta(r)\\
    & \leq \PR\Big(\max_{1 \leq j \leq d}N^{-1/2}\sum_{i=1}^m\tilde B_{ij}\leq t+\delta_1+qD_NN^{-1/2}\Big)+CN^{-c}+(m-1)\beta(r)\\
    & \leq  \PR\Big(\max_{1 \leq j \leq d}N^{-1/2}\sum_{i=1}^m\tilde B_{ij}\leq t\Big)+CN^{-c}+ C\Big(\frac{qD_N^2\log(dN)^{8}}{N}\Big)^{1/6}+(m-1)\beta(r)
\end{align}
where we used Lemma 5.2 from \cite{bastian:2024} for the last line. A similar argument also yields the reverse inequality so that for all $t>x_0$ we have
\begin{align}
\label{p301}
   \sup_{t >x_0}\left|\PR(T_N\leq t)-\PR\Big(\max_{1 \leq j \leq d}N^{-1/2}\sum_{i=1}^m\tilde B_{ij}\leq t\Big)\right|\leq C N^{-c}+(m-1)\beta(r).
\end{align}
Note that for at least one $j$ we have $\bar \sigma_q \simeq N^{-1}\sum_{i=1}^m\text{Var}(\tilde B_{ij})$ and that for all $j$ it holds that $\tilde B_{ij}\leq C qD_N$. We may therefore apply a modification of Theorem 4.1 from \cite{chetverikov:wilhelm:kim:2021} (using Theorem 2.1 from \cite{chernozhukov:chetverikov:kato:koike:2022} instead of the Gaussian approximation used therein) to further obtain that
\begin{align}
     \sup_{t >x_0}\left|\PR(T_N\leq t)-\PR\Big(\max_{1 \leq j \leq d}\sqrt{mq/N}Y_j\leq t\Big)\right|\leq  C\Big(\frac{qD_N^2\log(dN)^8}{N}\Big)^{1/4}+(m-1)\beta(r)
\end{align}
Equation \eqref{boundedGAP} then follows by the same arguments as in step 4 of the proof of Theorem E.1 from \cite{Chernozhukov:Chetverikov:Kato:2019} where we use Lemma C.3 from \cite{chetverikov:wilhelm:kim:2021} instead of Step 3. The inequality \eqref{anticonc} is a consequence of equation \eqref{p301} and Lemma 5.2 from \cite{bastian:2024}. \\
\end{proof}

\subsection{Unbounded data}
In the previous step (Theorem \ref{pt1}), we imposed uniform boundedness on all random vectors. We now remove this assumption from the Gaussian approximation at the cost of a slightly stricter assumption on the growth rate of $d$. 
\begin{theo}
\label{Thm:GaussApprox}
Assume that 
     \begin{align}
       \bar \sigma_q \label{e:vc}&\geq  \underline{c} \\
       \beta(r)&=o(N^{-1})\\      
       \max_{1 \leq i \leq N, 1 \leq j \leq d}\|\eta_{ij}\|_{\Psi_2}& \leq B_N\\
       \log(dN)&\leq \overline{C}N^\gamma\label{ass:th:seg:gaus:dim}\\
       \frac{B_NN^{6.5\gamma}}{\sqrt{N}}& \leq \overline{C}N^{-c_1}       \label{ass:th:seg:gaus}
   \end{align}
   for some $c_1>0$, where $q$ and $r$ are given by
    \begin{align}
        \label{def:qr}
        q=\Big\lfloor N^{4\gamma+\delta}\Big\rfloor;\qquad
        r=\Big\lfloor N^\gamma\Big\rfloor 
    \end{align}
    for some small $\delta<c_1$.   
   Then, for some constant $c>0$,
   \begin{align}
   \label{unboundedGAP}
     \sup_{t >x_0}\left|\PR(T_N\leq t)-\PR\Big(\max_{1 \leq j \leq d}Y_j\leq t\Big)\right|\leq CN^{-c}.
   \end{align}  
   where $Y$ and $T_N$ are given in \eqref{def:Y} and \eqref{def:TN}, respectively.
Additionally it holds that
   \begin{align}
\label{anticoncUnb}
     \sup_{t >x_0}\left|\PR(T_N\leq t+\epsilon)-\PR(T_N\leq t)\right|\leq C N^{-c}+\epsilon\log(dN)
\end{align}

\end{theo}
\begin{proof}

 For the proof, we define
 \begin{align}
         \eta_{ij}^{\le u}:=\eta_{ij}1\{|\eta_{ij}|\leq u\}, \qquad   \eta_{ij}^{> u}:=\eta_{ij}1\{|\eta_{ij}|> u\}
    \end{align}
    where $u>0$ will be specified later. The corresponding vectors are called $\eta_i^{\le u}, \eta_i^{>u}$. Note that
    \begin{align}
        &\PR\Big(\Big\|\sum_{i=1}^N\eta_i^{>u}\Big\|_\infty \neq 0\Big) \leq dN\max_{1 \leq i \leq N, 1 \leq j \leq d}\PR(\eta_{ij}^{> u}\neq 0)\\
        &=dN\max_{1 \leq i \leq N, 1\leq j \leq d}\PR(|\eta_{ij}|>u) \leq 2dN \exp\bigg( -\frac{u^2}{2B_N^2}\bigg).
    \end{align}
    In the last step, we have used the well-known tail-bound for subgaussian random variables (see Theorem 2.6.2 in \cite{Vershynin:2018}). Now, we may choose for an integer  $k \ge 2$ 
    \begin{align} \label{e:uchoice}
    u = \sqrt{2(k+1)} B_N\sqrt{\log(dN)},
    \end{align}
    giving for the right-hand side the further upper bound
    \begin{align*}
        2dN \exp\bigg(- (k+1)\log(dN) \bigg) = \frac{2}{(dN)^k}.
    \end{align*}
    Combining these results yields
    \begin{align}
    \label{pg2}
        \PR\Big(\|\sum_{i=1}^N\eta_{i}^{>u}\|_\infty \neq 0\Big) \le 2 (dN)^{-k}.
    \end{align}    
    A simple application of the Hölder inequality then also yields that by choosing $k=6$ we have 
    \begin{align}
    \label{pg3}
        \sqrt{N}\max_{1 \leq i \leq N, 1 \leq j \leq d}|\E[\eta_{ij}^{>u}]| \leq \frac{B_N}{(dN)^{6/2}}\leq CN^{-2}
    \end{align}
    Combining \eqref{pg2} and \eqref{pg3} we obtain that 
    \begin{align}
        \PR(N^{-1/2}\|\sum_{i=1}^N\eta_{i}^{>u}-\E[\eta_{i}^{>u}]\|_\infty>N^{-1})=o(N^{-1}).
    \end{align}
Equipped with this result (and still choosing $u$ as described in \eqref{e:uchoice} with $k=6$) we further deduce that 
     \begin{align}
    \label{pg4}
        &\PR\Big( \sup_{1 \leq j \leq d} N^{-1/2}\sum_{i=1}^N\eta_{ij} \leq t\Big)\\
        \geq& \PR\Big( \sup_{1 \leq j \leq d}N^{-1/2}\sum_{i=1}^N(\eta_{ij}^{\leq u}-\E[\eta_{ij}^{\leq u}]) \leq t-\sup_{1 \leq j \leq d} N^{-1/2}\sum_{i=1}^N (\eta_{ij}^{>u}-\E[\eta_{ij}^{>u}]) \Big)\\
        \geq &\PR\Big( \sup_{1 \leq j \leq d}N^{-1/2}\sum_{i=1}^N(\eta_{ij}^{\leq u}-\E[\eta_{ij}^{\leq u}]) \leq t-1/N \Big)+o(N^{-1})        
    \end{align}  
    where we used that $\eta_{ij}$ are centered random variables so that
    \begin{align}
        \eta_{ij}=(\eta_{ij}^{\leq u}-\E[\eta_{ij}^{\leq u}])+(\eta_{ij}^{> u}-\E[\eta_{ij}^{> u}])~.
    \end{align}
    Since the mixing property is stable under deterministic transformations (without increasing the mixing coefficients), the variables $(\eta_i^{\leq u})_{i=1,...,N}$ form a sequence of  bounded  $\beta$-mixing random variables and we may therefore apply \eqref{anticonc} from Theorem \ref{pt1} with the choices 
    \begin{align}
        D_N&=CB_NN^{\gamma/2} \\
        q&=N^{4\gamma+\delta}\\
        r&=N^\gamma
    \end{align} to obtain that
    \begin{align}
        \PR\Big( \sup_{1 \leq j \leq d} N^{-1/2}\sum_{i=1}^N\eta_{ij} \leq t\Big) \geq \PR\Big( \sup_{1 \leq j \leq d}N^{-1/2}\sum_{i=1}^N(\eta_{ij}^{\leq u}-\E[ \eta_{ij}^{\leq u}]) \leq t\Big)+o(1)        
    \end{align}
    Note that the conditions on $\overline{\sigma}_q$ and $\bar {\sigma}_r$ are an easy consequence of our assumptions and arguments similar to those leading to \eqref{pg2} and \eqref{pg3}.\\
    A similar argument for the reverse inequality and applying equation \eqref{boundedGAP} from Theorem \ref{pt1} yields a Gaussian approximation with the covariance determined by the covariance of the truncated variables $\eta^{\leq u}_{ij}$. Calculations similar to those following \eqref{e:truncated-covariance-bound} yield the desired conclusion.

\end{proof}

\section{Proof of Corollary \ref{cor:BrownianApprox:main}}
\label{appen:brownian-approx}
As before we provide a more general version for dependent data. 
\begin{cor}
\label{cor:BrownianApprox}
Denote by
\begin{align}
    \sigma_{ij}=\sum_{k \in \mathbb Z}\sigma_{ij}^{(k)}
\end{align}
the long run variances of $(\epsilon_{n,i})_{n \in \N, 1 \leq i \leq d(n)}$. Now let $B$ be a $d$ dimensional brownian motion with $\text{Cov}(B_i(1),B_j(1))=\sigma_{ij}$. Let $G^{(1)}$ and $G^{(2)}$ be the Gaussians from Theorem \ref{theo1:dep} and \ref{theo2:dep}, respectively. Further define the discretized modulus of continuity
\[\Psi_{\mathcal{S^\star}}^{\beta}(B)=\begin{cases}
    \max_{1 \leq i \leq d,(n,h)\in \mathcal{S}^\star}\frac{|B_i(\frac{n}{N})-B_i(\frac{n-h}{N})-(B_i(\frac{n+h}{N})-B_i(\frac{n}{N}))|}{(h/N)^\beta}, \quad \beta<1/2,\\
    \max_{1 \leq i \leq d,(n,h)\in \mathcal{S}^\star}\frac{|B_i(\frac{n}{N})-B_i(\frac{n-h}{N})-(B_i(\frac{n+h}{N})-B_i(\frac{n}{N}))|}{(h/N)^{1/2}\log(eN/h)^\beta}, \quad \beta>1/2.
\end{cases}    
\]
   \begin{enumerate}
       \item[i)] Let the assumptions of Theorem \ref{theo1:dep} hold.  Then
       \begin{align}
            \sup_{t \geq x_0}\Big| \PR(\Psi_{\mathcal{S^\star}}^{\beta}(B)\geq t)-\PR(\max_{1 \leq j \leq d, (n,h) \in \mathcal{S}^\star}|G^{(1)}_{j,n,h}|\geq t)\Big|\leq  o(1)
       \end{align}
       If $\beta>3/2$ we further have      
        \begin{align}
            \sup_{t \geq x_0}\Big| \PR(D^\beta\geq t)-\PR(\max_{1 \leq j \leq d, (n,h) \in \mathcal{S}^{\star}}|G^{(1)}_{j,n,h}|\geq t)\Big|\leq o(1)
        \end{align}
        \item[ii)] Let the assumptions of Theorem \ref{theo2:dep} hold.  Then
       \begin{align}
            \sup_{t \geq x_0}\Big| \PR(\Psi_{\mathcal{S}^\star}^{\beta}(B)\geq t)-\PR(\max_{1 \leq j \leq d, (n,h) \in \mathcal{S}^\star}|G^{(2)}_{j,n,h}|\geq t)\Big|\leq o(1)
       \end{align}
       We further have that
        \begin{align}
            \sup_{t \geq x_0}\Big| \PR(D^\beta\geq t)-\PR(\max_{1 \leq j \leq d, (n,h) \in \mathcal{S}^\star}|G^{(2)}_{j,n,h}|\geq t)\Big|\leq o(1)
        \end{align}
   \end{enumerate}
\end{cor}

\begin{proof}
In the following we will denote $G^{(i)}$ simply by $G$ as it will be clear from the context which is meant. \\
\textbf{Logarithmic weighting:}\\
We define $\mathcal{S}^\star_{1/2}=\{(n,h)\in \mathcal{S}| h> N^{1/2}\}$.   By Lemma \ref{Lemma:LongCovEst} and Lemma \ref{lem:compar_gaus} we have
    \begin{align}
    \label{pe:7}
        \sup_{t \geq x_0}\Big| \PR(\max_{1 \leq j \leq d, (n,h) \in \mathcal{S}^\star_{1/2}}|G_{j,n,h}|\geq t)-\PR(\Psi_{\mathcal{S}^\star_{1/2}}^{\beta}(B)\geq t)\Big|\leq CN^{-1/6}\log(N)\log^{2/3}(|\mathcal{S}^\star_{1/2}|)=o(1)
    \end{align}
    By the union bound (see the arguments following \eqref{e:Gaussex} and recall that $|S^\star|$ is of polynomial size if $d$ is) we further also obtain that there exists a $C_1>0$ such that
    \begin{align}
    \label{pe:8}
        \Pr( \max_{1 \leq j \leq d, (n,h) \in \mathcal{S}^\star\setminus\mathcal{S}^\star_{1/2}}|G_{j,n,h}|\geq x_0)&\leq C\sqrt{\log(N)}\exp\Big(-C_1\log(N)^{2\beta}\Big)=o(1)\\
        \Pr( \Psi_{\mathcal{S}^\star\setminus\mathcal{S}^\star_{1/2}}^{\beta}(B)\geq  x_0)&\leq C\sqrt{\log(N)}\exp\Big(-C_1\log(N)^{2\beta}\Big)=o(1)
    \end{align}
    Combining the last three inequalities yields the first part of the statement, i.e.
    \begin{align}
        \sup_{t \geq x_0}\Big| \PR(\max_{1 \leq j \leq d, (n,h) \in \mathcal{S}^\star}|G_{j,n,h}|\geq t)-\PR(\Psi_{\mathcal{S}^\star}^{\beta}(B)\geq t)\Big|=o(1)~.
    \end{align}
    Next we observe that
    \begin{align}
    \label{pe:9}
        &\sup_{t \geq x_0}\Big| \PR(\max_{1 \leq i \leq d}\|B_i\|_{\beta}\geq t)-\PR(\Psi_{\mathcal{S}^\star}^{\beta}(B)\geq t)\Big|\\
        = &\sup_{t \geq x_0}\Big| \PR\Big(\Psi_{\mathcal{S}^\star}^{\beta}(B)\geq t-(\max_{1 \leq i \leq d}\|B_i\|_{\beta}-\Psi_{\mathcal{S}^\star}^{\beta})\Big)-\PR\Big(\Psi_{\mathcal{S}^\star}^{\beta}(B)\geq t\Big)\Big|~.
    \end{align}
    By Lemma \ref{Lemma:Brownian:Discr:Univ} we have that 
    \begin{align}
    \label{pe:10}
        \PR\Big(\Big|\max_{1 \leq i \leq d}\|B_i\|_{\beta}-\Psi_{\mathcal{S}^\star}^{\beta}\Big|\leq \log(N)^{1/2-\beta+\delta}\Big)=1-o(1)~.
    \end{align}
    Using \eqref{pe:10}, we have, uniformly in $t\geq x_0$ and for any $0<\delta<\beta-1/2$ that
    \begin{align}
    \label{pe:14}
       &\PR\Big(\Psi_{\mathcal{S}^\star}^{\beta}(B)\geq t-(\max_{1 \leq i \leq d}\|B_i\|_{\beta}-\Psi_{\mathcal{S}^\star}^{\beta})\Big)\\
       \leq&\PR\Big(\Psi_{\mathcal{S}^\star}^{\beta}(B)\geq t-\log(N)^{1/2-\beta+\delta}\Big)+o(1) 
    \end{align}
    Now we use Lemma \ref{Lemma:AntiConc} to obtain
    \begin{align}
    \label{pe:13}
        &\PR\Big(\Psi_{\mathcal{S}^\star}^{\beta}(B)\geq t-\log(N)^{1/2-\beta+\delta}\Big)\\
        =&\PR\Big(\Psi_{\mathcal{S}^\star}^{\beta}(B)\geq t\Big)+O\Big(\log(N)\log(N)^{1/2-\beta+\delta}\Big)
    \end{align}
    Combining equations \eqref{pe:9}, \eqref{pe:14} and \eqref{pe:13} for $\delta$ sufficiently small we obtain
    \begin{align}
        \PR\Big(\Psi_{\mathcal{S}^\star}^{\beta}(B)\geq t-(\max_{1 \leq i \leq d}\|B_i\|_{\beta}-\Psi_{\mathcal{S}^\star}^{\beta})\Big)\leq \PR\Big(\Psi_{\mathcal{S}^\star}^{\beta}(B)\geq t\Big)+o(1)
    \end{align}
    A similar argument also yields the reverse inequality. Combining this with \eqref{pe:9} yields
    \begin{align}
        \sup_{t \geq x_0}\Big| \PR(\max_{1 \leq i \leq d}\|B_i\|_{\beta}\geq t)-\PR(\Psi_{\mathcal{S}^\star}^{\beta}(B)\geq t)\Big|= o(1)
    \end{align}
    as desired.\\

    \textbf{Polynomial weighting:}\\
    Follows by exactly the same arguments as the logarithmic case. We use the sets from the proof of Lemma \ref{Lemma:TruncateScalesPoly} instead of $\mathcal{S}^\star_{1/2}$. The analogue of equation \eqref{pe:10} uses the threshold $N^{-\gamma-\delta}$ for some sufficiently small $\delta$, instead. 
\end{proof}

\begin{lemma}
\label{Lemma:Brownian:Discr:Univ}
   Let $B$ be a univariate Brownian motion with variance $\sigma^2$. Then we have that there exists a random variable $M$ with finite $\Phi_2$ norm that depends only on $\sigma^2$ such that, for $\beta<1/2$, we have 
   \begin{align}
        &\Bigg|\sup_{(n,h)\in \mathcal{S}^\star}\frac{\Big|B(\frac{n}{N})-B(\frac{n-h}{N})-\Big(B(\frac{n+h}{N})-B(\frac{n}{N})\Big)\Big|}{(h/N)^\beta}- \sup_{\substack{0< h\leq 1/2\\ h \leq   t \leq 1-h}}\frac{\Big|B(t)-B(t-h)-\Big(B(t+h)-B(t)\Big)\Big|}{h^\beta}\Bigg|\\
        \leq &CM\sqrt{\log(N)}N^{-\frac{\beta\land (1/2-\beta)}{2\beta+1}}
    \end{align}
    and for $\beta>1/2$ we have
    \begin{align}
        &\Bigg|\sup_{(n,h)\in \mathcal{S}^\star}\frac{\Big|B(\frac{n}{N})-B(\frac{n-h+1}{N})-\Big(B(\frac{n+h}{N})-B(\frac{n}{N})\Big)\Big|}{\sqrt{h/N}\log(eN/h)^\beta}- \sup_{\substack{0< h\leq 1/2\\ h \leq   t \leq 1-h}}\frac{\Big|B(t)-B(t-h)-\Big(B(t+h)-B(t)\Big)\Big|}{\sqrt{h}\log(e/h)^\beta}\Bigg|\\
        \leq &CM\log(N)^{1/2-\beta}
    \end{align}
\end{lemma}
\begin{proof}

    \textbf{Polynomial weighting:}\\
    Let us first consider the difference
    \begin{align}
    \label{pe:12}
        \Bigg|\sup_{(n,h)\in \mathcal{S}^\star}\frac{\Big|B(\frac{n}{N})-B(\frac{n-h+1}{N})-\Big(B(\frac{n+h}{N})-B(\frac{n}{N})\Big)\Big|}{(h/N)^\beta}- \sup_{\substack{0< h\leq 1/2\\ h \leq   t \leq 1-h}}\frac{\Big|B(t)-B(t-h)-\Big(B(t+h)-B(t)\Big)\Big|}{h^\beta}\Bigg|.
    \end{align}
    Defining $\rho((t,s),(t',s'))=\max\{|t-t'|,|s-s'|\}$, we note that we can upper bound \eqref{pe:12} by 
    \begin{align}
    \label{pe:4}
        &2\sup_{\rho((s,t),(s',t'))\leq N^{-1}}\left|\frac{B(t)-B(s)}{(t-s)^\beta}-\frac{B(t')-B(s')}{(t'-s')^\beta}      \right|  .
    \end{align}
    We now consider, for some $c$ to be chosen later, the cases $|t-s|,|t'-s'|\leq N^{-c}$ and  $|t-s|,|t'-s'|> N^{-c}$ separately. 
    For the former we observe that
    \begin{align}
        2\sup_{\substack{\rho((s,t),(s',t'))\leq N^{-1}\\\rho((s,s'),(t,t')\leq N^{-c}}}\left|\frac{B(t)-B(s)}{(t-s)^\beta}-\frac{B(t')-B(s')}{(t'-s')^\beta} \right| 
        \leq  4\sup_{|t-s|\leq N^{-c}}\left|\frac{B(t)-B(s)}{(t-s)^\beta}\right|.
    \end{align}
    Now by Theorem 1 from \cite{chevallier:2023} we have for $|t-s|\leq N^{-c}$ that 
    \begin{align}
    \label{pe:5}
         \left|\frac{B(t)-B(s)}{(t-s)^\beta}\right|&\leq  \left|\frac{B(t)-B(s)}{\sqrt{(t-s)\log(\frac{1}{t-s})}}\right| \left|\frac{\sqrt{\log(\frac{1}{t-s})}}{(t-s)^{\beta-1/2}}\right|\\
         &\leq C M \sqrt{\log(N)}N^{-(1/2-\beta)c}
    \end{align}
    for some Subgaussian random variable $M$ whose $\Phi_2$ norm depends only on $\sigma^2$.\\

    For $|t-s|,|t'-s'|> N^{-c}$ we bring the fractions to a common denominator and add 0 to obtain
    \begin{align}
        &2\sup_{\substack{\rho((s,t),(s',t'))\leq N^{-1}\\\rho((s,s'),(t,t')> N^{-c}}}\left|\frac{B(t)-B(s)}{(t-s)^\beta}-\frac{B(t')-B(s')}{(t'-s')^\beta}\right|\\
        \leq & CN^{2c\beta}\sup_{\substack{\rho((s,t),(s',t'))\leq N^{-1}\\\rho((s,s'),(t,t')> N^{-c}}}\left|  (B(t)-B(s))((t-s)^\beta-(t'-s')^\beta)\right|\\
        &+ CN^{2c\beta}\sup_{\substack{\rho((s,t),(s',t'))\leq N^{-1}\\\rho((s,s'),(t,t')> N^{-c}}}\left|(t-s)^\beta(B(t)-B(s)-(B(t')-B(s')) \right|.
    \end{align}
    Using the mean value theorem and Theorem 1 from \cite{chevallier:2023} again we can upper bound this by
    \begin{align}
    \label{pe:6}
        CMN^{2c\beta}\Big(N^{(1-\beta)c-1}+\sqrt{\log(N)}N^{-1/2}\Big).
    \end{align}
    Choosing $c=\frac{1}{2(1+\beta)}$ we hence may combine the bounds \eqref{pe:5} and \eqref{pe:6} to obtain that
    \begin{align}
         &2\sup_{\rho((s,t),(s',t'))\leq N^{-1}}\left|\frac{B(t)-B(s)}{(t-s)^\beta}-\frac{B(t')-B(s')}{(t'-s')^\beta}    \right|\leq CM\sqrt{\log(N)}N^{-\frac{\beta\land (1/2-\beta))}{2\beta+1}}~,
    \end{align}
    as desired\\
    
    \textbf{Logarithmic weighting:}\\
    Arguing exactly as in the polynomial case (here we choose $c>0$ sufficiently small independent of the choice of $\beta$) we obtain
    \begin{align}
        &2\sup_{\rho((s,t),(s',t'))\leq N^{-1}}\left|\frac{B(t)-B(s)}{\sqrt{(t-s)}\log(\frac{e}{t-s})^\beta}-\frac{B(t')-B(s')}{\sqrt{(t'-s')}\log(\frac{e}{t'-s'})^\beta}    \right|\leq CM\log(N)^{1/2-\beta}~,
    \end{align}
    which finishes the proof.

\end{proof}

\begin{lemma}
\label{Lemma:LongCovEst}
For $(n,h)$ define
\[
    A(n,h) := \{n-h+1,\dots,n\},
    \qquad
    B(n,h) := \{n+1,\dots,n+h\},
\]
and
\[
    w_{n,h}(t)
    := \mathbf{1}_{\{t \in A(n,h)\}}
     - \mathbf{1}_{\{t \in B(n,h)\}} .
\]

Let
\[
    f_\beta(x) =
    \begin{cases}
         x^{\beta-1/2}, & \beta < 1/2, \\[1mm]
         \log(ex)^{-\beta}, & \beta > 1/2 .
    \end{cases}
\]
Then, if Assumption \ref{ass:segm} holds, there exists a constant $C>0$,
independent of $1 \le i_1,i_2 \le d$ and
$(n_1,h_1),(n_2,h_2) \in \mathcal S$, such that
\begin{align}
    &\Big|
    \operatorname{Cov}\big(
        \gamma_{i_1}(n_1,h_1),
        \gamma_{i_2}(n_2,h_2)
    \big)
    -
    \sigma_{i_1i_2}
    \kappa\big((n_1,h_1),(n_2,h_2)\big)
    f_\beta(N/h_1)f_\beta(N/h_2)
    \Big|
    \notag\\
    &\hspace{2cm}\le
    \frac{C}{\sqrt{h_1h_2}}
    f_\beta(N/h_1)f_\beta(N/h_2),
    \label{eq:LongCovEst}
\end{align}
where
\begin{align}
    \kappa\big((n_1,h_1),(n_2,h_2)\big)
    &:=
    \frac{1}{\sqrt{h_1h_2}}
    \sum_{t=1}^N
    w_{n_1,h_1}(t)w_{n_2,h_2}(t)
    \notag\\
    &=
    \frac{
        |A_1\cap A_2|
        +|B_1\cap B_2|
        -|A_1\cap B_2|
        -|B_1\cap A_2|
    }{\sqrt{h_1h_2}},
\end{align}
with
\[
    A_j=A(n_j,h_j),
    \qquad
    B_j=B(n_j,h_j),
    \qquad j=1,2.
\]

Moreover, define
\[
    d:=
    \operatorname{dist}
    \big(A_1\cup B_1,A_2\cup B_2\big).
\]
There exist constants $C,c>0$, independent of $i_1,i_2$ and the two
elements of $\mathcal S$, such that, whenever $d>0$,
\begin{align}
    \Big|
    \operatorname{Cov}\big(
        \gamma_{i_1}(n_1,h_1),
        \gamma_{i_2}(n_2,h_2)
    \big)
    \Big|
    \le
    \frac{C}{\sqrt{h_1h_2}}
    f_\beta(N/h_1)f_\beta(N/h_2)e^{-cd}.
    \label{eq:LongCovDistance}
\end{align}

In particular,
\[
    \Big|
    \operatorname{Cov}\big(
        \gamma_{i_1}(n_1,h_1),
        \gamma_{i_2}(n_2,h_2)
    \big)
    \Big|
    \le
    C f_\beta(N/h_1)f_\beta(N/h_2),
\]
and
\[
    \operatorname{Var}\big(\gamma_i(n,h)\big)
    =
    2\sigma_{ii}f_\beta(N/h)^2
    +
    O\left(
        h^{-1}f_\beta(N/h)^2
    \right).
\]
\end{lemma}

\begin{proof}
Write
\[
    \gamma_i(n,h)
    =
    \frac{f_\beta(N/h)}{\sqrt h}
    \sum_{t=1}^N w_{n,h}(t)\varepsilon_{t,i},
\]
where
\[
    w_{n,h}(t)
    =
    \mathbf 1_{\{t\in A(n,h)\}}
    -
    \mathbf 1_{\{t\in B(n,h)\}}.
\]
Hence
\begin{align}
    &\operatorname{Cov}\big(
        \gamma_{i_1}(n_1,h_1),
        \gamma_{i_2}(n_2,h_2)
    \big)
    \notag\\
    &\quad =
    \frac{f_\beta(N/h_1)f_\beta(N/h_2)}
         {\sqrt{h_1h_2}}
    \sum_{s,t=1}^N
    w_{n_1,h_1}(s)w_{n_2,h_2}(t)
    \operatorname{Cov}
    \big(\varepsilon_{s,i_1},\varepsilon_{t,i_2}\big).
    \label{eq:cov-double-sum}
\end{align}

By stationarity,
\[
    \operatorname{Cov}
    \big(\varepsilon_{s,i_1},\varepsilon_{t,i_2}\big)
    =
    \Gamma_{i_1i_2}(t-s),
    \qquad
    \Gamma_{i_1i_2}(k)
    :=
    \operatorname{Cov}
    \big(\varepsilon_{0,i_1},\varepsilon_{k,i_2}\big).
\]
Consequently,
\begin{align}
    &\operatorname{Cov}\big(
        \gamma_{i_1}(n_1,h_1),
        \gamma_{i_2}(n_2,h_2)
    \big)
    \notag\\
    &\quad =
    \frac{f_\beta(N/h_1)f_\beta(N/h_2)}
         {\sqrt{h_1h_2}}
    \sum_{k\in\mathbb Z}
    \Gamma_{i_1i_2}(k)
    \sum_{s=1}^N
    w_{n_1,h_1}(s)w_{n_2,h_2}(s+k),
    \label{eq:cov-lag-sum}
\end{align}
where we set $w_{n,h}(u)=0$ for
$u\notin\{1,\dots,N\}$.

Let
\[
    S(k)
    :=
    \sum_{s=1}^N
    w_{n_1,h_1}(s)w_{n_2,h_2}(s+k).
\]
Since
\[
    \sigma_{i_1i_2}
    =
    \sum_{k\in\mathbb Z}\Gamma_{i_1i_2}(k),
\]
we have
\begin{align}
    \sum_{k\in\mathbb Z}
    \Gamma_{i_1i_2}(k)S(k)
    &=
    \sigma_{i_1i_2}S(0)
    +
    \sum_{k\in\mathbb Z}
    \Gamma_{i_1i_2}(k)\big(S(k)-S(0)\big).
    \label{eq:cov-decomposition}
\end{align}

Because each $w_{n,h}$ takes values in $\{-1,0,1\}$ and has only a
fixed number of jump points, shifting the second block pattern by $k$
changes the overlap only within distance $|k|$ of its boundaries.
Therefore,
\[
    |S(k)-S(0)|
    \le C|k|
\]
uniformly in $n_1,n_2,h_1,h_2$. More explicitly,
\begin{align}
    |S(k)-S(0)|
    &\le
    \sum_{s=1}^N
    \left|
        w_{n_2,h_2}(s+k)-w_{n_2,h_2}(s)
    \right|
    \le C|k|.
\end{align}

By geometric $\beta$-mixing and the moment condition in Assumption
\ref{ass:segm}, for some $c>0$,
\[
    |\Gamma_{i_1i_2}(k)|
    \le
    C\beta(|k|)^{\eta/(2+\eta)}
    \le
    Ce^{-c|k|},
\]
uniformly in $i_1,i_2$. In particular,
\[
    \sum_{k\in\mathbb Z}
    |\Gamma_{i_1i_2}(k)|
    <\infty,
    \qquad
    \sum_{k\in\mathbb Z}
    |k|\,|\Gamma_{i_1i_2}(k)|
    <\infty
\]
uniformly in $i_1,i_2$. It follows that
\[
    \left|
    \sum_{k\in\mathbb Z}
    \Gamma_{i_1i_2}(k)S(k)
    -
    \sigma_{i_1i_2}S(0)
    \right|
    \le C.
\]
Together with \eqref{eq:cov-lag-sum}, this yields
\begin{align}
    &\left|
    \operatorname{Cov}\big(
        \gamma_{i_1}(n_1,h_1),
        \gamma_{i_2}(n_2,h_2)
    \big)
    -
    \frac{f_\beta(N/h_1)f_\beta(N/h_2)}
         {\sqrt{h_1h_2}}
    \sigma_{i_1i_2}S(0)
    \right|
    \notag\\
    &\hspace{2cm}\le
    \frac{C}{\sqrt{h_1h_2}}
    f_\beta(N/h_1)f_\beta(N/h_2).
\end{align}

Finally,
\begin{align}
    S(0)
    &=
    \sum_{t=1}^N
    w_{n_1,h_1}(t)w_{n_2,h_2}(t)
    \notag\\
    &=
    |A_1\cap A_2|
    +|B_1\cap B_2|
    -|A_1\cap B_2|
    -|B_1\cap A_2|.
\end{align}
Thus,
\[
    \kappa\big((n_1,h_1),(n_2,h_2)\big)
    =
    \frac{S(0)}{\sqrt{h_1h_2}},
\]
which proves \eqref{eq:LongCovEst}.

It remains to prove the distance-dependent estimate. Let
\[
    I_j:=A_j\cup B_j,
    \qquad j=1,2.
\]
Suppose that
\[
    d=\operatorname{dist}(I_1,I_2)>0.
\]
Then $I_1$ and $I_2$ are disjoint. Using
\eqref{eq:cov-double-sum}, the bound
$|w_{n_j,h_j}(t)|\le 1$, and the exponential decay of
$\Gamma_{i_1i_2}$, we obtain
\begin{align}
    &\left|
    \operatorname{Cov}\big(
        \gamma_{i_1}(n_1,h_1),
        \gamma_{i_2}(n_2,h_2)
    \big)
    \right|
    \notag\\
    &\quad\le
    \frac{C f_\beta(N/h_1)f_\beta(N/h_2)}
         {\sqrt{h_1h_2}}
    \sum_{s\in I_1}\sum_{t\in I_2}e^{-c|t-s|}.
    \label{eq:distance-double-sum}
\end{align}

Without loss of generality, suppose that $I_1$ lies to the left of
$I_2$, and write
\[
    r_1:=\max I_1,
    \qquad
    \ell_2:=\min I_2.
\]
Then
\[
    d=\ell_2-r_1
\]
and, for $s\in I_1$ and $t\in I_2$,
\[
    t-s
    =
    d+(r_1-s)+(t-\ell_2).
\]
Therefore,
\begin{align}
    \sum_{s\in I_1}\sum_{t\in I_2}e^{-c|t-s|}
    &=
    \sum_{s\in I_1}\sum_{t\in I_2}e^{-c(t-s)}
    \notag\\
    &\le
    e^{-cd}
    \left(
        \sum_{u=0}^{\infty}e^{-cu}
    \right)
    \left(
        \sum_{v=0}^{\infty}e^{-cv}
    \right)
    \notag\\
    &\le
    Ce^{-cd}.
\end{align}
Substitution into \eqref{eq:distance-double-sum} proves
\eqref{eq:LongCovDistance}.

To obtain the uniform covariance bound, note that
\[
    |S(0)|
    \le
    \sum_{t=1}^N
    |w_{n_1,h_1}(t)w_{n_2,h_2}(t)|
    \le
    2\min\{h_1,h_2\}.
\]
Hence
\[
    \left|
    \kappa\big((n_1,h_1),(n_2,h_2)\big)
    \right|
    \le
    \frac{2\min\{h_1,h_2\}}{\sqrt{h_1h_2}}
    \le 2.
\]
The claimed uniform bound now follows from \eqref{eq:LongCovEst} and
the uniform boundedness of $\sigma_{i_1i_2}$.

Finally, taking
\[
    (n_1,h_1)=(n_2,h_2)=(n,h),
    \qquad
    i_1=i_2=i,
\]
gives
\[
    S(0)
    =
    \sum_{t=1}^N w_{n,h}(t)^2
    =
    |A(n,h)|+|B(n,h)|
    =
    2h,
\]
and therefore
\[
    \kappa\big((n,h),(n,h)\big)=2.
\]
Applying \eqref{eq:LongCovEst} proves
\[
    \operatorname{Var}\big(\gamma_i(n,h)\big)
    =
    2\sigma_{ii}f_\beta(N/h)^2
    +
    O\left(
        h^{-1}f_\beta(N/h)^2
    \right).
\]
\end{proof}

\begin{lemma}
\label{Lemma:AntiConc}
    Let $(G_1,...,G_d)$ be a d-dimensional Gaussian vector with zero means and bounded variances. We then have for any $x_0>0$
    \begin{align}
        \sup_{t \geq x_0}\Big(\PR(\max_{1 \leq i \leq d}G_i\geq t)-\PR(\max_{1 \leq i \leq d}G_i\geq t+\epsilon)\Big)\leq C\Big(\epsilon  \log(d)+d^{-1}\Big),
    \end{align}
    where $C$ depends only on $x_0$.
\end{lemma}
\begin{proof}
    Define $I_\delta=\Big\{ 1 \leq i \leq d | \text{Var}(G_i)\leq \frac{1}{\delta \log(d)}\Big\} $ and note that by the union bound we have
    \begin{align}
        \PR(\max_{i \in I_\delta}G_i\geq t)\leq Cd\exp(-t^2\delta\log(d))=Cd^{1-t^2\delta}.
    \end{align}
    Choosing $\delta=2/x_0^2$ we hence obtain for any $x_0>0$ that
    \begin{align}
        \sup_{t \geq x_0}\Big(\PR(\max_{1 \leq i \le d}G_i\geq t)-\PR(\max_{i \in I_\delta^c }G_i\geq t)\Big)\leq Cd^{-1}~.
    \end{align}
    We may then apply Theorem C.3 from \cite{chetverikov:wilhelm:kim:2021} to $G$ restricted to $I_\delta^c$ to obtain the desired conclusion.
\end{proof}

\section{Background results}\label{appen:aux_res}
In this section, we state auxiliary results that facilitate the proof of our main results. We start with standard concentration bounds for subgaussian random variables.

\begin{lemma}\label{lem:lem_sub_gaussian}
Let $X_1,\ldots,X_N$ be independent random variables such that $\|X_i\|_{\psi_2}\le K$ for all $1\leq i \leq N$. Then for any $t>0$,  scalars $a_1,\ldots,a_N$, and some constants $c_1,c_2$ that do not depend on random variables, the following statements hold:
\begin{align}\label{eq:sum_subgauss}
    \|\sum_{i=1}^{N}a_iX_i\|_{\psi_2} & \le \sqrt{\sum_{i=1}^Na_i^2}K ,\\ 
\label{eq:max_subgauss}
        \PR(\max_{1\le i\le N}|X_i|>t)& \le 2Ne^{-\frac{t^2}{2c_1K^2}},
    \\ \label{eq:tail_subgauss}
        \PR(|X_i|>t)& \le 2e^{-\frac{t^2}{2c_2K^2}}.
    \end{align}

\end{lemma}
\begin{proof}
    Proofs can be found e.g. in \cite{Vershynin:2018}.
\end{proof}

 Since in this article we also deal with dependent subgaussian vectors, for typical bounds above cannot be applied, we introduce the concentration inequality for high dimensional $\beta$-mixing mean vectors. We point out that the result only requires polynomial, but not geometric, decay of the mixing coefficients (which is often imposed in the literature).

\begin{lemma}
\label{lem:SumSplit}
Let $(X_n)_{1 \leq n \leq N}$ be a collection of centered $\R^d $-valued random variables and suppose that for some constant $D_N$ it holds that
\begin{align*}
    \|X_{ij}\|_{\Psi_2}\leq D_N, \quad i=1,...,N, \,\,j=1,...,d.
\end{align*}
Then with $S_N:= \sum_{i=1}^NX_i$ we obtain, for arbitrary $m\leq N/2$, the following concentration result 
\begin{align*}
     \PR(\|S_N\|_\infty\geq t) \leq 2md\exp\Big(-\frac{t^2}{2D_N^2mN}\Big)+4N\beta(m).
\end{align*}
Here the mixing coefficient $\beta(m)$ is defined as in \eqref{e:defbet}.
\end{lemma}
\begin{proof}

As a first step, we reduce the problem of bounding the tail of $S_N$, a sum of mixing random vectors, to bounding the tail of a sum of independent random vectors. To that end, we use Theorem 1 from \cite{yoshihara:1978b}.  This theorem states that for any $N \in \mathbb{N}$, $m \in \{1,2,...,N\}$ and $t \in (0,\infty)$ the inequality
    \begin{align}    
        \PR(\|S_N\|_\infty \geq t) \leq \sum_{i=1}^m\PR(\|Z_{i}+...+Z_{i+k_im}\|_{\infty}\geq t/m)+4N\beta(m).
    \end{align} 
    holds.  Here  $k_i$ is the largest integer such that $i+k_im \leq N$ (note that $\frac{N}{2m}\leq k_i$) and $(Z_i)_{i \in \N}$ is a sequence of independent random vectors such that $Z_i\overset{d}{=}X_i$. \\ 
We now further analyze the first term on the right of the above equation. We have that 
\[
\|Z_{i}+...+Z_{i+k_im}\|_\infty\geq t/m \Leftrightarrow 
 \exists j \in \{1,...,d\}: \,|Z_{i,j}+...+Z_{i+k_im,j}|\geq t/m\,. 
\]
We can thus use the union bound (along dimensions; giving us a factor $d$) and standard concentration properties  of subgaussian random variables 
\eqref{eq:max_subgauss} to obtain that
    \begin{align*}
        \PR(\|S_N\|_\infty\geq t) \leq 2md\exp\Big(-\frac{t^2}{2D_N^2mN}\Big)+4N\beta(m)~.
    \end{align*}    
\end{proof}

Next Lemma is an anticoncentration result for maxima of Gaussian random vectors, which is an immediate consequence of Theorem 1 in \cite{chernozhukov:chetverikov:kato:2017}. 
\begin{lemma} \label{lem:anti_conc}
     Let $Y = (Y_1,...,Y_d)^\top$ be a centered Gaussian random vector in $\mathbb{R}^d$ such that $\E(Y^2_i ) \ge b^2$ for all $j = 1,\ldots,d$ and some constant $b > 0$. 
    Then, it holds for every $x\in\R$ and $a>0$
    \begin{align}
        \PR(\max_{1 \leq i \leq d}Y_i\leq x+a)-\PR(\max_{1 \leq i \leq d}Y_i\leq x)\leq \frac{a}{b} (\sqrt{2\log d}+2).
    \end{align}
\end{lemma}

Next, we present an error bound for multivariate Gaussian approximation from \cite{chetverikov:wilhelm:kim:2021}  which requires some preparation. 
 For a sequence $\{A_i\}_{i=1}^M$ of $d$-dimensional random variables with zero mean, we define 
    \begin{align}
        S_M:=\max_{1\le j\le d}\frac{1}{\sqrt{M}}\sum_{i=1}^MA_{ij}, \quad M\in \N.
    \end{align}
The error bound is formulated under the following assumptions for some sequence of positive constants $\{a_i\}_{i \in \mathbb N}$. 
\begin{enumerate}[label=(C-\arabic*)] 
     \item\label{cond1_var} $\frac{1}{M}\sum_{i=1}^{M}\mathbb{E}(A_{ij}^2)\ge \sigma^2 $ for some  $j$; 
    \item\label{cond2_momen}  $\frac{1}{M}\sum_{i=1}^{M}\mathbb{E}(|A_{ij}|^{2+r})\le a_M^r $ for  $r=1,2$ and all $1\le j\le d$; 
    \item\label{cond3_orlizs} $\mathbb{E}[\exp(|A_{ij}|/a_M)]\le 2$ for all  $1\le j\le d$.
\end{enumerate}
Also, let $\{\tilde A_i\}_{i=1}^M$ be independent Gaussian vectors with mean zero and the same covariance profile as $\{A_i\}_{i=1}^M$. We define the analog of $S_M$ as
\begin{align} \label{eq_def_tilde_S_M}
           \tilde S_M:=\max_{1\le j\le d}\frac{1}{\sqrt{M}}\sum_{i=1}^M\tilde A_{ij}.
\end{align}
We are now in the position to state the  auxiliary result quantifying a Kolmogorov-type distance between $S_M$ and $\tilde S_M$ in relation to $M$, $a_M$ and $d$. 

\begin{lemma} [Theorem 4.1 from \cite{chetverikov:wilhelm:kim:2021}]\label{lem:th_chetverikov}
Suppose that $\{A_i\}_{i=1}^M$ is a sequence of independent centered random vectors satisfying assumptions \ref{cond1_var}-\ref{cond3_orlizs}. Let $\tilde S_M$ be given as in \eqref{eq_def_tilde_S_M}.
Then for any $x_0>0$, we have
\begin{align}
    \sup_{x\ge x_0}|\PR(S_M\le x)-\PR(\tilde S_M\le x)|\le \lb\frac{c a_M^2\log^{10}(dM)}{M}\rb^{1/6} .
\end{align}

\end{lemma}

Next, we state a bound for the Kolmogorov-distance between the maxima of Gaussian random variables with similar variance profiles. 

\begin{lemma}[Theorem 2 in \cite{chernozhukov:chetverikov:kato:2013:comparison}]\label{lem:compar_gaus}
Let $d\ge 2$ and $ X = (X_1,\ldots,X_d)^\intercal$, $Y = (Y_1,\ldots,Y_d)^\intercal$ be two centered Gaussian random vectors in $\mathbb{R}^d$ with covariance matrices $\Sigma^X = (\sigma^{X}_{ jk})$ for $1\le j,k\le d$ and $\Sigma^Y = (\sigma^{Y}_{ jk})$ for $1\le j,k\le d$, respectively. Define $\Delta= \max_{ 1\le j,k\le d} |\sigma^X_{jk}- \sigma^Y_{jk}|$ and
      $\sigma^2=\min_{1\le j\le d}\sigma^{Y}_{jj} > 0$. Then 
    \begin{align} \label{eq_ineq}
        \sup_{x\in\R} |\PR( \max_{1\le j\le d} X_j \le x)-\PR(\max_{1\le j\le d} Y_j \le x)| \le c\sigma^{-1}\Delta^{1/3} (1\vee  \log(d/\Delta))^{1/3} \log^{1/3} d,
    \end{align}
    where $c > 0$ depends only on  $\max_{1\le j\le d}\sigma^Y_{ jj}$. The right hand side of \eqref{eq_ineq} is understood to be 0 when $\Delta= 0$.
\end{lemma}
\begin{rem}
    Theorem 2 in \cite{chernozhukov:chetverikov:kato:2013:comparison} is formulated using a generic constant $C$ whose value depends on
      $\min_{1\le j\le d}\sigma^{Y}_{jj} $, without specifying the precise form of this dependence. However, in our application of Lemma \ref{lem:compar_gaus}, $\min_{1\le j\le d}\sigma^{Y}_{jj} $ will tend to 0 as the dimension $d$ grows. 
     To deal with this issue, we find an upper bound in Lemma \ref{lem:compar_gaus} that explictly depends on  $\min_{1\le j\le d}\sigma^{Y}_{jj} $. This bound can be obtained by following the proof of  Theorem 2 in \cite{chernozhukov:chetverikov:kato:2013:comparison} but replacing the anti-concentration result used there by the one presented in Lemma~\ref{lem:anti_conc}. 
\end{rem}

We conclude this section with a bound for the partial-sum process of subgaussian random variables. 
\begin{lemma}
    \label{lem:psum_conc}
    Let $(\epsilon_i)$ denote a sequence of independent mean zero random variables that satisfy
    \[ 
        \|\epsilon_i\|_{\Psi_2}\leq c_1.
    \]
    Then there exists a constant $c_2$ that depends only on $c_1$ and a random variable $M$ with $\|M\|_{\Psi_2}\leq c_2$ such that
    \[
        \Big|\sum_{n=l+1}^k\epsilon_n\Big|\leq M\sqrt{(k-l)\log(ek)}, \quad l\leq k, k\geq 1~.
    \]
\end{lemma}
\begin{proof}[Proof of Lemma \ref{lem:psum_conc}]
    Let $P_N(k/N)=N^{-1/2}\sum_{n=1}^{k}\epsilon_n$ denote the partial sum process. We apply Lemma 2.1 from \cite{bastian:kutta:2025} to obtain that
    \[
        M:=\sup_{l,k\geq 1}\frac{\Big|P_N(k/N)-P_N(l/N)\Big|}{\sqrt{(k-l)/N(1+\log(k/(k-l))+|\log(k/N)|}}
    \]
    satisfies $\|M\|_{\Psi_2}<\infty$ with the precise value depending only on $c_1$. Appropriate multiplication then yields that
    \[
         \Big|\sum_{n=l+1}^k\epsilon_n\Big|\leq M \sqrt{N}\sqrt{(k-l)/N(1+\log(k/(k-l))+|\log(k/N)|}\leq M\sqrt{3(k-l)\log(e(k\lor N))}
    \]
    from which the desired conclusion follows.
\end{proof}

\end{document}